\documentclass[11pt,letterpaper]{article}

\usepackage{fullpage}
\usepackage{xspace}
\def\showauthornotes{0}

\def\stocmode{0}

\def\arxivmode{0}
\def\fastmode{0}

\def\showkeys{0}
\def\showdraftbox{1}
\def\showcolorlinks{1}
\def\usemicrotype{1}
\def\showfixme{1}

\newcommand{\llangle}{\left\langle}
\newcommand{\rrangle}{\right\rangle}
\newcommand{\diag}{\mathrm{diag}}
\newcommand\sr[2]{\ \stackrel{#1}{#2}\ }

\newcommand{\GLMIterate}{\textsc{GLMIterate}}
\newcommand{\TGLM}{\mathcal{T}_{\mathrm{GLM}}}
\newcommand{\mm}{\ell}

\ifnum\arxivmode=0  \else  \fi

\ifnum\fastmode=0
\usepackage{etex}
\fi

\ifnum\fastmode=0
\usepackage[l2tabu, orthodox]{nag}
\fi

\usepackage{xspace,enumerate}

\usepackage[dvipsnames]{xcolor}

\ifnum\fastmode=0
\usepackage[T1]{fontenc}
\usepackage[full]{textcomp}
\fi

\usepackage[american]{babel}

\usepackage{mathtools}

\usepackage{amsthm}

\newtheorem{theorem}{Theorem}[section]
\newtheorem*{theorem*}{Theorem}

\newtheorem*{proposition*}{Proposition}
\newtheorem{lemma}[theorem]{Lemma}
\newtheorem*{lemma*}{Lemma}
\newtheorem{corollary}[theorem]{Corollary}
\newtheorem*{conjecture*}{Conjecture}
\newtheorem{fact}[theorem]{Fact}
\newtheorem*{fact*}{Fact}

\newtheorem*{exercise*}{Exercise}

\newtheorem*{hypothesis*}{Hypothesis}

\theoremstyle{definition}
\newtheorem{definition}[theorem]{Definition}

\newtheorem{assumption}[theorem]{Assumption}
\newtheorem{exercise-easy}[theorem]{Exercise}
\newtheorem{exercise-med}[theorem]{Exercise}
\newtheorem{exercise-hard}[theorem]{Exercise$^\star$}

\newtheorem*{claim*}{Claim}

\newtheorem*{remark*}{Remark}

\newtheorem*{observation*}{Observation}

\ifnum\arxivmode=0
\usepackage{newpxtext} %
\usepackage{textcomp} %
\usepackage[varg,bigdelims]{newpxmath}
\usepackage[scr=boondoxo,scrscaled=1.10]{mathalfa}

\usepackage{bm} %
\let\mathbb\varmathbb
\fi

\ifnum\arxivmode=1
\usepackage[varg]{pxfonts} %
\usepackage{textcomp} %
\usepackage[scr=rsfso]{mathalfa}%
\usepackage{bm} %
\fi

\ifnum\showkeys=1
\usepackage[color]{showkeys}
\fi

\makeatletter
\@ifpackageloaded{hyperref}
{}
{\usepackage[pagebackref]{hyperref}}

\definecolor{bleudefrance}{rgb}{0.01, 0.1, 1.0}
\definecolor{azure}{rgb}{0.0, 0.5, 1.0}

\ifnum\showcolorlinks=1
\hypersetup{
colorlinks=true,
urlcolor=blue,
linkcolor=blue,
citecolor=OliveGreen}
\fi

\ifnum\showcolorlinks=0
\hypersetup{
pagebackref=true,
colorlinks=false,
pdfborder={0 0 0}}
\fi

\usepackage{prettyref}

\newcommand{\savehyperref}[2]{\texorpdfstring{\hyperref[#1]{#2}}{#2}}

\newrefformat{eq}{\savehyperref{#1}{\textup{(\ref*{#1})}}}
\newrefformat{lem}{\savehyperref{#1}{Lemma~\ref*{#1}}}
\newrefformat{def}{\savehyperref{#1}{Definition~\ref*{#1}}}
\newrefformat{obs}{\savehyperref{#1}{Observation~\ref*{#1}}}
\newrefformat{ass}{\savehyperref{#1}{Assumption~\ref*{#1}}}
\newrefformat{thm}{\savehyperref{#1}{Theorem~\ref*{#1}}}
\newrefformat{cor}{\savehyperref{#1}{Corollary~\ref*{#1}}}
\newrefformat{cha}{\savehyperref{#1}{Chapter~\ref*{#1}}}
\newrefformat{sec}{\savehyperref{#1}{Section~\ref*{#1}}}
\newrefformat{app}{\savehyperref{#1}{Appendix~\ref*{#1}}}
\newrefformat{tab}{\savehyperref{#1}{Table~\ref*{#1}}}
\newrefformat{fig}{\savehyperref{#1}{Figure~\ref*{#1}}}
\newrefformat{hyp}{\savehyperref{#1}{Hypothesis~\ref*{#1}}}
\newrefformat{alg}{\savehyperref{#1}{Algorithm~\ref*{#1}}}
\newrefformat{rem}{\savehyperref{#1}{Remark~\ref*{#1}}}
\newrefformat{item}{\savehyperref{#1}{Item~\ref*{#1}}}
\newrefformat{step}{\savehyperref{#1}{step~\ref*{#1}}}
\newrefformat{conj}{\savehyperref{#1}{Conjecture~\ref*{#1}}}
\newrefformat{fact}{\savehyperref{#1}{Fact~\ref*{#1}}}
\newrefformat{prop}{\savehyperref{#1}{Proposition~\ref*{#1}}}
\newrefformat{prob}{\savehyperref{#1}{Problem~\ref*{#1}}}
\newrefformat{claim}{\savehyperref{#1}{Claim~\ref*{#1}}}
\newrefformat{relax}{\savehyperref{#1}{Relaxation~\ref*{#1}}}
\newrefformat{red}{\savehyperref{#1}{Reduction~\ref*{#1}}}
\newrefformat{part}{\savehyperref{#1}{Part~\ref*{#1}}}
\newrefformat{ex}{\savehyperref{#1}{Exercise~\ref*{#1}}}
\newrefformat{property}{\savehyperref{#1}{Property~\ref*{#1}}}

\newcommand{\Sref}[1]{\hyperref[#1]{\S\ref*{#1}}}

\usepackage{nicefrac}

\ifnum\fastmode=0
\ifnum\usemicrotype=1
\usepackage{microtype}
\fi
\fi

\ifnum\showauthornotes=2
\newcommand{\mynotes}[1]{{\sffamily\small\color{teal}{#1}}\medskip}
\newcommand{\Authornote}[2]{{\sffamily\small\color{Maroon}{[#1: #2]}}\medskip}
\newcommand{\Authornotecolored}[3]{{\sffamily\small\color{#1}{[#2: #3]}}}
\newcommand{\Authorcomment}[2]{{\sffamily\small\color{gray}{[#1: #2]}}}
\newcommand{\Authorstartcomment}[1]{\sffamily\small\color{gray}[#1: }

\newcommand{\Authorfnote}[2]{\footnote{\color{red}{#1: #2}}}
\newcommand{\Authorfixme}[1]{\Authornote{#1}{\textbf{??}}}
\newcommand{\Authormarginmark}[1]{\marginpar{\textcolor{red}{\fbox{\Large #1:!}}}}
\newcommand{\myexplain}[1]{{\sffamily\small\color{red}{\noindent [Explanation:\medskip\newline \begin{quote}#1\hfill]\end{quote}}}\medskip}
\newcommand{\explain}[1]{{\sffamily\small\color{red}{#1}}\medskip}

\else

\newcommand{\mynotes}[1]{}
\newcommand{\Authornote}[2]{}
\newcommand{\Authornotecolored}[3]{}
\newcommand{\Authorcomment}[2]{}
\newcommand{\Authorstartcomment}[1]{}

\newcommand{\Authorfnote}[2]{}
\newcommand{\Authorfixme}[1]{}
\newcommand{\Authormarginmark}[1]{}
\newcommand{\myexplain}[1]{}
\newcommand{\explain}[1]{}

\ifnum\showauthornotes=1
\renewcommand{\myexplain}[1]{{\sffamily\small\color{red}{\noindent \begin{quote}{\bf Explanation:} \medskip\newline #1\end{quote}}}\medskip}
\fi

\fi

\ifnum\showfixme=0

\fi

\usepackage{boxedminipage}

\newcommand{\wt}[1]{\widetilde{#1}}
\newcommand{\Esymb}{\mathbb{E}}
\newcommand{\Psymb}{\mathbb{P}}

\DeclareMathOperator*{\E}{\Esymb}

\DeclareMathOperator*{\ProbOp}{\Psymb}

\renewcommand{\Pr}{\ProbOp}

\newcommand{\textparen}[1]{\text{(#1)}}

\ifx\because\undefined
\newcommand{\because}[1]{\textparen{because #1}}
\else
\renewcommand{\because}[1]{\textparen{because #1}}
\fi

\newcommand{\defeq}{\stackrel{\mathrm{def}}=}

\newcommand{\seteq}{\mathrel{\mathop:}=}

\newcommand\bdot\bullet

\ifx\mathds\undefined %

\else

\fi

\DeclareMathOperator{\supp}{supp}

\DeclareMathOperator{\rank}{rank}

\newcommand{\Z}{\mathbb Z}

\newcommand{\R}{\mathbb R}

\newcommand{\cA}{\mathcal A}

\newcommand{\cF}{\mathcal F}

\newcommand{\cI}{\mathcal I}
\newcommand{\cJ}{\mathcal J}

\newcommand{\cN}{\mathcal N}

\newcommand{\baltW}{W}

\newcommand{\sfB}{\mathsf B}

\newcommand{\sfK}{\mathsf K}

\newcommand{\sfN}{\mathsf N}

\renewcommand{\leq}{\leqslant}
\renewcommand{\le}{\leqslant}
\renewcommand{\geq}{\geqslant}
\renewcommand{\ge}{\geqslant}

\ifnum\showdraftbox=1

\else

\fi

\let\epsilon=\varepsilon

\numberwithin{equation}{section}

\newcommand\MYcurrentlabel{xxx}

\newcommand{\MYstore}[2]{%
  \global\expandafter \def \csname MYMEMORY #1 \endcsname{#2}%
}

\newcommand{\MYload}[1]{%
  \csname MYMEMORY #1 \endcsname%
}

\newcommand{\MYnewlabel}[1]{%
  \renewcommand\MYcurrentlabel{#1}%
  \MYoldlabel{#1}%
}

\newcommand{\MYdummylabel}[1]{}

\newcommand{\torestate}[1]{%
  \let\MYoldlabel\label%
  \let\label\MYnewlabel%
  #1%
  \MYstore{\MYcurrentlabel}{#1}%
  \let\label\MYoldlabel%
}

\newcommand{\restatetheorem}[1]{%
  \let\MYoldlabel\label
  \let\label\MYdummylabel
  \begin{theorem*}[Restatement of \prettyref{#1}]
    \MYload{#1}
  \end{theorem*}
  \let\label\MYoldlabel
}

\newcommand{\restatelemma}[1]{%
  \let\MYoldlabel\label
  \let\label\MYdummylabel
  \begin{lemma*}[Restatement of \prettyref{#1}]
    \MYload{#1}
  \end{lemma*}
  \let\label\MYoldlabel
}

\newcommand{\restateprop}[1]{%
  \let\MYoldlabel\label
  \let\label\MYdummylabel
  \begin{proposition*}[Restatement of \prettyref{#1}]
    \MYload{#1}
  \end{proposition*}
  \let\label\MYoldlabel
}

\newcommand{\restatefact}[1]{%
  \let\MYoldlabel\label
  \let\label\MYdummylabel
  \begin{fact*}[Restatement of \prettyref{#1}]
    \MYload{#1}
  \end{fact*}
  \let\label\MYoldlabel
}

\newcommand{\restate}[1]{%
  \let\MYoldlabel\label
  \let\label\MYdummylabel
  \MYload{#1}
  \let\label\MYoldlabel
}

\newcommand{\addreferencesection}{
  \phantomsection
\ifnum\stocmode=0
  \addcontentsline{toc}{section}{References}
\else
  \addcontentsline{toc}{section}{References \hspace*{1in} --------- End of extended abstract ---------}
\fi

}

\newcommand{\e}{\epsilon}
\newcommand{\eps}{\epsilon}

\renewcommand{\paragraph}[1]{\medskip\noindent{\bf #1.}}

\allowdisplaybreaks

\usepackage{paralist}

\usepackage{comment}
\usepackage[algo2e,linesnumbered,lined,ruled,boxed]{algorithm2e}
\usepackage{algpseudocode}

\definecolor{orange}{rgb}{0.8, 0.33, 0.0}

\let\pref=\prettyref

\newcommand{\cov}{\mathcal{N}}
\newcommand{\diam}{\mathrm{diam}}

\DeclareMathOperator{\tr}{tr}

\newcommand{\vertiii}[1]{{\left\vert\kern-0.25ex\left\vert\kern-0.25ex\left\vert #1 
          \right\vert\kern-0.25ex\right\vert\kern-0.25ex\right\vert}}

\newcommand\f{\varphi}

\DeclareMathOperator{\argmin}{\mathrm{argmin}}

\newcommand{\tEval}{\mathcal{T}_{\mathrm{eval}}}

\usepackage{accents}

\usepackage{stmaryrd}
\mathtoolsset{showonlyrefs=true}

\usepackage{enumitem}
\usepackage{titletoc}

\renewcommand{\mathbb}{\vvmathbb}

\renewcommand{\supp}{\mathrm{supp}}

\renewcommand{\O}{\tilde{O}}
\newcommand{\ls}{\lesssim}
\renewcommand{\l}{\llangle}
\renewcommand{\r}{\rrangle}
\newcommand{\FindWeights}{\textsc{FindWeights}}
\newcommand{\LevApprox}{\textsc{LevApprox}}
\newcommand{\Iterate}{\textsc{Iterate}}

\newcommand{\assign}{\leftarrow}
\renewcommand{\bar}{\overline}
\newcommand{\bdelta}{\hat{\delta}}
\renewcommand{\cI}{\tilde{\varphi}}
\newcommand{\nnz}{\mathsf{nnz}}

\newcommand{\SolveGLM}{\textsc{SolveGLM}}
\newcommand{\lo}{\theta}
\newcommand{\smin}{s_{\min}}
\newcommand{\smax}{s_{\max}}
\newcommand{\hatsmax}{\hat{s}_{\max}}
\newcommand{\hatsmin}{\hat{s}_{\min}}

\newcommand{\upp}{u}
\newcommand{\Rplus}{\R_+}

\begin{document}

\title{Sparsifying generalized linear models}

\author{Arun Jambulapati\footnotemark[1]
\and
James R. Lee\footnotemark[2]
\and
Yang P. Liu\footnotemark[3]
\and
Aaron Sidford\footnotemark[4]}

\date{}

\maketitle

\let\oldthefootnote\thefootnote
\renewcommand{\thefootnote}{\fnsymbol{footnote}}
\footnotetext[1]{University of Washington, {\tt jmblpati@uw.edu}, \footnotemark[2]{University of Washington, {\tt jrl@cs.washington.edu}}}
\footnotetext[3]{Institute for Advanced Study, {\tt yangpliu@ias.edu}, \footnotemark[4] Stanford University, {\tt sidford@stanford.edu}}
\let\thefootnote\oldthefootnote
\renewcommand{\thefootnote}{\arabic{footnote}}

\vspace*{-0.2in}

\begin{abstract}
We consider the sparsification of sums $F : \R^n \to \Rplus$ where $F(x) = f_1(\langle a_1,x\rangle) +
\cdots + f_m(\langle a_m,x\rangle)$ for 
vectors $a_1,\ldots,a_m \in \R^n$ and functions $f_1,\ldots,f_m : \R \to \Rplus$.
We show that $(1+\e)$-approximate sparsifiers of $F$ with support
size $\frac{n}{\e^2} (\log \frac{n}{\e})^{O(1)}$ exist whenever the functions
$f_1,\ldots,f_m$ are symmetric, monotone, and satisfy natural growth bounds.
Additionally, we give efficient algorithms to compute such a sparsifier assuming each $f_i$ 
can be evaluated efficiently.

Our results generalize the classic case of $\ell_p$ sparsification, where
$f_i(z) = |z|^p$, for $p \in (0, 2]$, and give the first near-linear size sparsifiers in
the well-studied setting of the Huber loss function and its generalizations, e.g., $f_i(z) = \min\{|z|^p, |z|^2\}$ for $0 < p \leq 2$.
Our sparsification algorithm can be applied to give near-optimal
reductions for optimizing a variety of generalized linear models including $\ell_p$ regression for $p \in (1, 2]$ to high accuracy, via solving $(\log n)^{O(1)}$ sparse regression
instances with $m \le n(\log n)^{O(1)}$, plus runtime proportional to the number of nonzero entries in
the vectors $a_1, \dots, a_m$.
\end{abstract}

\makeatletter
\renewcommand\tableofcontents{%
  \@starttoc{toc}%
}
\makeatother

\begingroup
\titlecontents{title}{}{}{}{}
\hypersetup{linktocpage=false}
\setcounter{tocdepth}{2}
\tableofcontents
\endgroup

\section{Introduction}
\label{sec:intro}

Empirical risk minimization (ERM) is a widely studied problem in learning theory and statistics (see, e.g., \cite{LSZ19}, for relevant references to the expansive literature on this topic). A prominent special case is the problem of optimizing a \emph{generalized linear model (GLM)}, i.e.,
\begin{equation}\label{eq:glm}
\min_{x \in \R^n} F(x) \quad \textrm{ for } \quad F(x) := \sum_{i=1}^m f_i(\langle a_i, x \rangle - b_i)\,,
\end{equation}
where the \emph{total loss} $F : \R^n \to \R$, is defined by vectors $a_1,\ldots,a_m \in \R^n$, $b \in \R^m$, and \emph{loss
functions} $f_1, \dots, f_m: \R \to \R$. Different choices of the loss functions $\{ f_i \}$ capture important problems, including linear regression, logistic regression, and $\ell_p$ regression \cite{BCLL18,AKPS19}.

Recently, efficient algorithms for solving \eqref{eq:glm} to high-accuracy have been developed in
many settings \cite{BLSS20,BLLSSSW21,GPV21} such as linear programming and $\ell_1$-regression, where $f_i(x) = |x|$. For example, when $m$ is on the order of $n$,
it is known how to solve linear programs and some GLMs in
roughly (up to logarithmic factors) the time it currently takes to multiply two general 
$n \times n$ matrices \cite{AKPS19,CLS19,LSZ19,JSWZ21} which is, up to logarithmic factors,
the best-known, running time for solving a single linear system
in a dense $n \times n$ matrix.

When $m \gg n$, a natural approach for fast algorithms is to apply sparsification techniques
to reduce the value of $m$, while maintaining a good multiplicative
approximation of the objective value. More precisely, say that the objective
\emph{$F$ admits an $s$-sparse $\eps$-approximation} if there are non-negative
weights $w_1, \dots, w_m \in \R^m_{+}$, at most $s$ of which are non-zero, and such that

\begin{equation*}
   |F(x) - \tilde{F}(x)| \le \eps \, F(x) \enspace \text{ for all } x
   \in \R^n, \text{ where } \tilde{F}(x) := \sum_{i=1}^m
w_i f_i(\langle a_i, x \rangle - b_i)\,.
\end{equation*}

When $f_i(z) = |z|^p$ are $\ell_p$ losses, near-optimal sparsification results
are known: If $p > 0$, then $F$ admits an $s$-sparse $\eps$-approximation for $s
\le \O(n^{\max\{1,p/2\}}\eps^{-2})$;\footnote{Throughout, we use $\O(f)$ to suppress polylogarithmic quantities in $m, n, \eps^{-1}$, and $f$.} this sparsity bound is known to be optimal up to polylogarithmic
factors \cite{BLM89,Tal90,Tal95,SZ01}. In particular, for $p \in (0,  2]$, the
size is $\O(n \eps^{-2})$, near-linear in the underlying dimension $n$. The $p
= 2$ case has been especially influential in the development of several fast
algorithms for linear programming and graph optimization over the last two
decades \cite{ST14,SS11,BLNPSSSW20}.

However, as far as the authors know,
$\ell_p$ losses are the only class of natural loss functions for which
linear-size sparsification results are known for GLMs. For instance, for the widely-studied class of Huber
loss functions (see \eqref{eq:gammap}) and related variants, e.g., $f_i(z)=\min\{|z|,|z|^2\}$,
the best known sparsity bound was $\O(n^{4-2\sqrt{2}} \eps^{-2})$
\cite{MMWY22}.
Improving this bound to near-linear (in $n$) is an established important open problem that has potential applications to regression for Huber and $\ell_p$ losses \cite{AS20,ABKS21,GPV21,MMWY22,WY23}.

The main result of this paper is near-optimal sparsification for a large family
of loss functions $\{ f_i \}$ that include the Huber losses, $\ell_p$ losses, and generalizations. Informally, we show that if
the loss functions $\{ f_i \}$ are nonnegative, symmetric, and grow at most
quadratically, then there exists an $s$-sparse $\eps$-approximation of $F$ with
$s \le \O(n \eps^{-2})$. Moreover, the sparse approximation can be found very efficiently, in time proportional to the time used for $\O(1)$
instances of $\ell_2$-sparsification (\pref{thm:main}).
A particularly nice
application of our result is an algorithm that solves $\ell_p$-regression to
high accuracy for $1 < p \le 2$ by reducing to $\O_p(1)$ instances of
$\ell_p$-regression with $m = \O(n)$ (\pref{thm:pnorm}). Our framework can also be applied to minimizing sums of $\gamma_p$ functions for $p \in (1, 2]$ (see \eqref{eq:gammap}) to high accuracy, and to approximate Huber regression.

The main technical hurdle in obtaining these results is that the loss functions are not necessarily homogeneous and they can exhibit different behaviors at different scales. Note that this hurdle arises already for losses like $f_i(z) = \min\{|z|, |z|^2\}$, even though the loss function only has two different scaling regimes. To overcome this hurdle we develop a multiscale notion of ``importance scores'' for appropriately down-sampling $F$ into a sparse representation. 

\subsection{Hypotheses and results for sparsification}
\label{sec:axioms}

Consider a generalized linear model as in \eqref{eq:glm}, with loss functions  $f_1,\ldots,f_m : \R \to \R_+$ and vectors $a_1,\ldots,a_m \in \R^n$. For simplicity, we assume that $b=0$ in \eqref{eq:glm}.
This is without loss of generality, as $\langle a_i,x\rangle - b_i = \langle (a_i,b_i), (x,-1)\rangle$, and $(a_i,b_i), (x,-1) \in \R^{n+1}$,
so we can re-encode the problem in $n+1$ dimensions with $b=0$.

\newcommand{\pseudonorm}{pseudonorm\xspace}
\newcommand{\Pseudonorm}{Pseudonorm\xspace}

We will often think of the case $f_i(z) = h_i(z)^2$ for some $h_i : \R \to \R_+$, as the assumptions
we need are stated more naturally in terms of $\sqrt{f_i}$.
To that end, consider a function $h : \R^k \to \R_+$
and the following two properties, where $L \geq 1$ and $c, \lo > 0$ are some positive constants.\footnote{The setting $k=1$ suffices for the present work, though we state them for general $k \geq 1$.}

\begin{enumerate}[label=\text{(P\arabic*)}]
   \item \label{item:lipschitz} ($L$-auto-Lipschitz) $|h(z)-h(z')| \leq L\,h(z-z')$ for all $z,z' \in \R^k$.
   \item \label{item:lower-growth} (Lower $\lo$-homogeneous) $h(\lambda z) \geq c \lambda^{\lo} h(z)$ for all $z \in \R^k$ and $\lambda \geq 1$.
\end{enumerate}
Note that if $h : \R \to \R$ is concave and symmetric, then it is $1$-auto-Lipschitz (see \pref{lem:approx-concave}).

We can now state our main theorem, whose proof appears in \pref{sec:sparsification-analysis}.

\begin{theorem}
\label{thm:main}
Consider $f_1, \dots, f_m: \R \to \R_{+}$, and
suppose there are numbers $L \geq 1 ,c,\lo > 0$ such that
each $\sqrt{f_i}$ is $L$-auto-Lipschitz and lower $\lo$-homogeneous (with constant $c$).
Then for any $a_1,\ldots,a_m \in \R^n$, and numbers $0 < \e < \tfrac12$ and $\smax > \smin \geq 0$, there are nonnegative
weights $w_1,\ldots,w_m \geq 0$ such that
\[
\left| 
 F(x) - \sum_{i=1}^m w_i f_i(\langle a_i,x\rangle)
\right|
\leq \epsilon F(x)\,,\quad \forall x \in \R^n \ \textrm{ s.t. }\ \smin \leq F(x) \leq \smax\,,
\]
where $F(x) \seteq f_1(\langle a_1,x\rangle) + \cdots + f_m(\langle a_m,x\rangle)$, and
\begin{equation}\label{eq:main-sparsity}
   \left | i \in \{1,\ldots,m\} : w_i > 0 \right | \lesssim_{L,c,\theta} \frac{n}{\e^2} \log \left(\frac{n}{\e} \frac{\smax}{\smin}\right) \left(\log S\right)^3, \enspace \text{ where } \enspace S := \frac{n}{\eps} \log\left(\frac{2\smax}{\smin}\right) \,.
\end{equation}
Moreover, with high probability, the weights $\{w_i\}$ can be computed in time \[ \O_{L,c,\lo}\!\left((\nnz(a_1,\ldots,a_m) + n^{\omega} + m \tEval) \log(m\smax/\smin)\right). \]
\end{theorem}

Here, $\tEval$ is the maximum time needed to evaluate each $f_i$, $\nnz(a_1,\ldots,a_m)$
is the total number of non-zero entries in the vectors $a_1,\ldots,a_m$, and $\omega$ is the matrix multiplication exponent.
``High probability'' means that the failure probability can be made less than $n^{-\ell}$ for any $\ell > 1$
by increasing the running time by an $O(\ell)$ factor.

We use the notation $O_{L,c,\lo}$ and $\O_{L,c,\lo}(\cdot)$ to indicate an implicit dependence on the parameters $L,c,\lo$, and
$A \lesssim_{L,c,\lo} B$ is shorthand for $A \leq O_{L,c,\lo}(B)$.
The constant hidden by the $O_{L,c,\lo}(\cdot)$ notation
is about $(L/c)^{O(\lo^{-2})}$,
though we made no significant effort to optimize this dependence.

\smallskip

It is not difficult to see that for $0 < p \leq 2$, the function $f_i(z) =
|z|^p$ satisfies the required hypotheses of \pref{thm:main}. In \pref{sec:gammap}, we show that
$\gamma_p$ functions, defined as
\begin{equation}\label{eq:gammap}
   \gamma_p(z) := \begin{cases} \frac{p}{2}z^2
& \enspace \text{ for } \enspace |z| \le 1 \\ |z|^p - (1-\frac{p}{2}) &
\enspace \text{ for } \enspace |z| \ge 1, \end{cases}
\end{equation}
for $p \in (0, 2]$, also satisfy the conditions.  The special case of $\gamma_1$ is known as the Huber loss.
(See \pref{sec:gammap} for a generalization to general thresholds.)

The $\gamma_p$ functions were introduced in \cite{BCLL18} and have since
been used in several works on high-accuracy $\ell_p$ regression \cite{AKPS19,ABKS21,GPV21}.
Due to these connections, the works \cite{GPV21,MMWY22} studied sparsification
with $\gamma_p$ losses, providing sparsity bounds of $\O(n^3)$ and
$\O(n^{4-2\sqrt{2}}) \approx \O(n^{1.172})$, respectively. More precisely, \cite{MMWY22} establish
a bound of $\O(n^{1+\delta(p)})$ for $p \in [1, 2]$ with $\delta(1) =
3-2\sqrt{2}$, and $\delta(p) \to 0$ as $p \to 2$.

\subsection{Fast $\ell_p$ regression}
\label{sec:lpintro}

Combining our sparsification theorem with iterative refinement \cite{AKPS19}
yields near-optimal reductions for solving $\ell_p$ regression to \emph{high accuracy}. More specifically, we show that $\ell_p$ regression for matrices
$A \in \R^{m \times n}$ can be reduced to a sequence of $\O_p(1)$ instances with $\wt{A} \in \R^{\O(n) \times n}$.  
It is known how to solve such
instances in time $n^{\omega_0}$ for $\omega_0 := 2 + \max\left\{\frac{1}{6}, \omega-2, \frac{1-\alpha}{2}\right\}$ \cite{LSZ19}, where $\alpha$ is the dual matrix multiplication exponent. Alternatively, they can each be solved in roughly $n^{1/3}$ iterations and time $n^{\max\{\omega, 2+1/3\}}$ \cite{AKPS19}, where an ``iteration'' refers to an operation that is dominated by the cost of solving a particular $n \times n$ linear system.

\begin{theorem}[Fast $\ell_p$ regression]
\label{thm:pnorm}
There is an algorithm that given any $A \in \R^{m \times n}$, $b \in \R^m$, and $p \in (1,2]$ computes an $x$ satisfying
\[ \|A x - b\|_p^p \le (1+\e) \min_{x \in \R^n} \|A x - b\|_p^p \] in either $\O_p(n^\frac{2-p}{p+2})$ iterations and $\O_p(\nnz(A) + n^{\max\{\omega, 2+1/3\}})$ time, or $\O_p(\sqrt{n})$ iterations and $\O_p(\nnz(A) + n^{\omega_0})$ time, with high probability.
\end{theorem}

It is standard to turn a high accuracy algorithm for an optimization problem into one that solves a corresponding dual problem.
We present such an argument for $\ell_p$-regression in \pref{sec:duallp}.

\begin{theorem}[Dual of $\ell_p$ regression]
\label{thm:qnorm}
There is an algorithm that given $A \in \R^{m \times n}$, $c \in \R^m$, and $q \in [2,\infty)$ computes a $y \in \R^m$ satisfying $A^\top y = c$ and
\[ \|y\|_q^q \le (1+\eps) \min_{A^\top y = c} \|y\|_q^q \] in either $\O_q(n^\frac{q-2}{3q-2})$ iterations and $\O_q(\nnz(A) + n^{\max\{\omega, 2+1/3\}})$ time, or $\O_q(\sqrt{n})$ iterations and $\O_q(\nnz(A) + n^{\omega_0})$ time, with high probability.
\end{theorem}
Prior work \cite{JLS22} shows that $\ell_p$-regression can be solved in $\O_p(n^{1/3})$ iterations of solving a linear system for $p \in [2, \infty)$. Combining that result with \pref{thm:pnorm} shows that $\ell_p$-regression can be solved using $\O_p(n^{1/3})$ linear systems for all $p > 1$.

\subsection{Discussion of the hypotheses}
\label{sec:propintro}

Let us now discuss various hypotheses and the extent to which they are necessary for sparsifiers of
nearly-linear size to exist. In addition to the properties \ref{item:lipschitz} and \ref{item:lower-growth}, let us consider three others
that we will use frequently.
In what follows, $C, u > 0$ are positive constants and $h : \R^n \to \R_+$.

\begin{enumerate}[label=\text{(P\arabic*)}]
   \setcounter{enumi}{2}
   \item \label{item:symmetric} ($C$-symmetric) $h(z) \leq C h(-z)$ for all $z \in \R^n$.
   \item \label{item:monotone} ($C$-monotone) $h(z) \leq C h(\lambda z)$ for $\lambda \geq 1$.
   \item \label{item:upper-growth} (Upper $\upp$-homogeneous) $h(\lambda z) \leq C \lambda^{\upp} h(z)$ for all $z \in \R^n$ and $\lambda \geq 1$.
\end{enumerate}

First, note that \ref{item:lipschitz} and \ref{item:lower-growth} imply \ref{item:symmetric}--\ref{item:upper-growth}.

\begin{lemma}
\label{lem:relate}
   The following implications hold:
   \begin{enumerate}
      \item $h$ is $L$-auto-Lipschitz $\implies$ $h$ is $L$-symmetric.
      \item \label{item:lower-to-monotone} $h$ is lower $\lo$-homogeneous with constant $c$ $\implies$ $h$ is $1/c$-monotone.
      \item \label{item:implying_homogeneuous} $h$ is $L$-auto-Lipschitz and $C$-monotone $\implies$ $h$ is upper $1$-homogeneous with constant $2CL$.
   \end{enumerate}
\end{lemma}

\begin{proof}
   Since $h(0)=0$, applying the definition of $L$-auto-Lipschitz with $z=0$ gives $h(-z) \leq L h(z)$ for any $z \in \R^n$. The second implication is immediate.
   For the third, note that
   for a positive integer $k$, we have
      $h(kz) \leq \sum_{j=0}^{k-1} |h((j+1) z) - h(j z)| \leq k L h(z)$.
   Using the $L$-auto-Lipschitz property again gives
   \[
      h(\lambda z) \leq h(\lceil \lambda \rceil z) + L h((\lceil \lambda \rceil - \lambda) z) 
      \leq \lceil \lambda \rceil L \cdot h(z) + L C h(z) \leq 2 C L \lambda h(z)\,,
   \]
   where the penultimate inequality uses $C$-monotonicity.
\end{proof}

\paragraph{Symmetry \ref{item:symmetric}}
To illustrate the need for approximate symmetry, let us consider
gluing together two functions that are otherwise ``nice'' in our framework:
\[
   f(z) \seteq \begin{cases}
                  |z|^2, & z \geq 0 \\
                  |z|, & z < 0\,.
   \end{cases}
\]

Suppose that $f_1=\cdots=f_m=f$.
Consider unit vectors $\hat{a}_1,\ldots,\hat{a}_m \in \R^n$ such that $\delta_{ij} \seteq |\langle \hat{a}_i,\hat{a}_j\rangle| < \frac12$ for $i \neq j$.
A basic volume computation shows that one can choose $m \geq 2^{\Omega(n)}$.
Denote $a_i \seteq (\hat{a}_i, 1) \in \R^{n+1}$ for $i=1,\ldots,m$.

Then for $\lambda > 0$ and $x \seteq \lambda (\hat{a}_i, -\frac12)$, we have
\[
   f_j(\langle a_j,x\rangle) = f(\langle a_j, \lambda (\hat{a}_i, -\tfrac12 )\rangle) =
   \begin{cases}
      f(\lambda/2) \asymp \lambda^2 & i = j\,, \\
      f(\lambda (\delta_{ij}-\tfrac12)) \lesssim \lambda & \textrm{otherwise.}
   \end{cases}
\]
Thus in any approximate sparsifier $\tilde{F} = w_1 f_1 + \cdots + w_m f_m$, it must be that
either $w_i > 0$, or $\sum_{j \neq i} w_j \gtrsim \lambda$.
Sending $\lambda \to \infty$ shows that the latter is impossible.

\paragraph{Lower growth and monotonicity \ref{item:lower-growth}, \ref{item:monotone}}
We consider these properties together since monotonicity is a weaker property than lower homogeneity.
A natural function that does not
satisfy lower homogeneity is the \emph{Tukey loss} which, for the sake
of the present discussion, one can take as
$f_i(z) := \min\{1, |z|^2\}$, which is a natural analog of $\gamma_p$ (recall \eqref{eq:gammap}) for $p=0$.

For sparsifying GLMs with the Tukey loss, previous works have made
additional assumptions. For example, that one only ensures sparsification when
$\|a_i\|_2 \le n^{O(1)}$, and for inputs $x \in \R^n$ satisfying $\|x\|_2 \le n^{O(1)}$; see
\cite[Assumption 2]{CWW19} and the discussion afterwards, and \cite[\S
8.3]{MMWY22}.
In \pref{sec:tukey}, we show how to achieve a
$\O(n^{1+o(1)} \eps^{-2})$-sparse $\eps$-approximations under these assumptions. At a high level, the simple idea is to
consider the proxy loss functions $\hat{f}_i(z) \seteq \min\{|z|^p, |z|^2\}$ with $p$ sufficiently small.

\paragraph{Upper quadratic growth \ref{item:upper-growth}} Note that, by \pref{lem:relate},
if $\sqrt{f_i}$ satisfies \ref{item:lipschitz}, then $f_i$ is upper $2$-homogeneous.
For near-linear size sparsifiers, $2$-homogeneity is a natural condition,
since sparsifying with loss functions $f_i(x) = |x|^p$ and $p > 2$
requires the sparsifier to have at least $\Omega(n^{p/2})$ terms \cite{BLM89}.

\paragraph{The auto-Lipschitz property \ref{item:lipschitz}}
As \pref{lem:relate} shows, this property gives us approximate symmetry \ref{item:symmetric}
and upper $1$-homogeneity \ref{item:upper-growth}.
Crucially, this property also allows us to exploit the geometry of the vectors $a_1,\ldots,a_m \in \R^n$.
Note that \ref{item:lipschitz} implies
\[
   (f_i(z) - f_i(z'))^2 = (f_i(z)^{1/2} - f_i(z')^{1/2})^2\ (f_i(z)^{1/2} + f_i(z')^{1/2})^2 \le 2L^2
      f_i(z-z')(f_i(z) + f_i(z'))\,.
\]
In particular, we have
\[
   \left(\vphantom{x^{2}}f_i(\langle a_i,x\rangle) - f_i(\langle a_i, y\rangle)\right)^2
   \leq 2 L^2 \underbrace{f_i(\langle a_i,x-y\rangle)} \left(f_i(\langle a_i,x\rangle + f_i(\langle a_i,y\rangle\right).
\]
The braced term is what us allows to access the linear structure of the vectors
in our analysis.

\paragraph{Comparison to $M$-estimators} The works \cite{CW15,MMWY22} consider regression
and sparsification for what they call general $M$-estimators. Essentially, this corresponds
to the special case of our framework where all the loss functions are the same: $f_1=\cdots=f_m=M$, and one assumes $M(0)=0$, monotonicity,
and upper and lower growth lower bounds.
They additionally assume that $M$ is $p$-subadditive (for $p = 1/2$) in the sense that $M(x+y)^p \leq M(x)^p + M(y)^p$,
which is a stronger condition than the auto-Lipschitz property \ref{item:lipschitz} for $h = f_i^{1/2}$.

Under this stronger set of assumptions, the authors of \cite{MMWY22} achieve approximations with sparsity $\tilde{O}(n^{\max\{2,p/2+1\}})$, which is a factor $n$ larger than what one might hope for.
In the regime $p \leq 2$ of possible near-linear-sized sparsifiers, we close this gap: \pref{thm:main} gives sparsity $\tilde{O}(n)$.

\subsubsection{Discussion of the $\smax/\smin$ dependence}
\label{sec:kkintro}

Note that \pref{thm:main} only achieves an approximation for $\smin \leq F(x) \leq \smax$, and
there is a logarithmic dependence on $\smax/\smin$ in the sparsity bound. 
Intuitively, some dependence on $\smax/\smin$ is
necessary in the generality of \pref{thm:main} because nothing in our
assumptions precludes the functions $f_i$ from behaving nearly independently
on different scales (at least if the scales are sufficiently well separated).

In the case that each of the functions $f_1,\ldots,f_m$ is $p$-homogeneous, in the sense
that $f_i(\lambda z)=|\lambda|^p f_i(z)$, then $F$ and the sparsifier $\tilde{F}$ are both $p$-homogeneous,
and therefore the guarantee $|F(x)-\tilde{F}(x)| \leq \e$ for $F(x)=1$ already suffices to obtain $|F(x)-\tilde{F}(x)| \leq \e F(x)$ for all $x \in \R^n$,
meaning there is no scale dependence.

More generally, for $F$ satisfying the hypotheses of \pref{thm:main}, the growth
assumptions on $f_1,\ldots,f_m$ allow one to obtain weak guarantees even for $F(x) \notin [\smin,\smax]$.
For tamer functions with only a constant number of different
scaling regimes, this allows one to avoid the $\smax/\smin$ dependence
by applying such scaling arguments and a simple reduction.
For the sake of concreteness, we demonstrate this for the Huber loss (the $\gamma_1$ function as in \eqref{eq:gammap}).
A similar argument applies for all the $\gamma_p$ functionals.

\newcommand{\tw}{\tilde{w}}

\begin{lemma}
\label{lem:huber}
Consider $a_1,\ldots,a_m \in \R^n$ for $m \geq 2$, and $1/m < \e < 1$. Denote
\begin{align*}
   F(x) &\seteq w_1 \gamma_1(\langle a_1,x\rangle) + \cdots + w_m \gamma_m(\langle a_m,x\rangle) \\
   \tilde{F}(x) &\seteq \tw_1 \gamma_1(\langle a_1,x\rangle) + \cdots + \tw_m \gamma_m(\langle a_m,x\rangle)
\end{align*}
for some nonnegative weights $w,\tw \in \Rplus^m$.
Suppose that
\[
   |F(x)-\tilde{F}(x)| \leq \e F(x) \ \text{ for $x \in \R^n$ such that } w_{\min} \leq F(x) \leq 4 m^2 w_{\max}\,,
\]
where $w_{\max} \seteq \max(\max(w),\max(\tw))$ and $w_{\min} \seteq \min(w)$.
Then $\tilde{F}$ is a $2\e$-approximation to $F$.
\end{lemma}

Combining this with an analysis of the weights produced by our construction
and the guarantee of \pref{thm:main}
yields the following consequence. The proof of \pref{lem:huber} and the next result
are presented in \pref{sec:huberproofs}.

\begin{corollary}
\label{cor:huber}
For every $\e > 0$,
the function $F(x) := \gamma_1(\l a_1, x\r) + \cdots + \gamma_1(\l a_m,x\r)$ admits an $s$-sparse $\eps$-approximation for
\[ s \lesssim \frac{n}{\eps^2} \left(\log m\right) \left(\log \left(\frac{n}{\e} \log m\right)\right)^3.\]
\end{corollary}

Note that our sparsity bound has an $m$ dependence, as opposed to the classical cases of $\ell_p$ sparsification, where sparsity bounds depend only on $n$ and $\eps$. However, some $m$ dependence is not surprising, as \cite[\S 4.5]{MMWY22} present vectors $a_1,\ldots,a_m \in \R^n$ for which the sum of the sensitivities (see \eqref{eq:ss-def}) can grow doubly-logarithmically with $m$:
   \[
      \sum_{i=1}^m \max_{0 \neq x \in \R^n} \frac{\gamma_1(\langle a_i,x\rangle)}{F(x)} \gtrsim n \log \log \frac{m}{n}\,.
   \]
\cite{MMWY22} also shows that $\sum_{i=1}^m \max_{0 \neq x \in \R^n} \frac{\gamma_p(\langle a_i,x\rangle)}{F(x)} \gtrsim n \log \frac{m}{n}$ is possible for $p \in [0,1)$.

\subsection{Importance sampling and multiscale weights}
\label{sec:overview}

Given $F(x) = f_1(\langle a_1,x\rangle) + \cdots + f_m(\langle a_m,x\rangle)$, our approach to sparsification is via
importance sampling. Given a probability vector $\rho \in \R^m$ with $\rho_1,\ldots,\rho_m > 0$ and $\rho_1 + \cdots + \rho_m = 1$,
we sample $M \geq 1$ coordinates $\nu_1,\ldots,\nu_M$ i.i.d. from $\rho$, and define our potential approximator by
\[
   \tilde{F}(x) \seteq \frac{1}{M} \sum_{j=1}^M \frac{f_{\nu_j}(\langle a_{\nu_j},x\rangle)}{\rho_{\nu_j}}\,.
\]
One can easily check that this gives an unbiased estimator for every $x \in \R^n$, i.e., $\E[\tilde{F}(x)]=F(x)$.

Since we want an approximation guarantee to hold simultaneously for many $x \in \R^n$, it is natural to  analyze expressions of the form
\[
   \E \max_{F(x) \leq s} \left|F(x)-\tilde{F}(x)\right|.
\]
Analysis of this expression involves the size of discretizations of the set
$B_F(s) \seteq \{ x \in \R^n : F(x) \leq s \}$ at various granularities,
as explained in \pref{sec:subgaussian}. The key consideration (via Dudley's entropy inequality, \pref{lem:dudley})
is how well $B_F(s)$ can be covered by cells on which we have uniform control on how much the terms $f_i(\langle a_i,x\rangle)/\rho_i$
vary within each cell.

\paragraph{The $\ell_2$ case}
Let's consider the case $f_i(z) = |z|^2$ so that $F(x) = |\langle a_1,x\rangle|^2 + \cdots + |\langle a_m,x\rangle|^2$.
Here, $B_F(s) = \{ x \in \R^n : \|Ax\|_2^2 \leq s \}$, where $A$ is the matrix with $a_1,\ldots,a_m$ as rows.

A cell at scale $2^{j}$ looks like
\[
   \sfK_j \seteq \left\{ x \in \R^n : \max_{i \in [m]} \frac{|\langle a_i,x\rangle|^2}{\rho_i} \leq 2^{j} \right\},
\]
and the pertinent question is how many translates of $\sfK_j$ it takes to cover $B_F(s)$.
In the $\ell_2$ case, this is the well-studied problem of covering Euclidean balls by $\ell_{\infty}$ balls.

If $N_j$ denotes the minimum number of such cells required, then the dual-Sudakov inequality (see \pref{lem:dualsudakov}
and \pref{cor:ds})
tells us that
\[
   \log N_j \lesssim \frac{s}{2^j} \log (m) \max_{i \in [m]} \frac{\|(A^{\top} A)^{-1/2} a_i\|^2_2}{\rho_i}\,.
\]
Choosing $\rho_i \seteq \frac{1}{n} \|(A^\top A)^{-1/2} a_i\|_2^2$, i.e., normalized leverage scores, yields uniform
control on the size of the coverings:
\[
   \log N_j \lesssim \frac{s}{2^j} n \log m\,.
\] 

\paragraph{The $\ell_p$ case, $1 \leq p < 2$}
Consider the case $f_i(z) = |z|^p$ so that $F(x) = \|Ax\|_p^p$.
A cell at scale $2^j$ now looks like
\[
   \sfK_j \seteq \left\{ x \in \R^n : \max_{i \in [m]} \frac{|\langle a_i,x\rangle|^p}{\rho_i} \leq 2^j \right\},
\]
To cover $B_F(s)$ by translates of $\sfK_j$, we again employ Euclidean balls,
and use $\ell_p$ Lewis weights to relate the $\ell_p$ structure to an $\ell_2$ structure.

A classical result of Lewis \cite{Lewis79} (see also \cite{BLM89,CP15}) establishes
that there are nonnegative weights $w_1,\ldots,w_m \geq 0$
such that if $W = \diag(w_1,\ldots,w_m)$ and
$U \seteq (A^{\top} W A)^{1/2}$, then
\begin{equation}\label{eq:wis}
   w_i = \frac{\|U^{-1} a_i\|_2^p}{\|U^{-1} a_i\|_2^2} = \frac{f_i(\|U^{-1} a_i\|_2)}{\|U^{-1} a_i\|_2^2}\,.
\end{equation}
Assuming that $A$ has full rank, a straightforward calculation gives $\sum_{i=1}^m w_i \|U^{-1} a_i\|_2^2 = \tr(U^2 U^{-2}) = n$.

Therefore, we can choose $\rho_i \seteq \frac{1}{n} w_i \|U^{-1} a_i\|_2^2$ for $i=1,\ldots,m$, and our cells become
\[
   \sfK_j \seteq \left\{ x \in \R^n : \max_{i \in [m]} |\langle a_i,x\rangle|^p \leq \frac{2^j}{n} w_i \|U^{-1} a_i\|_2^2 \right\}.
\]
(Note that the values $\{ w_i \|U^{-1} a_i\|_2^2 : i =1,\ldots,m\}$ are typically referred to as the ``$\ell_p$ Lewis weights''.)

If we are trying to use $\ell_2$-$\ell_{\infty}$ covering bounds, we face an immediate problem:
Unlike in the $\ell_2$ case, we don't have prior control on $\|U x\|_2$ for $x \in B_F(s)$.
One can obtain an initial bound using the structure of $U = (A^{\top} W A)^{1/2}$:
\begin{align}
\|U x\|_2^2 = \sum_{i=1}^m w_i \langle a_i,x\rangle^2 &\stackrel{\eqref{eq:wis}}{=} \sum_{i=1}^m \|U^{-1} a_i\|_2^{p-2} \langle a_i,x\rangle^2  \nonumber \\
               &= \sum_{i=1}^m \left(\frac{|\langle a_i,x\rangle|}{\|U^{-1} a_i\|_2}\right)^{2-p} |\langle a_i,x\rangle|^p \leq \|U x\|_2^{2-p} \sum_{i=1}^m |\langle a_i,x\rangle|^p\,,
               \label{eq:intro-nc0}
\end{align}
where the last inequality is Cauchy-Schwarz: $|\langle a_i, x\rangle| = |\langle U^{-1} a_i, U x\rangle| \leq \|U^{-1} a_i\|_2 \|U x\|_2$.
This gives the bound $\|U x\|_2 \leq \|Ax\|_p \leq s^{1/p}$ for $x \in B_F(s)$.

Problematically, this uniform $\ell_2$ bound is too weak, but there is a straightforward solution:
Suppose we cover $B_F(s)$ by translates of $\sfK_{j_0}$. This gives an $\ell_{\infty}$ bound
on the elements of each cell, meaning that we can apply \eqref{eq:intro-nc0} 
and obtain a better upper bound on $\|Ux\|_2$ for $x \in \sfK_{j_0}$.
Thus to cover $B_F(s)$ by translates of $\sfK_j$ with $j < j_0$, we will cover 
first by translates of $\sfK_{j_0}$, then cover each translate $(x + \sfK_{j_0}) \cap B_F(s)$ by translates of $\sfK_{j_0-1}$, and so on.

The standard approach in this setting (see \cite{BLM89} and \cite[\S 15.19]{LedouxTalagrand2011})
is to instead use interpolation inequalities and duality of covering numbers for a cleaner analytic
version of such an iterated covering bound.
However, the iterative covering argument can be adapted to the non-homogeneous setting, as we discuss next.

\paragraph{Generalized linear models}
When we move to more general loss functions $f_i : \R \to \R$, we lose the homogeneity
property $f_i(\lambda x) = \lambda^p f_i(x), \lambda > 0$ that holds for $\ell_p$ losses.
Because of this, we need to replace the single Euclidean structure present in \eqref{eq:wis}
(given by the linear operator $U$) with a family of structures, one for every relevant scale.

\begin{definition}[Approximate weights]
\label{def:approxweights}
Fix $a_1,\ldots,a_m \in \R^n$ and loss functions $f_1,\ldots,f_m~:~\R \to \R_+$.
We say that a vector $w \in \R_+^m$ is an {\em $\alpha$-approximate weight at scale $s$} if
\begin{equation}\label{eq:approx-weight}
   \frac{s}{\alpha} \leq \frac{f_i(\|M_w^{-1/2} a_i\|_2)}{w_i \|M_w^{-1/2} a_i\|_2^2} \leq \alpha s\,,\quad i=1,\ldots,m\,,\text{ where }M_w \seteq \sum_{j=1}^m w_j a_j a_j^{\top}\,.
\end{equation}
\end{definition}

To motivate this definition, let us define scale-specific sentivities:
\begin{equation}\label{eq:ss-def}
   \xi_i(s) \seteq \max \left\{ \frac{f_i(\langle a_i,x\rangle)}{F(x)} : x \in \R^n, F(x) \in [s/2, s] \right\},\quad i=1,\ldots,m\,.
\end{equation}
As shown in \pref{cor:ss-sensitivity}, if the functions $\{f_i\}$ are lower $\theta$-homogeneous, upper $2$-homogeneous, and $O(1)$-symmetric
(in the sense of \ref{item:symmetric}), then an $\alpha$-approximate weight at scale $s$ allows us to upper bound sensitivies by leverage scores:
\begin{equation}\label{eq:ss-xi}
   \xi_i(s) \ls \sigma_i(W^{1/2} A)\,,
\end{equation}
where $W = \diag(w_1,\ldots,w_m)$, and the implicit constant depends on $\alpha$ and the homogeneity parameters.
Here, $\sigma_i(V)$ denotes the $i$th leverage score of a matrix $V$ with rows $v_1,\ldots,v_m$:
\begin{equation}\label{eq:ls-def}
   \sigma_i(V) \seteq \langle v_i, (V^{\top} V)^{+} v_i\rangle\,,
\end{equation}
where $(V^{\top} V)^+$ denotes the Moore-Penrose pseudoinverse. Notably,
one always has $\sigma_1(V)+\cdots+\sigma_m(V) = \rank(V)$, and therefore \eqref{eq:ss-xi} gives
an upper bound $\xi_1(s)+\cdots+\xi_m(s) \ls n$.

In order to generalize the iterated covering argument for $\ell_p$ losses,
we need there to be a relationship between weights at different scales.

\begin{definition}[Weight schemes]
\label{def:weight-scheme}
Let $\cJ \subseteq \Z$ be a contiguous interval. A family $\{w^{(j)} \in \R_+^m : j \in \cJ \}$ is an {\em $\alpha$-approximate weight scheme}
if each $w^{(j)}$ is an $\alpha$-approximate weight at scale $2^j$ and, furthermore,
for every pair $j,j+1 \in \cJ$ and $i \in \{1,\ldots,m\}$,
\begin{equation}\label{eq:smooth-weights}
   w_i^{(j+1)} \leq \alpha w_i^{(j)}\,.
\end{equation}
\end{definition}

Given a weight scheme, we choose sampling probabilities 
\[
   \rho_i \propto \max_{j \in \cJ} w_i^{(j)} \|M_{w^{(j)}}^{-1/2} a_i\|_2^2 = \max_{j \in \cJ} \sigma_i(W_j^{1/2} A)\,,\quad i=1,\ldots,m\,,
\]
where $W_j = \diag(w^{(j)}_1, \ldots, w^{(j)}_m)$.
In our setting, $|\cJ| \le O(\log(m\smax/\smin))$, which results in the sparsity increasing by a corresponding factor.

In \pref{sec:weights}, we establish the existence of approximate weight schemes
for general families of loss functions satisfying certain growth bounds,
along with efficient algorithms to compute the corresponding weights.

\subsection{Regression via iterative refinement}
\label{sec:overviewpnorm}

Previous works have observed that combining \emph{iterative refinement} with
sparsification of $\gamma_p$-functions (recall \eqref{eq:gammap}) leads to
improved algorithms for $\ell_p$-regression \cite{AKPS19,ABKS21,GPV21}. For the
benefit of the reader, we give a description of these ideas in somewhat more generality.

Recall that our goal is to find a point $x \in \R^n$ that computes an approximate minimizer of
$F(x) \seteq \sum_{i=1}^m f_i(\l a_i, x \r - b_i)$, up to high accuracy. For
now, we assume that $F$ is a differentiable convex function and denote $F_* \seteq \inf_{x \in \R^n} F(x)$.
Later, we will introduce additional conditions that allow for iterative refinement to succeed.

Broadly, iterative refinement minimizes $F(x)$ be repeatedly solving sub-problems, each of which make multiplicative progress in reducing the error of 
the current solution. Given a current point $x_0$, prior works on iterative refinement define a local approximation of $F$ suitably symmetrized and centered around $x_0$ such that approximately minimizing this local approximation yields the desired decrease in function error. One way to derive such local approximations is through Bregman divergences, which give a natural way of recentering convex functions.

\begin{definition}[The $F$-divergence]
   \label{def:F-divergence}
For $x,y\in \R^n$, use
$T_x^F(y) \seteq F(x) + \nabla F(x)^\top (y - x)$ to denote the first order Taylor
approximation of $F$ at $x$, and define the
\emph{$F$-induced Bregman divergence} by
$D_x^F(y) \seteq F(y) - T_x^F(y)$.
\end{definition}

Note that for convex $F$, the function $D_x^F(y)$ is convex and minimized at $x$. Consequently, given a point $x_0$, the function $ D_{x_0}^F(y)$ is a natural function induced by $F$ and minimized at $x_0$. Note that minimizing $F(x)$ is the same as minimizing $\langle \nabla F(x_0), x - x_0 \rangle + D_{x_0}^F(x)$ which in turn is the same as minimizing $\langle \nabla F(x), \Delta \rangle + D_{x_0}^F(x_0 + \Delta)$ over $\Delta$ and adding the minimizer to $x_0$.

Iterative refinement strategies approximately minimize $\langle \nabla F(x), \Delta \rangle + r(\Delta)$, where $r(\Delta)$ is a suitable approximation of $D_{x_0}^F(x_0 + \Delta)$.
One step of refinement moves to $x_1 \seteq x_0 + \eta \tilde{\Delta}$, where $\eta$ is a suitably chosen step-size and $\tilde{\Delta}$ is the approximate minimizer. As motivation for our approach, here we consider the scheme suggested by prior work, where $r(\Delta)$ is a sparsification of a simple approximation to the divergence.

Informally, one can show that if the square root of the Bregman divergence of each $f_i$ is
$L$-auto-Lipschitz \ref{item:lipschitz} and lower
$\lo$-homogeneous for some $\lo > 1$ \ref{item:lower-growth}, then one step of 
sparsification/refinement decreases the error in the objective value multiplicatively by an absolute constant. Interestingly, auto-Lipschitzness is only required for sparsification and not for refinement. However, we critically need the Bregman
divergence to be lower $\lo$-homogeneous for $\lo > 1$ for iterative refinement, while our sparsification
results (\pref{thm:main}) only require $\lo > 0$.

\begin{lemma}[Refinement Lemma]
\label{lem:refine}
Suppose $r : \R^n \rightarrow \R_{+}$ is lower $\lo$-homogeneous with constant $c < 1$ for $\lo > 1$, and for $x_0 \in \R^n$ and all $\Delta \in \R^n$ and $\eta \in [0, 1]$,
\begin{align*}
   r(\Delta) &\leq D_{x_0}^F(x_0 + \Delta) \leq \alpha \, r(\Delta)
\end{align*}
where $\alpha \geq 1$ is fixed.
Then $\inf_{\Delta \in \R^n}
\left\{
T_{x_0}^F(x_0 + \hat{\Delta}) + r(\hat{\Delta})
\right\} \leq F_*$ and if 
$\hat{\Delta} \in \R^n$ satisfies
\begin{equation}\label{eq:delta-hat}
T_{x_0}^F(x_0 + \hat{\Delta}) + r(\hat{\Delta})
\leq F_*\,,
\end{equation}
then
\begin{equation}
\label{eq:progress}
F(x_0 + \hat{\eta} \hat{\Delta}) - F_*
\leq (1 - \hat{\eta}) \left(f(x_0) - F_*\right),
\quad \text{ where } \hat{\eta} \seteq (\alpha/c)^{-1/(q - 1)}\,.
\end{equation}
\end{lemma}

\begin{proof}
First note that, for all $\Delta \in \R^n$,
\begin{align}
   F(x_0 + \Delta)
   &= T_{x_0}^F(x_0 + \Delta) + D_{x_0}^F(x_0 + \Delta)\,.
   \label{eq:div_expand}
\end{align}
Since, $D_{x_0}^F(x_0 + \Delta) \geq r(\Delta)$ this implies the desired bound
\[
\inf_{\Delta \in \R^n}
\left\{
T_{x_0}^F(x_0 + \Delta) + r(\Delta)
\right\}
\leq \inf_{\Delta \in \R^n} F(x_0 + \Delta)
= F_* \,.
\]

Next, note that for all $\eta \in [0,1]$ and $\Delta \in \R^n$,
\begin{align}
\label{eq:damp_taylor}
T_{x_0}^F(x_0 + \eta \Delta)
&= 
F(x_0) + \eta \nabla F(x_0)^\top \Delta 
= (1 - \eta) F(x_0) + \eta\, T_{x_0}^F(x_0 + \Delta) \\
   \label{eq:damp_diverge}
   D_{x_0}^F(x_0 + \eta \Delta)
   &\leq
   \alpha r(\eta \Delta)
   \leq \alpha/c \cdot  \eta^{\lo} r(\Delta)
\end{align}
Suppose that $\hat{\Delta} \in \R^n$ satisfies \eqref{eq:delta-hat}.
Then plugging \eqref{eq:damp_taylor} and \eqref{eq:damp_diverge} into \eqref{eq:div_expand} with $\Delta = \hat{\eta} \hat{\Delta}$ yields
\begin{align*}
F(x_0 + \hat{\eta} \hat{\Delta})
&\leq (1 - \hat{\eta}) F(x_0) + \hat{\eta} \, T_{x_0}^F(x_0 + \hat{\Delta}) + \alpha/c \cdot \hat{\eta}^{\lo} r(\hat{\Delta})
\\
&= (1 - \hat{\eta}) F(x_0) 
+ \hat{\eta} \left(T_{x_0}^F(x_0 + \hat{\Delta})  + r(\hat{\Delta})\right)
\leq (1 - \hat{\eta}) F(x_0) + \hat{\eta} F_*
\end{align*}
where we used that $\alpha/c \cdot \hat{\eta}^{\lo - 1} = 1$ and $\hat{\eta} \in [0, 1]$. Rearranging yields \eqref{eq:progress}.
\end{proof}

To apply \pref{lem:refine} to ERM for general linear models, note that
\[ D_{x_0}^F(x_0 + \Delta) = \sum_{i=1}^m D_{\l a_i, x_0 \r - b_i}^{f_i}(\l a_i, \Delta\r)\,. \]
Now, if the square root of each divergence $D_z^{f_i}$ is $L$-auto-Lipschitz and lower $\lo$-homogeneous,
then \pref{thm:main} gives weights $w \in \R^m_{+}$ with sparsity $\O(n)$ such
that
\[ 0.9 \cdot D_{x_0}^F(x_0 + \Delta) \le r(\Delta) \le D_{x_0}^F(x_0 + \Delta)\quad \enspace\text{where}\quad \enspace r(\Delta) \seteq \sum_{i=1}^m w_i D_{\l a_i, x_0 \r - b_i}^{f_i}(\l a_i, \Delta\r). \]
Thus, \pref{lem:refine} applies, and we can decrease the objective value error
by a multiplicative factor by minimizing $r(\Delta)$. Since $r(\Delta)$ only has $\O(n)$ nonzero terms,
one can apply previous solvers for $m = \O(n)$ \cite{LSZ19,AKPS19} to obtain the desired runtimes.

In \pref{sec:lp}, we verify that the $f$-divergence of $f(z) = |z|^p$ is the $\gamma_p$ function,
which is lower $p$-homogeneous, and has an auto-Lipschitz square root. This yields our algorithm for $\ell_p$ regression (\pref{thm:pnorm}).
A formal version of the argument is presented in \pref{sec:app}. The main
difference is that some technical work is needed because
\pref{thm:main} only provides sparsification for a range of inputs
$\{x \in \R^n : \smin \le F(x) \le \smax\}$. Moreover, our algorithm uses an
approximate oracle for GLMs (\pref{def:glmoracle}) with the functions $\{f_i\}$ (rather than with the approximate divergence $r(\cdot)$),
and we show that the oracle does not need to be solved to high accuracy to make sufficient progress.

\subsection{Preliminaries}
\label{sec:prelim}

Throughout the paper, we 
denote $[n] := \{1, 2, \dots, n\}$.
We use the notation $a \ls b$ to denote that there is a universal constant $C$ such that $a \le Cb$, and $a \ls_L b$
to denote that $C$ may depend on $L$. We use $a \asymp b$ to denote the conjunction of $a \ls b$ and $b \ls a$, and $a \asymp_L b$ analogously.
We also denote $\R_+ \seteq \{ x \in \R : x \geq 0 \}$ and $\R_{++} \seteq \{ x \in \R : x > 0 \}$.

For simplicity of presentation, we assume that the vectors $a_1,\ldots,a_m \in \R^n$ in \eqref{eq:glm} span $\R^n$ and all are nonzero.
In particular, this means that the matrix $A$ with rows $a_1,\ldots,a_m$ has rank $n$ and $A^{\top} W A$ is invertible
for any diagonal matrix $W$ with strictly positive entries on the diagonal.

\subsubsection{Covering numbers and chaining}
\label{sec:chaining}

Consider a metric space $(T,d)$.
For $x \in T$ and $r > 0$, define the ball $B(x,r) \seteq \{ y \in T : d(x,y) \leq r \}$.
\begin{definition}[Covering numbers]

\label{def:cov}
For a radius $r > 0$, we define the {\em covering number} $\cov(T,d,r)$ as the
smallest number of balls of radius $r$ (in the distance $d$) that are required to cover $T$.
For $S,S' \subseteq \R^n$, we overload notation and use $\cov(S,S')$ to denote the smallest
number of translates of $S'$ needed to cover $S$.
\end{definition}

We require the following ``dual Sudakov inequality'' (see \cite{PT85} and \cite[(3.15)]{LedouxTalagrand2011})  which gives bounds for covering the Euclidean ball
using balls in an arbitrary norm.

\begin{lemma}[Dual Sudakov inequality]
\label{lem:dualsudakov}
Let $B_2^n$ denote the unit ball in $n$ dimensions, and $\|\cdot\|_X$ an
arbitrary norm on $\R^n$.
If $\bm{g}$ is a standard $n$-dimensional Gaussian, then
\[ \sqrt{\log \cov(B_2^n, B_X)} \ls \E \|\bm{g}\|_X\,, \]
where $B_X \seteq \{ y \in \R^n : \|y\|_X \leq 1 \}$.
\end{lemma}

We recall Talagrand's generic chaining functional \cite[Def. 2.2.19]{TalagrandBook2014}:
\begin{equation}\label{eq:gamma2}
   \gamma_2(T,d) \seteq \inf_{\{\cA_h\}} \sup_{x \in T} \sum_{h=0}^{\infty} 2^{h/2} \diam(\cA_h(x),d)\,,
\end{equation}
where the infimum runs over all sequences $\{\cA_h : h \geq 0\}$ of partitions of $T$ satisfying $|\cA_h| \leq 2^{2^h}$
for each $h \geq 0$.
Note that we use the notation $\cA_h(x)$ for the unique set of $\cA_h$ that contains $x$.

The chaining functional is used to control the maximum of subgaussian processes (see, e.g.,
the discussion in \cite[\S 2.2]{JLLS23} where it is applied precisely in the setting
of sparsification).
Our use of the functional occurs only in the statement of \pref{lem:sparsification-meta} below,
and in this paper we will only require the following classical upper bound.
(See, eg., \cite[Prop 2.2.10]{TalagrandBook2014}.)

\begin{lemma}[Dudley's entropy bound]
   \label{lem:dudley}
   For any metric space $(T,d)$, it holds that
   \[
      \gamma_2(T,d) \lesssim \sum_{j \in \Z} 2^j \sqrt{\log \cN(T,d,2^j)}\,.
   \]
\end{lemma}

The interested reader will note that this is (up to constants)
precisely the upper bound one obtains by choosing $\cA_h$ as a uniform discretization of $(T,d)$,
i.e., to minimize $\sup \{ \diam(\cA_h(x),d) : x \in T \}$ over all partitions satisfying $|\cA_h| \leq 2^{2^h}$.

\subsubsection{Sparsification via subgaussian processes}
\label{sec:subgaussian}
We discuss sparsification via subgaussian processes. Consider $\f_1,\f_2,\ldots,\f_m : \R^n \to \R$, and define
\[
   F(x) \seteq \sum_{j=1}^m \f_j(x)\,.
\]
Given a strictly positive probability vector $\rho \in \R_{++}^m$,
and an integer $s \geq 1$ and $\nu=(\nu_1,\ldots,\nu_s) \in [m]^s$,
define the distance
\begin{equation}
\label{eq:d_rho_nu}
   d_{\rho,\nu}(x,y) \seteq \left(\sum_{j=1}^s \left(\frac{\f_{\nu_j}(x)-\f_{\nu_j}(y)}{\rho_{\nu_j} \,s}\right)^2\right)^{1/2}.
\end{equation}
and the function $\tilde{F}_{\rho,\nu} : \R^n \to \R$ 
\begin{align*}
   \tilde{F}_{\rho,\nu}(x) &\seteq \frac{1}{s} \sum_{j=1}^s \frac{\f_{\nu_j}(x)}{\rho_{\nu_j}}\,.
\end{align*}

\newcommand{\Comega}[1]{\left\|#1\right\|_{C(\Omega)}}
\newcommand{\Fomega}{\Comega{F}}
\newcommand{\Ftomega}{\Comega{\tilde{F}_{\rho,\nu}}}
\newcommand{\Ftomegabold}{\Comega{\tilde{F}_{\rho,\bm{\nu}}}}

We require the following lemma which employs a variant of a standard symmetrization argument to control $\E \max_{x \in \Omega} |F(x)-\tilde{F}_{\rho,\bm{\nu}}(x)|$ using an
associated Bernoulli process (see, for example, \cite[Lem 9.1.11]{TalagrandBook2014}).
For a subset $\Omega \subseteq \R^n$, denote $\Fomega \seteq \sup_{x \in \Omega} |F(x)|$.
The reader will note our typical application of the lemma
to sets of the form $\Omega = \{ x \in \R^n : F(x) \le \lambda \}$ for some parameter $\lambda > 0$.

\begin{lemma}[{\cite[Lemma 2.6]{JLLS23}}]
   \label{lem:sparsification-meta}
   Consider $s \geq 1$, a subset $\Omega \subseteq \R^n$, and a probability vector $\rho \in \R^m_+$.
   Assume that
   \begin{equation}\label{eq:haszero}
      \exists x_0 \in \Omega \quad \textrm{s.t.} \quad \f_1(x_0)=\cdots=\f_m(x_0)=0\,.
   \end{equation}
   Suppose, further, that for some $0 < \delta \leq 1$, and every $\nu \in [m]^s$, it holds that
   \begin{equation}\label{eq:gamma-assumption}
      \gamma_2(\Omega, d_{\rho,\nu}) \leq \delta \left(\Fomega \Ftomega\right)^{1/2}\,.
   \end{equation}
   If $\bm{\nu}_1,\ldots,\bm{\nu}_s$ are sampled independently from $\rho$, then
   \begin{equation}\label{eq:sm-expectation}
      \E \max_{x \in \Omega} \left|F(x)-\tilde{F}_{\rho,\bm{\nu}}(x)\right| \lesssim \E\left[\gamma_2(\Omega,d_{\rho,\bm{\nu}})\right] \leq 8 \delta \Comega{F}.
   \end{equation}
   If it also holds that, for all $\nu \in [m]^s$,
   \begin{equation}\label{eq:diam-assumption}
      \diam(\Omega, d_{\rho,\nu}) \leq \bdelta \left(\Fomega \Ftomega\right)^{1/2}\,,
   \end{equation}
   then there is a universal constant $K > 0$ such that for all $0 \leq t \leq \frac{1}{2K\bdelta}$,
   \begin{equation}\label{eq:max-tail}
      \Pr \left(\max_{x \in \Omega} \left|F(x) - \tilde{F}_{\rho,\bm{\nu}}(x)\right| > K(\delta + t \bdelta) \Fomega \right) \leq e^{-Kt^2/4}.
   \end{equation}
\end{lemma}

\section{Multiscale importance scores}
\label{sec:weights}

Recall the definitions of approximate weights (\pref{def:approxweights})
and weight schemes (\pref{def:weight-scheme}).
In the present section, we prove the following
two results.

\begin{theorem}\label{thm:approx-weights}
   Suppose that $f_1,\ldots,f_m : \R \to \R_+$ are lower $\lo$-homogeneous 
   and upper $u$-homogeneous with $u > \lo > 0$ and uniform constants $c, C > 0$.
   Then there is some $\alpha = \alpha(\lo,c,u,C) > 1$ such that for
   every choice of vectors $a_1,\ldots,a_m \in \R^n$ and $s > 0$,
   there is an $\alpha$-approximate weight at scale $s$.
\end{theorem}

This is proved in \pref{sec:variational} by considering critical points of the functional
$U \mapsto \det(U)$ subject to the constraint $G(U) \leq s$, where 
$G(U) \seteq f_1(\|U a_1\|_2) + \cdots + f_m(\|U a_m\|_2)$, which
can be seen as a generalization of Lewis' original method.

\paragraph{Single-scale sensitivities}
Let us now observe that, in the case $u \leq 2$, \pref{thm:approx-weights} allows us to bound sensitivities
(recall \eqref{eq:ss-def}).
The next lemma is a generalization of \eqref{eq:intro-nc0}.

\begin{lemma}\label{lem:s2-norm-bound}
   Suppose $f_1,\ldots,f_m : \R \to \R_+$ satisfy the assumptions of \pref{thm:approx-weights} with $u \leq 2$, and
   they are additionally $K$-symmetric in the sense of \ref{item:symmetric}.
   If $w \in \Rplus^m$ is an $\alpha$-approximate weight at scale $s$,
   then for any $x \in \R^n$, it holds that
   \[
      \|M_w^{1/2} x\|_2^{\theta} \leq \max\left(1,\alpha \frac{C K}{c} \frac{F(x)}{s}\right)\,.
   \]
\end{lemma}

\begin{proof}
   We may clearly assume that $\|M_w^{1/2} x\|_2 \geq 1$.
   Then using the Cauchy-Schwarz inequality $|\langle a_i,x\rangle| \leq \|M_w^{-1/2} a_i\|_2 \|M_w^{1/2} x\|_2$ together with
   the upper quadratic growth assumption gives 
   \begin{equation}\label{eq:s2bn1}
      \frac{f_i(|\langle a_i,x\rangle|)}{|\langle a_i,x\rangle|^2}  \geq \frac1{C} \frac{f_i(\|M_w^{1/2} x\|_2 \|M_w^{-1/2} a_i\|_2)}{\|M_w^{1/2} x\|^2_2 \|M_w^{-1/2} a_i\|^2_2}
      \geq \frac{c}{C} \frac{\|M_w^{1/2} x\|_2^{\lo} \,f_i(\|M_w^{-1/2} a_i\|_2)}{\|M_w^{1/2} x\|^2_2 \|M_w^{-1/2} a_i\|^2_2}\,,
   \end{equation}
   where the last inequality uses the lower growth assumption.

   Using $M_w = \sum_{i=1}^m w_i a_i a_i^{\top}$, we can bound
   \begin{align*}
      \|M_w^{1/2} x\|_2^2 = \sum_{i=1}^m w_i\langle a_i,x\rangle^2
      &\sr{\eqref{eq:s2bn1}}{\leq} \frac{C}{c} \|M_w^{1/2} x\|_2^{2-\lo} \sum_{i=1}^m w_i f_i(|\langle a_i,x\rangle|) \frac{\|M_w^{-1/2} a_i\|^2_2}{f_i(\|M_w^{-1/2} a_i\|_2)}
      \\
      &\leq \alpha \frac{C K}{c} \frac{1}{s} \|M_w^{1/2} x\|_2^{2-\lo} \sum_{i=1}^m f_i(\langle a_i,x\rangle)\,,
   \end{align*}
   where the last inequality uses the defining property of an $\alpha$-approximate weight at scale $s$, along with the assumption
   of $K$-symmetry: $f_i(|\langle a_i,x\rangle|) \leq K f_i(\langle a_i,x\rangle)$ for all $i=1,\ldots,m$.
\end{proof}

\begin{corollary}[Sensitivity upper bound]
   \label{cor:ss-sensitivity}
   Under the assumptions of \pref{lem:s2-norm-bound}, it holds that
   \[
      \xi_1(s) + \cdots + \xi_m(s) \ls n\,,
   \]
   where the implicit constant depends on the parameters $\alpha,\theta,C,c,K$.
\end{corollary}

\begin{proof}
   From \pref{lem:s2-norm-bound}, if $F(x) \leq s$, then $\|M_w^{1/2} x\|_2 \leq (\alpha CK/c)^{1/\lo}$. By Cauchy-Schwarz, this gives
   $|\langle a_i,x\rangle| \leq (\alpha CK/c)^{1/\lo}\|M_w^{-1/2} a_i\|_2$, and therefore 
   \begin{align*}
      f_i(\langle a_i,x\rangle) \leq K f_i(|\langle a_i,x\rangle|) &\leq \frac{K}{c} f_i\left((\alpha CK/c)^{1/\lo}\|M_w^{-1/2} a_i\|_2\right) \\ 
      &\leq \frac{CK}{c}\left(\frac{\alpha CK}{c}\right)^{2/\lo} f_i(\|M_w^{-1/2} a_i\|_2)
      \leq \frac{\alpha CK}{c}\left(\frac{\alpha CK}{c}\right)^{2/\lo}  s \cdot w_i \|M_w^{-1/2} a_i\|_2^2\,,
   \end{align*}
   where the first inequality uses $K$-symmetry, the second uses $1/c$-monotonicity (which holds by \pref{lem:relate}),
   the third inequality uses upper homogeneity,
   and the last inequality uses that $w$ is an $\alpha$-approximate weight at scale $s$.
   Finally, one notes that $\sum_{i=1}^m w_i \|M_w^{-1/2} a_i\|_2^2 = \tr\left(M_w^{-1} M_w\right) = n$, completing the proof.
\end{proof}

We are only able to establish the existence of entire weight schemes (where
the weights at adjacent scales are related) for $u < 4$,
which suffices our applications, as $u \leq 2$ is a requirement for \pref{thm:main}.
The following theorem is proved in \pref{sec:weight-schemes},
based on the contractive iteration method introduced
by Cohen and Peng \cite{CP15}.

\begin{theorem}\label{thm:weight-schemes}
   Suppose that $f_1,\ldots,f_m : \R \to \R_+$ are lower $\lo$-homogeneous 
   and upper $u$-homogeneous with $4 > \upp > \lo > 0$ and uniform constants $c, C > 0$.
   Then there is some $\alpha = \alpha(\upp,C,\lo,c)$ such that
   for every choice of vectors $a_1,\ldots,a_m \in \R^n$,
   there is an $\alpha$-approximate weight scheme $\{ w_i^{(j)} : j \in \Z \}$. 
\end{theorem}

In the next section,
we show how to compute an approximate weight scheme $\{ w_i^{(j)} : j \in \cJ \}$ using $\O(|\cJ|)$ computations of
leverage scores $(\sigma_1(V),\ldots,\sigma_m(V))$ for matrices of
the form $V = A^{\top} W A$.

\subsection{Contractive algorithm}
\label{sec:weight-schemes}

For a weight $w \in \R_{+}^m$ and $i \in \{1,\ldots,m\}$, define
\[
   \tau_i(w) \seteq \frac{\sigma_i(W^{1/2} A)}{w_i} = \langle a_i, (A^{\top} W A)^{-1} a_i\rangle\,, \quad W \seteq \diag(w_1,\ldots,w_m)\,,
\]
and denote $\tau(w) \seteq (\tau_1(w),\ldots,\tau_m(w))$.

Fix a scale parameter $s > 0$ and define the iteration $\f_s : \R_+^m \to \R_+^m$ by
\begin{equation}\label{eq:phi-iter}
   (\f_s(w))_i \seteq \frac{1}{s} \frac{f_i(\sqrt{\tau_i(w)})}{\tau_i(w)}\,.
\end{equation}
Write $\f^k \seteq \f \circ \cdots \circ \f$ for the $k$-fold composition of $\f$.
In this case where $f_i(z)=|z|^p$ and $1 \leq p \leq 2$, it is known, for $s = 1$,
starting from any $w_0 \in \R_{+}^m$, the sequence $\{\f_1^k(w_0) : k \geq 1\}$ converges
to the unique fixed point of $\f$, which are the corresponding $\ell_p$ Lewis weights \eqref{eq:wis}.

Define now a metric $d$ on $\R_+^m$ by
\begin{align*}
   d(u,w) &\seteq \max \left\{ \left|\log \frac{u_i}{w_i}\right| : i = 1,\ldots,m \right\}.
\end{align*}
We note the following characterization.

\begin{fact}\label{fact:tau-char}
A vector $w \in \R_+^m$ is an $\alpha$-approximate weight at scale $s$
if and only if
\[
   d(w,\f_s(w)) \leq \log \alpha\,.
\]
\end{fact}

First, we observe that $\tau$ is $1$-Lipschitz on $(\R_+^m, d)$.
In the next proof, $\preceq$ denotes the ordering of two real, symmetric matrices in the Loewner order,
i.e., $A \preceq B$ if and only if $B-A$ is positive semi-definite.

\begin{lemma}\label{lem:lev_contract}
   For any $w,w' \in \R_{+}^m$, it holds that
   $d(\tau(w), \tau(w')) \leq d(w,w')$.
\end{lemma}
\begin{proof}
   Denote $W = \diag(w), W' = \diag(w')$, and $\alpha \seteq \exp(d(w,w'))$.
   Then $\alpha^{-1} W \preceq W' \preceq \alpha W$, therefore
   $\alpha^{-1} A^\top  W A \preceq A^\top  W' A \preceq \alpha A^\top  W A$, and
   by monotonicity of the matrix inverse in the Loewner order,
   $\alpha^{-1} (A^\top  W A)^{-1} \preceq (A^\top  W' A)^{-1} \preceq \alpha (A^\top  W A)^{-1}$.
   This implies $d(\tau(w),\tau(w')) \leq \log \alpha$, completing the proof.
\end{proof}

\begin{proof}[Proof of \pref{thm:weight-schemes}]
Consider the map $\psi: \R_+^m \to \R_+^m$ whose $i$-th coordinate is defined as
\[
   \psi_i(x) \seteq \frac{f_i(\sqrt{x_i})}{x_i}\,.
\]

Our assumptions on lower and upper-homogeneity give, for all $y_i \geq x_i$,
\[
   c \left(\frac{y_i}{x_i}\right)^{\lo/2-1} \leq  \frac{f_i(\sqrt{y_i})/y_i}{f_i(\sqrt{x_i})/x_i} \leq C \left(\frac{y_i}{x_i}\right)^{\upp/2-1}\,,
\]
yielding, for $C_1 \seteq \max\{C, 1/c\}$,
\begin{equation}\label{eq:psi-contract}
   d(\psi(x),\psi(y)) \leq \max \left(\left|\frac{\lo}{2}-1\right|, \left|\frac{\upp}{2}-1\right|\right) d(x,y) + \log(C_1)\,.
\end{equation}

Fix $s > 0$ and consider the mapping $\f : \R_+^m \to \R_+^m$ defined in \eqref{eq:phi-iter}.
Then for $\upp < 4$ and $\delta \seteq \max \left(\left|\frac{\lo}{2}-1\right|, \left|\frac{\upp}{2}-1\right|\right) < 1$,
\eqref{eq:psi-contract} in conjunction with \pref{lem:lev_contract}, shows that
\begin{align}
   d(\f_s(w), \f_s(w')) < \delta\,d(w,w') + \log(C_1)\,. \label{eq:onestep}
\end{align}
Applying this bound inductively, for any weight $w \in \R_{+}^m$ and
$k \geq 1$, we have
\begin{equation}\label{eq:iter-contract0}
   d\left(\f_s^{k}(w), \f_s^{k+1}(w)\right) \leq \frac{\delta^k d(\f_s(w),w) + \log C_1}{1-\delta}\,,
\end{equation}
Now define
\[
   w^{(0)} \seteq \f^{k}_{1}(1,\ldots,1)\,,
\]
where $k \geq 1$ is chosen large enough so that $d(w^{(0)}, \f_{1}(w^{(0)})) \leq \frac{2 \log C_1}{1-\delta}$.
From \pref{fact:tau-char}, one sees that $w^{(0)}$ is an $\alpha$-approximate weight at scale $1$
for $\alpha = C_1^{2/(1-\delta)}$.

Define inductively, for $j =1,2,\ldots$,
\begin{align*}
   w^{(j)} &\seteq \f_{2^j}(w^{(j-1)}) \\
   w^{(-j)} &\seteq \f_{2^{-j}}(w^{(1-j)})\,.
\end{align*}

Note that
\begin{align*}
   d(\f_{2^j}(w^{(j)}), w^{(j)}) &= d(\f_{2^j}^2(w^{(j-1)}), \f_{2^j}(w^{(j-1)})) \\
                                 &\leq \delta d(\f_{2^j}(w^{(j-1)}), w^{(j-1)}) + \log(C_1) \\
                                 &\leq \delta d(\f_{2^{j-1}}(w^{(j-1)}), w^{(j-1)}) + \delta \log(2) + \log(C_1)\,,
\end{align*}
where the last inequality uses $\f_{2s}(w) = 2 \f_s(w)$ for all $w \in \R_+^m$.

Therefore, by induction, 
$d(\f_{2^j}(w^{(j)}), w^{(j)}) \leq \frac{2 \log(C_1) + \log 2}{1-\delta}$ for all $j > 0$.
To see that the family of weights $\{ w^{(j)} : j \in \Z \}$ forms a weight scheme,
note that
\[
   d(w^{(j)}, w^{(j-1)}) 
    = d(\f_{2^j}(w^{(j-1)}), w^{(j-1)}) \leq d(\f_{2^{j}}(w^{(j-1)}), w^{(j-1)}) + \log 2\,,
 \]
thus $\{w^{(j)} : j \in \Z \}$ is an $\alpha$-approximate weight scheme for $\alpha = \frac{2 \log(2C_1)}{1-\delta}$, completing the proof.
\end{proof}

\subsubsection{Efficient implementation of the iteration}
\label{sec:weightalgo}

In this section we give an efficient algorithm for implementing the iteration \eqref{eq:phi-iter}.
The primary difficult is that we need to exhibit convergence even when the iterates are only computed approximately.
For convenience, we use the notation $x \approx_{\alpha} y$ to denote that $\alpha^{-1} \leq |x/y| \leq \alpha$.

\newcommand{\winit}{w^{\circ}}
\newcommand{\jmin}{j_{\min}}
\newcommand{\jmax}{j_{\max}}

\begin{theorem}[Algorithm for weight construction]
\label{thm:weightalgo}
\pref{alg:findweights} takes as input functions $f_1, \dots, f_m: \R \to \Rplus$ are lower $\lo$-homogeneous with constant $c$ \ref{item:lower-growth} and upper $\upp$-homogeneous with constant $C$ \ref{item:upper-growth}, for some $\lo > 0$ and $\upp < 4$, vectors $a_1, \dots, a_m \in \R^n$, integers $j_{\min} < j_{\max}$, and $\winit \in \Rplus^m$ that satisfies
\begin{equation}
\label{eq:init_weight}
d(\f_{2^{\jmax}}(\winit), \winit) \leq \beta.
\end{equation}
For some $\alpha = \alpha(\lo,\upp,c,C)$, the algorithm returns an $\alpha$-approximate weight scheme $\{ w_i^{(j)} : j \in \cJ \}$ with $\cJ = \Z \cap [\jmin,\jmax]$, and succeeds with high probability in time \[ \O_{\lo,\upp,c,C}\!\left((\nnz(a_1,\ldots,a_m) + n^\omega + m \tEval) \left(|\cJ| + \log \max\{\beta, 1\} \right)\right)\,. \]
\end{theorem}

The algorithm proceeds along an iterative procedure akin to \eqref{eq:phi-iter}.
At each step, the weights are updated by computing approximate leverage scores of a matrix $A^{\top} W A$.

\begin{theorem}[Leverage score approximation, \cite{SS11,LMP13,CLMMPS15}]
\label{thm:levapprox}\label{thm:lever_over}
There is an algorithm $\LevApprox(A, \baltW, \eps)$ that takes a matrix $A \in \R^{m \times n}$, a non-negative diagonal matrix $\baltW$, and $\eps > 0$,
and produces $(1+\eps)$-approximations $\widetilde{\sigma}_i \approx_{1+\eps} w_i a_i^\top (A^\top \baltW A)^{-1} a_i$ for all $i \in [m]$, in $\O(\eps^{-2}(\mathsf{nnz}(A) + n^\omega)\log(1/\delta))$ time, with probability at least $1-\delta$.
\end{theorem}

\begin{algorithm2e}[!ht]
\caption{$\Iterate(\{f_1, \dots, f_m\}, \{a_1, \dots, a_m\}, w, s, \eps)$}
\label{alg:iterate}
\SetKwInOut{Input}{input}
\Input{Functions $f_1, \dots, f_m: \R \to \Rplus$, vectors $a_1, \dots, a_m \in \R^n$, weights $w \in \R_{+}^m$, a scale $s \in \R_{+}$, and $\eps \in \R_{+}$.}
$\widetilde{\sigma} \assign \LevApprox(A,\baltW,\eps)$. \\
$\bar{w}_i \assign \frac{1}{s} \frac{f_i(\sqrt{\widetilde{\sigma}_i/w_i})}{\widetilde{\sigma}_i/w_i}$ for $i=1,\ldots,m$. \\
\textbf{return} $\bar{w}$.
\end{algorithm2e}

\begin{algorithm2e}[!ht]
\caption{$\FindWeights(\{f_1, \dots, f_m\}, \{a_1, \dots, a_m\}, \jmin, \jmax, \winit, \beta, \lo, \upp, c, C)$}
\label{alg:findweights}
\SetKwInOut{Input}{input}
\Input{Functions $f_1, \dots, f_m: \R \to \Rplus$, vectors $a_1, \dots, a_m \in \R^n$, integers $\jmin < \jmax$, an initial weight $\winit \in \Rplus^m$ satisfying \eqref{eq:init_weight} for $\beta > 1$, parameters $\lo, \upp, c, C$
}
$\eps \seteq 0.1$\\
$T \assign \left\lceil \frac{\log\left( \frac{1 + \beta}{ \log \max(2C, 2/c)} \right)}{\log \min(| \frac{2}{\lo-2}| , |\frac{2}{\upp-2}|)} \right\rceil$ \label{line:defqt} \\
$w^{(\jmax)}_0 \assign \winit$ \\
\lFor*{$i = 0, 1, \dots,T$}{
   $w^{(\jmax)}_{i+1} \assign \Iterate(\{f_1, \dots, f_m\}, \{a_1, \dots, a_m\}, w_i^{(\jmax)}, 2^{\jmax}, \epsilon)$. \label{line:iterate1}
} \\
$w^{(0)} \assign w^{(0)}_{T}$\\
\lFor*{$i = \jmax, \jmax-1, \ldots, \jmin$}{
    $w^{(i)} \assign \Iterate(\{f_1, \dots, f_m\}, \{a_1, \dots, a_m\}, w^{(i+1)}, 2^i, \epsilon)$. \label{line:iterate2}
} \\
\textbf{return} $\{ w^{(j)} : \jmin \leq j \leq \jmax \}$.
\end{algorithm2e}

\pref{alg:iterate} called with $\e = 0$ is able to directly implement the iteration \eqref{eq:phi-iter}.
We now show that $\Iterate$ remains approximately contracting for $\eps > 0$.

Let $\f_1,\ldots,\f_m$ and $a_1,\ldots,a_m \in \R^n$ be given as in \pref{thm:weightalgo}, and
define the function $\cI_{s,\eps} : \Rplus^m \to \Rplus^m$ by
\begin{align} \label{eq:ci}
\cI_{s,\eps}(w) \seteq \frac{1}{s}\frac{f_i(\sqrt{\tilde{\sigma}_i/w_i})}{\tilde{\sigma}_i/w_i}\,, \text{ where } \tilde{\sigma}_i \approx_{1+\e} w_i a_i^\top (A^\top WA)^{-1} a_i \text{ for all } i = 1, \dots, m\,,
\end{align}
i.e., $\tilde{\sigma}_i$ are arbitrary approximate leverage scores.
\begin{lemma}
\label{lem:contract}
For any $w,w' \in \Rplus^m$ and $0 < \e < 1/3$, it holds that
\[
d(\cI_{s,\eps}(w),\cI_{s,\eps}(w')) \leq \delta d(w,w') + \log(2 C_1),
\]
with $\delta \seteq \max \left(\left|\frac{\lo}{2}-1\right|, \left|\frac{\upp}{2}-1\right|\right)$, and $C_1 \seteq \max(C, 1/c)$.
\end{lemma}

\begin{proof}
Define $\baltW \seteq \diag(w_1,\ldots,w_m)$ and $\baltW' \seteq \diag(w'_1,\ldots,w'_m)$.
Suppose that $\rho,\rho' \in \R^m$ are such that
$\rho_i \approx_{1+\eps} a_i^\top (A^\top \baltW A)^{-1} a_i$ and $\rho'_i \approx_{1+\eps} a_i^\top (A^\top \baltW' A)^{-1} a_i$
for each $i=1,\ldots,m$, and
\[
   \cI_{s,\eps}(w)_i = \frac{1}{s} \frac{f_i(\sqrt{\rho_i})}{\sqrt{\rho_i}}  \quad \text{and} \quad \cI_{s,\eps}(u)_i = \frac{1}{s} \frac{f_i(\sqrt{\rho'_i})}{\sqrt{\rho'_i}}\,.
\]
From~\eqref{eq:psi-contract}, we have
\[
   d(\cI_{s,\eps}(w), \cI_{s,\eps}(w')) = d \left(\psi(\rho), \psi(\rho')\right) \leq \delta d(\rho, \rho') + \log(C_1).
\]

Because $\rho_i \approx_{1+\eps} \tau_i(w)$ and $\rho'_i \approx_{1+\eps} \tau_i(w')$,
\[
   d(\rho,\rho') = \max_{i \in [m]} \left| \log \frac{\rho_i}{\rho'_i} \right| \leq \max_{i \in [m]} \left| \log \frac{\tau_i(w)}{\tau_i(w')} \right| + \log\left(\frac{1+\eps}{1-\eps}\right) = d(\tau(w), \tau(w')) + \log\left(\frac{1+\eps}{1-\eps}\right).
\]
Combined with \pref{lem:lev_contract}, this gives
\[
d(\cI_{s,\eps}(w), \cI_{s,\eps}(w')) \leq \delta d(w,w')+ \log(C_1) + \log\left(\frac{1+\eps}{1-\eps}\right) \leq \delta d(w,w') + \log(2 C_1)\,.\qedhere
\]
\end{proof}

\newcommand{\si}[1]{s^{(#1)}}
\begin{proof}[Proof of \pref{thm:weightalgo}]
   Because we use \pref{alg:findweights} to find the desired weights, the claimed running time bound follows from \pref{thm:levapprox} for the choice $\delta = (m(|\cJ| + \log \max\{\beta, 1\}))^{-O(1)}$.
   Thus it suffices to argue that the output weights $\{ w^{(j)} : j \in \cJ \}$ form an approximate weight scheme.
   For this choice of $\delta$, taking $\cI_{\smax,0.1}(w) \seteq \Iterate(\{f_1, \dots, f_m\}, \{a_1, \dots, a_m\}, w, s, 0.1)$ satisfies, with high probability, \eqref{eq:ci} for all $T + |\cJ|$ calls to \pref{alg:iterate} from \pref{alg:findweights}.

   For ease of notation, let us denote $\smax \seteq 2^{\jmax}$.
   The analysis is identical to that in the proof of \pref{thm:weight-schemes} (recall \eqref{eq:iter-contract0}), except that
   \eqref{eq:onestep} is replaced by \pref{lem:contract}, and the initial weight bound is replaced by
   \begin{align*}
      d\!\left(w_0^{(\jmax)}, \cI_{\smax,0.1}\!\left(w_0^{(\jmax)}\right)\right) &\leq 
   d\!\left(\f_{\smax}\!\left(w_0^{(\jmax)}\right), w_0^{(\jmax)}\right)+
   d\!\left(\cI_{\smax,0.1}\!\left(w_0^{(\jmax)}\right), \cI_{\smax,0}\!\left(w_0^{(\jmax)}\right)\right) \\
                                                                                 &\leq \log \beta + \log(2 C_1)\,,
\end{align*}
where the last inequality follows from \eqref{eq:init_weight} and \pref{lem:lev_contract}.
\end{proof}

\subsubsection{Constructing initial weights}
\label{sec:init}

\newcommand{\sz}{\hat{z}}

A mild problem arises when applying \pref{thm:weightalgo}, which is that it
may be computationally non-trivial to locate an initial weight $\winit \in \Rplus^m$ satisfying \eqref{eq:init_weight} with
$\beta$ sufficiently small.
In this section, we show how to efficiently compute small
perturbations
$\hat{f}_1,\ldots,\hat{f}_m : \R^n \to \Rplus$ of the functions $f_1,\ldots,f_m : \R^n \to \Rplus$ along with good initial weights $\winit$ for $\{\hat{f}_i\}$.

Fix $0 < \smin < \smax$ and vectors $a_1,\ldots,a_m \in \R^n$. Consider $f_1,\ldots,f_m : \R^n \to \Rplus$ such that $f_1^{1/2},\ldots,f_m^{1/2}$ are $L$-auto-Lipschitz \ref{item:lipschitz} and 
$\lo$-lower homogeneous with constant $c$ \ref{item:lower-growth}. In that case, each $f_i$ is continuous, and therefore for any $0 < \gamma \leq 1$ there exist numbers $\sz_1,\ldots,\sz_m > 0$ such that
\begin{equation}\label{eq:zidef}
   \gamma \smax \leq f_i(\sz_i) \leq \smax\,,\quad i=1,\ldots,m\,.
\end{equation}

Define the matrix $U \seteq \sum_{i=1}^m \sz_i^{-2} a_i a_i^{\top}$, and let $\tilde{\tau}_i \approx_2 \langle a_i, U^{-1} a_i\rangle$ for $i=1,\ldots,m$.
Note that these values can computed using a single call to $\LevApprox(A, (\sz_1^{-2},\ldots, \sz_m^{-2}), 1/2)$.
Define $w_i \seteq \delta/\tilde{\tau}_i$ for some $\delta > 0$ and $i=1,\ldots,m$, and
finally define
\begin{equation}\label{eq:def-hf}
   \hat{f}_i(z) \seteq f_i(z) + \smax w_i z^2\,,\quad i=1,\ldots,m\,,
\end{equation}
The following fact is straightforward.
\begin{fact}\label{fact:hf-props}
   If $f_1^{1/2},\ldots,f_m^{1/2}$ are $L$-auto-Lipschitz and $\lo$-lower homogeneous with constant $c$, then $\hat{f}_1^{1/2},\ldots,\hat{f}_m^{1/2}$
   are $\max\{1, L\}$-auto-Lipschitz and $\lo$-lower homogeneous with constant $c$.
\end{fact}
\begin{proof}
Lower homogeneity is clear. Note that for numbers $a,b,c,d \in \R$, it holds that
\[
   |(a^2+b^2)^{1/2} - (c^2+d^2)^{1/2}| = \left|\|(a,b)\|_2-\|(c,d)\|_2\right| \le \|(a,b)-(c,d)\|_2 = ((a-c)^2 + (b-d)^2)^{1/2}\,.
\]
Employ this to write
\begin{align*}
|\hat{f}_i(z)^{1/2} - \hat{f}_i(z')^{1/2}| &\le ((f_i(z)^{1/2} - f_i(z')^{1/2})^2 + \smax w_i (z-z')^2)^{1/2} \\
&\le (L^2 f_i(z-z') + \smax w_i (z-z')^2)^{1/2} \le \max\{1, L\} \hat{f}_i(z-z')^{1/2}.
\end{align*}
Thus each $\hat{f}_i^{1/2}$ is $\max\{1, L\}$-auto-Lipschitz for $i = 1, 2, \ldots, m$.
\end{proof}

\begin{theorem}\label{thm:perturb}
\label{thm:init}
   Let $F(x) \seteq \sum_{i=1}^m f_i(\langle a_i,x\rangle)$ and $\hat{F}(x) \seteq \sum_{i=1}^m \hat{f}_i(\langle a_i,x\rangle)$
   for  $\delta > 0$ and any $\sz_i > 0 $ satisfying \eqref{eq:zidef}. For all $x \in \R^n$,
   \begin{equation}\label{eq:hat-approx}
      F(x) \leq \smax \implies
      0 \leq \hat{F}(x) - F(x) \leq 2 \delta m^2 \smax (L/(\gamma c))^{2/\theta}\,.
   \end{equation}
   Moreover, $w$ is an $O((L/c)^2 m/\delta)$-approximate weight at scale $\smax$ for $\{\hat{f}_i\}$ and $\{a_i\}$.
\end{theorem}

\begin{proof}
Define $M_w \seteq A^{\top} W A$, where $W \seteq \diag(w_1,\ldots,w_m)$.
We first claim that
\begin{equation}\label{eq:MwUz}
   \frac{\delta}{2} U \preceq M_w \preceq 2 \delta m U\,.
\end{equation}
To prove the upper bound in \eqref{eq:MwUz}, note that
\[
   a_i a_i^{\top} \preceq \langle a_i, U^{-1} a_i\rangle U \preceq 2 \tilde{\tau}_i U\,,
\]
Summing over $i=1,\ldots,m$ indeed gives $M_w = \delta \sum_{i=1}^m \frac{1}{\tilde{\tau}_i} a_i a_i^{\top} \preceq 2\delta m U$.
 For the lower bound, note that $U \succeq \sz_i^{-2} a_i a_i^{\top}$ and hence $U^{-1} \preceq (\sz_i^2/\|a_i\|_2^2) a_i a_i^{\top}$. Therefore,
  $\tilde{\tau}_i \leq 2 \langle a_i, U^{-1} a_i\rangle \leq 2 \sz_i^2$ which implies that $w_i \geq (\delta/2) \sz_i^{-2}$ and indeed gives the lower bound in \eqref{eq:MwUz}.

Next we prove \eqref{eq:hat-approx}.  The lower bound of \eqref{eq:hat-approx} is trivial. 
To prove the upper bound, we first show that if $f_i(\langle a_i,x\rangle) \leq \smax$, then it holds that $\langle a_i,x\rangle^2 \leq (L/(\gamma c))^{2/\theta} \sz_i^2$ for $i=1,\ldots,m$. Indeed, if $|\l a_i, x\r| = \lambda \sz_i$ for $\lambda \ge 1$, then lower homogeneity, symmetry, and the definition of $\sz_i$ in \eqref{eq:zidef},
\[ \smax \ge f(\l a_i, x\r) \ge \frac{1}{L}f(|\l a_i, x\r|) \ge \frac{c\lambda^\lo}{L} f(\sz_i) \geq \frac{\gamma c\lambda^\lo}{L} \smax \,. \]
Thus $\lambda \le (L/(\gamma c))^{1/\lo}$, as desired.
Now note that
\begin{align*}
   \hat{F}(x) - F(x) &= \smax \langle x, M_w x\rangle \stackrel{\mathclap{\eqref{eq:MwUz}}}{\leq} 2 \delta m \langle x, U x\rangle 
                                                                                         = 2 \delta m \smax \sum_{i=1}^m \sz_i^{-2} \langle a_i,x\rangle^2 
                                                                                         \leq 2 \delta m^2 \smax (L/(\gamma c))^{2/\theta}\,.
\end{align*}

\newcommand{\qfm}{(a_i^{\top} M_w^{-1} a_i)^{1/2}}

Finally, we establish that $w$ are suitable approximate weights. Observe that 
\[
   d(\f_{\smax}(w), w) = \max_{i \in [m]} \log \left| \frac{1}{\smax} \frac{\hat{f}_i(\qfm)}{w_i \cdot a_i^\top M_w^{-1} a_i} \right|.
\]
Then by the definition \eqref{eq:def-hf}, we have
\[ \frac{1}{\smax} \frac{\hat{f}_i(\qfm)}{w_i \cdot a_i^\top M_w^{-1} a_i} = \frac{1}{\smax} \left(\frac{f_i(\qfm)}{w_i \cdot a_i^\top M_w^{-1} a_i} + \frac{\smax\delta/\tilde{\tau}_i}{w_i}\right) = 1 + \frac{1}{\smax} \frac{f_i(\qfm)}{w_i \cdot a_i^\top M_w^{-1} a_i}\,. \]
Now \eqref{eq:MwUz} and $\tilde{\tau}_i \leq 2 \sz_i^2$ together give
\[
   a_i^\top M_w^{-1} a_i \leq \frac{2}{\delta} a_i^\top U^{-1} a_i \le \frac{4}{\delta} \tilde{\tau}_i \leq \frac{8}{\delta} \sz_i^2\,,\quad i=1,\ldots,m\,.
\]
If $\qfm \le \sz_i$, then by monotonicity, $f_i(\qfm) \le (1/c)^2 f_i(\sz_i) = (1/c)^2 \smax.$
Otherwise, the auto-Lipschitz and lower homogeneous properties for $f_i^{1/2}$ give
$f_i(\lambda x) \leq (2L/c)^2\lambda^2 f_i(x)$ for every $\lambda \geq 1$ by \pref{lem:relate}, and
this implies
\[ 
   f_i(\qfm) \le 4(L/c)^2 \Bigg(\frac{\qfm}{\sz_i} \Bigg)^2 f_i(\sz_i) \lesssim (L/c)^2 \frac{\smax}{\delta}\,.
 \]
 In either case, $f_i(\qfm) \lesssim (L/c)^2 \smax/\delta$.
 For the denominator, use \eqref{eq:MwUz} to write
 \[ w_i \cdot a_i^\top M_w^{-1} a_i \ge \frac{1}{2\delta m} w_i a_i^\top U^{-1} a_i \geq \frac{1}{4m} \tilde{\tau}_i \tilde{\tau}_i^{-1} = \frac{1}{4m}\,.\]
Combining everything and using that $f_i$ is non-negative, we arrive at
\[
   d(\f_{\smax}(w),w) \leq \log\left(1 + \frac{1}{\smax} \frac{O((L/c)^2\delta^{-1} \smax)}{1/(4m)}\right) \leq \log(1 + O((L/c)^2 \delta^{-1} m))\,.\qedhere
\]
\end{proof}

\subsection{A variational approach to approximate weights}
\label{sec:variational}
\label{sec:optmap}

In this section we show that even if each $f_i$ is upper $\upp$-homogeneous for $\upp \ge 4$, approximate weights still exist.
Our sparsification analysis relies on the existence of weight schemes, where there is a relationship between
weights at different scales. But the following existence proof is instructive.

\begin{theorem}\label{thm:variational}
   Suppose $f_1,\ldots,f_m : \R_+ \to \R_+$ are lower $\lo$-homogeneous with constant $c$ and upper $\upp$-homogeneous with constant $C$ with $\upp > \lo > 0$.
   Then there is a constant $\alpha = \alpha(\lo,c,\upp,C)$ such that for every choice of vectors $a_1,\ldots,a_m \in \R^n$ and $s > 0$, there is an $\alpha$-approximate weight at scale $s$.
\end{theorem}
The idea behind the proof is to set up a variational problem whose critical points produce approximate weights at a given scale.
As observed in \cite{SZ01}, this analysis technique does not require convexity.

\begin{lemma}\label{lem:variational-1}
   Suppose $g_1,\ldots,g_m : \R_+ \to \R_+$ are monotone increasing, continuously differentiable, and satisfy $g_1(0)=\cdots=g_m(0)=0$.
   Then for every $\beta > 0$, there are weights $\{w_i \geq 0 : i =1,\ldots,m\}$ such that
   \begin{align}
      w_i &= \gamma \frac{g'_i(\|M_w^{-1/2} a_i\|_2)}{\|M_w^{-1/2} a_i\|_2}\quad i=1,\ldots,m\,, \label{eq:var-weights}\\
      \gamma &= n \left(\sum_{i=1}^m g'_i(\|M_w^{-1/2} a_i\|_2)\|M_w^{-1/2} a_i\|_2\right)^{-1}\,, \label{eq:var-gamma} \\
      \beta &= \sum_{i=1}^m g_i(\|M_w^{-1/2} a_i\|_2)\,. \nonumber
   \end{align}
\end{lemma}

\newcommand{\sfd}{\mathsf{d}}

\begin{proof}
For a linear operator $U : \R^n \to \R^n$, define
\[
   G(U) \seteq \sum_{i=1}^m g_i(\|U a_i\|_2)\,,
\]
and consider the optimization
\begin{equation}\label{eq:lambda-opt}
   \textrm{maximize} \ \left\{ \det(U) : G(U) \leq \beta \right\}\,.
\end{equation}

Since $G(0)=0$ and each $g_i$ is monotone increasing, it holds that $G(cI)=\beta$ for some $c > 0$.
Therefore for any maximizer $U^*$, it holds that $\det(U^*) > 0$, i.e., $U^*$ is invertible, and $G(U^*) = \beta$.

For $U$ invertible, we have
\begin{equation}\label{eq:det-change}
   \nabla \det(U) = \det(U) U^{-\top}.
\end{equation}
Let us also calculate
\[
   \sfd G(U) = \sum_{i=1}^m g'_i(\|U a_i\|_2)\, \frac{\sfd \|U a_i\|_2^2}{2 \|U a_i\|_2}\,,
\]
and use $\|U a_i\|_2^2 = \tr(U^{\top} U a_i a_i^{\top})$ to write
\[
   \frac12 \sfd \|U a_i\|_2^2 = \tr\left((\sfd U)^{\top} U a_i a_i^{\top}\right),
\]
so that
\[
   \nabla_U G(U) = U A^{\top} D_U A\,,
\]
where
$D_U$ is the $m \times m$ diagonal matrix with $(D_U)_{ii} = \frac{g_i'(\|U a_i\|_2)}{\|U a_i\|_2}$ and
$A \in \R^{m \times n}$ is the matrix with rows $a_1,\ldots,a_m$.

Combined with \eqref{eq:det-change}, we see that if $U$ is an optimal solution to \eqref{eq:lambda-opt}, then for some Lagrange multiplier $\gamma > 0$, we have
\[
   (U^\top U)^{-1} = \gamma A^\top D_U A\,.
\]

Take $V \seteq (U^\top U)^{-1}$ so that $V = \sum_{i=1}^m w_i a_i a_i^\top$ with $w_i \seteq \gamma \frac{g'_i(\|V^{-1/2} a_i\|_2)}{\|V^{-1/2} a_i\|_2}$.
To compute the value of $\gamma$, calculate
\[
   n = \tr(V V^{-1}) = \gamma \sum_{i=1}^m \frac{g'_i(\|V^{-1/2} a_i\|_2)}{\|V^{-1/2} a_i\|_2} \tr(V^{-1} a_i a_i^\top) = \gamma \sum_{i=1}^m g'_i(\|V^{-1/2} a_i\|_2) \|V^{-1/2} a_i\|_2\,.\qedhere
\]
\end{proof}

The preceding lemma assumes that the functions $g_i$ are monotone increasing. However, the functions $f_i$ only satisfy lower homogeneity, which is a weaker condition.
To prove \pref{thm:variational}, one can take monotone approximations of functions $f_i$ by averaging over intervals.
\begin{lemma}\label{lem:g-reduce}
   Suppose $f$ is lower $\lo$-homogeneous with constant $c$ and upper $\upp$-homogeneous with constant $C$ for $u \ge 1$.
   Define $K \seteq \max(e, (2/c)^{1/\lo})$ and $g : \R_+ \to \R_+$ by
   \[
      g(x) \seteq \int_x^{Kx} \frac{f(t)}{t}\,dt\,.
   \]
   Then $g$ is continuously differentiable, monotone increasing, and satisfies $g(0)=0$. Moreover, for all $x > 0$,
   \begin{align*}
      \frac{f(x)}{x} &\leq g'(x) \leq C K^\upp \,\frac{f(x)}{x}\,,\quad \forall x \geq 0 \\
      c\, f(x) &\leq g(x) \leq C K^\upp \,f(x)\,.
   \end{align*}
\end{lemma}

\begin{proof}
   Note that $g'(x) = \frac{f(Kx)-f(x)}{x}$, and $f(Kx) \ge cK^\lo f(x) \ge 2f(x)$ by lower $\lo$-homogeneity. Thus $g$ is monotone increasing. Moreover, we have
\begin{align*}
g(x) = \int_x^{Kx} \frac{f(t)}{t} dt \ge c f(x) \int_x^{Kx} \frac{1}{t} dt \ge c \cdot f(x),
\end{align*}
and for $u \ge 1$,
\begin{align*}
g(x) = \int_x^{Kx} \frac{f(t)}{t} dt \le C f(x) \int_x^{Kx} \frac{(t/x)^\upp}{t} dt \le C K^\upp f(x).
\end{align*}
\end{proof}

\begin{proof}[Proof of \pref{thm:variational}]
   Throughout the proof, we use $\asymp$ in place of $\asymp_{\lo,c,\upp,C}$.
   Let $K$ be as in \pref{lem:g-reduce} and define $g_i(x) \seteq \int_x^{Kx} \frac{f_i(t)}{t}\,dt$.
   By \pref{lem:g-reduce}, we have $g_i'(x) \asymp f_i(x)/x$ and $g(x) \asymp f_i(x)$ for $i=1,\ldots,m$.

   Let $\{w_i\}$ and $M_w$ be as in \pref{lem:variational-1} when applied to $g_1,\ldots,g_m$ with $\beta = n/\lambda$.
   Then,
   \[
      \sum_{i=1}^m g_i'(\|M_w^{-1/2} a_i\|_2) \|M_w^{-1/2} a_i\|_2 \asymp \sum_{i=1}^m f_i(\|M_w^{-1/2} a_i\|_2) \asymp \sum_{i=1}^m g_i(\|M_w^{-1/2} a_i\|_2) = \beta = \frac{n}{\lambda}\,.
   \]
   Therefore $\gamma \asymp \lambda$ (recall \eqref{eq:var-gamma}), and thus \eqref{eq:var-weights} gives the desired result.
\end{proof}

\section{Covering number bounds}
\label{sec:covering}

Consider loss functions $f_1,\ldots,f_m : \R \to \R_+$, vectors $a_1,\ldots,a_m \in \R^n$,
and a contiguous interval $\cJ \subseteq \Z$.
Define $h_i(x) \seteq f_i(x)^{1/2}$ for $i=1,\ldots,m$, $\mm \seteq \max(\cJ)+1$, and
\[
   F(x) \seteq \sum_{i=1}^m f_i(\langle a_i,x\rangle) = \sum_{i=1}^m h_i(\langle a_i,x\rangle)^2\,.
\]
For $s > 0$, denote $B_F(s) \seteq \{ x \in \R^n : F(x) \leq s \}$.

Suppose that 
$\{ w^{(j)} : j \in \cJ \}$ is an $\alpha$-approximate weight scheme (\pref{def:weight-scheme})
for the families $\{f_i\},\{a_i\}$, and recall that
\[
   M_{w^{(j)}} = \sum_{i=1}^m w_i^{(j)} a_i a_i^{\top}.
\]
Define $U_j \seteq M_{w^{(j)}}^{1/2}$ for $j \in \cJ$ so that
\begin{equation}\label{eq:Uexpr}
   \|U_j x\|_2^2 = \sum_{i=1}^m w_i^{(j)} \langle a_i,x\rangle^2\,,\quad x \in \R^n\,.
\end{equation}

Analogously to \pref{sec:overview}, define the sets
\begin{align*}
   \sfK_j &\seteq \left\{ x \in \R^n : h_i(\langle a_i,x\rangle)^2 \leq 2^j w_i^{(j)} \|U_j^{-1} a_i\|^2_2\,, i \in [m] \right\}\,,\quad j \in \cJ \\
   \sfK_{\mm} &\seteq B_F(2^{\mm})\,.
\end{align*}
Our primary technical goal is an estimate on the covering numbers $\cov(B_F(2^{\mm}), \sfK_j)$
when the functions $h_1,\ldots,h_m$ are sufficiently nice.

Recall our assumption in \pref{thm:main} that $h_i = f_i^{1/2}$ is auto-Lipschitz (property \ref{item:lipschitz}) and lower $\lo$-homogeneous (property \ref{item:lower-growth}).
The following properties \ref{ass:monotone}--\ref{ass:subadditive} follow
from these two assumptions, but we label them specifically as they will be employed
numerous times in our arguments.

\begin{assumption}\label{ass:h}
   Consider $h_1,\ldots,h_m : \R \to \R_+$ for some $L, C \geq 1$ and $\lo > 0, c < 1$:
   \begin{enumerate}[leftmargin=*,labelindent=0.5em,label=\textrm{(H\arabic*)}]
   \item \label{ass:monotone} $h_i(\lambda x) \geq c h_i(\pm x)$ for all $x \in \R$ and $\lambda \geq 1$.
   \item \label{ass:lipschitz}  $|h_i(x)-h_i(y)| \leq L h_i(x-y)$ for all $x,y \in \R$.
   \item \label{ass:upper-growth} $h_i(\lambda x) \leq C \lambda h_i(\pm x)$ for all $x \in \R$ and $\lambda \geq 1$.
   \item \label{ass:lower-growth} $h_i(\lambda x) \geq c \lambda^{\lo} h_i(\pm x)$ for all $x \in \R$ and $\lambda \geq 1$.
   \item \label{ass:subadditive} $h_i(x \pm y) \leq L \left(h_i(x)+h_i(y)\right)$ for all $x,y \in \R$.
   \end{enumerate}
\end{assumption}
Compared to the assumptions discussed previously, we have added \ref{ass:subadditive}, as it is used many times.
Note that \ref{ass:subadditive} follows from \ref{ass:lipschitz}:
$h(x+y) \leq h(y) + L h(x)$ and $h(x-y) \leq h(-y) + L h(x) \leq L(h(x)+h(y))$,
and \ref{ass:monotone} is an immediate consequence of \ref{ass:lower-growth}.

Let us now state our primary covering estimate, which is established over the next two sections.

\begin{theorem}\label{thm:main-cov}
   There is a number $\hat{C} = \hat{C}(L,C,\theta,c,\alpha)$ such that for any $\{h_1,\ldots,h_m\}$ satisfying \pref{ass:h},
   it holds that
   \[
      \log \cov\!\left(B_F(2^{\mm}), \sfK_{j}\right) \leq \hat{C} 2^{\mm -j} \log m\,,\quad \forall j \in \cJ\,.
   \]
\end{theorem}

\subsection{Iterative covering}

In order to prove \pref{thm:main-cov}, we will relate $\sfK_j$ to $\ell_2$ and $\ell_{\infty}$ balls
with respect to an appropriate inner product structure.
To this end, let us define the norms, for $j \in J$,
\begin{align*}
   \sfN_j^{\infty}(x) &\seteq \max_{i \in [m]} \frac{|\langle U_j^{-1} a_i, U_j x\rangle|}{\|U_j^{-1} a_i\|_2}\,, \\
   \sfN_j^{2}(x) &\seteq \|U_j x\|_2\,,
\end{align*}
and let $\sfB_j^{\infty} \seteq \{x \in \R^n : \sfN_j^{\infty}(x) \leq 1 \}$ and $\sfB_j^2 \seteq \{ x \in \R^n : \sfN_j^2(x) \leq 1 \}$
denote the corresponding unit balls.
We observe the following consequence of \pref{lem:dualsudakov}.
\begin{corollary}\label{cor:ds}
   For every $j \in \cJ$ and $\eta  > 0$, it holds that
   \[
      \sqrt{\log \cov(\sfB_j^2, \eta \sfB_j^{\infty})} \lesssim \frac{\sqrt{\log m}}{\eta}\,.
   \]
\end{corollary}

\begin{proof}
   By scaling, we may assume that $\eta = 1$. Then \pref{lem:dualsudakov} gives
   \[
      \sqrt{\log \cov(\sfB_j^2, \sfB_j^{\infty})}
      \lesssim \E \max_{i \in [m]} \left|\llangle \frac{U_j^{-1} a_i}{\|U_j^{-1} a_i\|_2}, \bm{g}\rrangle\right|
      \lesssim \sqrt{\log m}\,,
   \]
   using the fact that if $g_1,\ldots,g_k$ are random variables that each have law $N(0,1)$, then
   $\E \max_{i \in [k]} g_i \lesssim \sqrt{\log k}$ since $\Pr[|g_i| > t] \leq 2 e^{-t^2/2}$.
\end{proof}

Thus our goal will be a pair of containment results for translates of the sets $\sfK_j$.
These are proved in the next section.

\begin{lemma}[$\ell_{\infty}$ control]
   \label{lem:contain1}
   For $j \in \cJ$ and  $c_0 \seteq (c^4/(4 \alpha))^{1/2\lo}$ it holds that
   \begin{equation}\label{eq:contain1}
      c_0 \sfB_{j}^{\infty} \subseteq \sfK_j \subseteq \frac{1}{c_0} \sfB_j^{\infty}\,.
   \end{equation}
\end{lemma}

\begin{lemma}[$\ell_2$ control]
   \label{lem:contain2}
   For any $j \in \cJ$ and $z \in B_F(2^{\mm})$, it holds that
   \[ B_F(2^{\mm}) \cap (z + \sfK_{j+1}) \subseteq z + (4L^2 C_0 2^{\mm-j})^{1/2} \sfB_{j}^2\,, \]
where
   $C_0 \seteq \max\{(2\alpha C/c)^{1/\lo}, (\alpha C/c)^2 (\alpha/c^2)^{1/\lo}\}$.
\end{lemma}

Note that the preceding lemma relates $\sfK_{j+1}$ to $\sfB_j^2$. This 
is the one place we will employ the key property of weight schemes (\pref{def:weight-scheme}), which is the containment $\sfB_{j+1}^2 \subseteq \alpha \sfB_j^2$:
For $j,j+1 \in \cJ$ and $x\in \R^n$,
\eqref{eq:Uexpr} gives
\[
   \|U_j x\|_2^2 = \sum_{i=1}^m w^{(j)}_i \langle a_i,x\rangle^2 \leq \alpha \sum_{i=1}^m w^{(j+1)}_i \langle a_i,x\rangle^2 = \alpha \|U_{j+1} x\|_2^2\,.
\]

With these two results in hand, we prove \pref{thm:main-cov}.

\newcommand{\ccov}[1]{\cov\!\left(#1\right)}

\begin{proof}[Proof of \pref{thm:main-cov}]
   First, note that for any $j \in \cJ$ and $z \in B_F(2^{\mm})$, we have
   \begin{align}
      \log \ccov{B_F(2^{\mm}) \cap (z+\sfK_{j+1}), \sfK_j} \nonumber
      &\leq
      \log \ccov{(4 L^2 C_0 2^{\ell-j})^{1/2} \sfB_j^2, c_0 \sfB_j^{\infty}} \\
      &\leq C_1 2^{\ell-j} \log m\, \label{eq:main-cover}
   \end{align}
   where $C_1 \seteq 4L^2 C_0/c_0^2$, the first inequality follows from \pref{lem:contain1} and \pref{lem:contain2},
   and the second inequality is a consequence of \pref{cor:ds}.

   An application of \eqref{eq:main-cover}
   with $j=\mm-1$ and $z=0$ gives vectors $x_{1},\ldots,x_{T_1} \in B_F(2^{\mm})$ with $T_1 \leq 2 C_1 \log m$, and such that
   \[
      B_F(2^{\mm}) \subseteq \bigcup_{t=1}^{T_1} \left(B_F(2^{\mm}) \cap (x_{t} + \sfK_{\mm -1})\right).
   \]
   Now apply \eqref{eq:main-cover} with $j=\mm-2$ and $z=x_{1},\ldots,x_{T_1}$
   to find vectors $x_{(t_1,t_2)} \in B_F(2^{\mm})$ for $1 \leq t_1 \leq T_1, 1 \leq t_2 \leq T_2$ with $T_2 \leq 4 C_1 \log m$, and such that
   \[
      B_F(2^{\mm}) \cap (x_{t_1} + \sfK_{\mm -1}) \subseteq \bigcup_{t_2=1}^{T_{2}} \left(B_F(2^{\mm}) \cap (x_{(t_1,t_2)} + \sfK_{\mm -2})\right)\,,\quad t_1=1,\ldots,T_1
   \]
   Continuing inductively, we cover $B_F(2^{\mm})$ by $T_1 \cdot T_2\, \cdots\, T_r$ translates of $\sfK_{\mm -r}$, and
   $\log |T_r| \leq C_1 2^{r} \log m$.

   We conclude that for $j \in \cJ$,
   \[
      \log \cov(B_F(2^{\mm}), \sfK_j) \leq \log |T_1| + \cdots + \log |T_{\mm -j}| \ls \left(2 + 2^2 + \cdots + 2^{\mm -j}\right) C_1 \log m\,.\qedhere
   \]
\end{proof}

\subsection{Norm control}

Our goal now is to prove \pref{lem:contain1} and \pref{lem:contain2}.
We will frequently use the following consequence of the Cauchy-Schwarz inequality
\begin{equation}\label{eq:cs_ai}
   |\langle a_i,x\rangle| = |\langle U_j^{-1} a_i, U_j x\rangle| \leq \|U_j^{-1} a_i\|_2 \|U_j x\|_2\,,\quad j \in \cJ, i=1,\ldots,m, x\in \R^n\,.
\end{equation}
We also restate the guarantees of our weight scheme $\{ w^{(j)} : j \in \cJ \}$ (recall \eqref{eq:approx-weight}):
\begin{equation}\label{eq:aw2}
   \alpha^{-1} h_i(\|U_j^{-1} a_i\|_2)^2 \leq 2^j w_i^{(j)} \|U_j^{-1} a_i\|^2_2 \leq \alpha h_i(\|U_j^{-1} a_i\|_2)^2\,,\quad i=1,\ldots,m, \ j \in \cJ\,.
\end{equation}
\pref{lem:contain1} follows immediately from the next two lemmas.

\begin{lemma}
   If $j \in \cJ$ and $|\langle a_i,x\rangle| \leq c_0 \|U_j^{-1} a_i\|_2$ for each $i \in \{1,\ldots,m\}$, then $x \in \sfK_j$.
\end{lemma}

\begin{proof}
   Suppose that $|\langle a_i,x\rangle| \leq \delta \|U_j^{-1} a_i\|_2$ for some $0 < \delta < 1$.
   Then
\[
   \frac{h_i(\langle a_i,x\rangle)^2}{2^j w_i^{(j)} \|U_j^{-1} a_i\|_2^2} \stackrel{\ref{ass:monotone}}{\leq} \frac{1}{c^2} \frac{h_i(\delta \|U_j^{-1} a_i\|_2)^2}{2^j w_i^{(j)} \|U_j^{-1} a_i\|^2_2} 
   \stackrel{\eqref{eq:aw2}}{\leq} \frac{\alpha}{c^2}  \frac{h_i(\delta \|U_j^{-1} a_i\|_2)^2}{h_i(\|U_j^{-1} a_i\|_2)^2}
   \stackrel{\ref{ass:lower-growth}}{\leq} \frac{\alpha}{c^4} \delta^{2\lo} \,.
\]
Taking $\delta \seteq (c^4/(4 \alpha))^{1/2\lo}$ gives $x \in \sfK_j$.
\end{proof}

\begin{lemma}\label{lem:contain1.2}
   If $j \in \cJ$ and $x \in \sfK_j$, then
   \begin{equation}\label{eq:K2s1}
      |\langle a_i,x\rangle| \leq (\alpha/c^2)^{1/2\lo} \|U_{j}^{-1} a_i\|_2\,,\quad \forall i \in \{1,\ldots,m\}\,.
   \end{equation}
\end{lemma}

\begin{proof}
   Fix $i \in \{1,\ldots,m\}$.
   Clearly we may assume that $|\langle a_i,x\rangle| \geq \|U_{j}^{-1} a_i\|_2$.
   In that case, we can bound
\[
   c^2 \left(\frac{|\langle a_i, x\rangle|}{\|U_{j}^{-1} a_i\|_2}\right)^{2\lo}
   \sr{\ref{ass:lower-growth}}{\leq} \frac{h_i(\langle a_i,x\rangle)^2}{h_i(\|U_{j}^{-1} a_i\|_2)^2} 
   \sr{x \in \sfK_{j}}{\leq} \frac{2^{j} w_i^{(j)} \|U_{j}^{-1} a_i\|_2^2}{h_i(\|U_{j}^{-1} a_i\|_2)^2} \sr{\eqref{eq:aw2}}{\leq} \alpha\,,
\]
   establishing \eqref{eq:K2s1}.
\end{proof}

Let us now move to the proof of \pref{lem:contain2}.
The next lemma follows from an argument identical to that of \pref{lem:s2-norm-bound}.

\begin{lemma}\label{lem:base-norm}
   For any $j \in \cJ$, it holds that
   \[
      \|U_{j} x\|^{2\lo}_2 \leq \max\left(1, \alpha \frac{C}{c} 2^{-j} F(x) \right).
   \]
\end{lemma}

\begin{lemma}\label{lem:K2s}
   If $j,j+1 \in \cJ$ and
   $x \in \sfK_{j+1}$, then
   \[
      \|U_j x\|_2^2 \le C_0 2^{-j} F(x)\,.
   \]
\end{lemma}

\begin{proof}
   Since $j+1 \in \cJ$, from \pref{lem:contain1.2} we have $|\langle a_i,x\rangle| \leq (\alpha/c^2)^{1/(2\theta)}$ for each $i =1,2\ldots,m$.
   Therefore
\[
   \frac{h_i(\langle a_i,x\rangle)^2}{\langle a_i,x\rangle^2} \stackrel{\ref{ass:upper-growth}}{\geq} 
   \frac1{C^2} \frac{h_i((\alpha/c^2)^{1/2\lo} \|U_{j+1}^{-1} a_i\|_2)^2}{(\alpha/c^2)^{1/\lo} \|U_{j+1}^{-1} a_i\|_2^2}
   \stackrel{\ref{ass:monotone}}{\geq} \frac{c^2}{C^2 (\alpha/c^2)^{1/\lo}} \frac{h_i(\|U_{j+1}^{-1} a_i\|_2)^2}{\|U_{j+1}^{-1} a_i\|_2^2}
   \stackrel{\eqref{eq:aw2}}{\geq}
   \frac{c^2}{\alpha} \frac{2^{j+1} w_i^{(j+1)}}{C^2 (\alpha/c^2)^{1/\lo}}\,.
\]
It follows that
\[
   \|U_{j} x\|^2 \stackrel{\eqref{eq:Uexpr}}{=} 
   \sum_{i=1}^m w_i^{(j)} \langle a_i, x\rangle^2 \stackrel{\eqref{eq:smooth-weights}}{\leq} \alpha \sum_{i=1}^m w_i^{(j+1)} \langle a_i,x\rangle^2 
   \leq \frac{(\alpha C/c)^2 (\alpha/c^2)^{1/\lo}}{2^{j+1}} \sum_{i=1}^m h_i(\langle a_i,x\rangle)^2\,.\qedhere
\]
\end{proof}

\begin{corollary}\label{cor:K2s}
   For any $\gamma \geq 1$, it holds that
   if $j \in \cJ$ and $x \in B_F(\gamma 2^{\mm}) \cap \sfK_{j+1}$, then
   \[
      \|U_j x\|_2^2 \le \gamma C_0 2^{\mm -j}\,.
   \]
\end{corollary}

\begin{proof}
   If $j,j+1 \in \cJ$, the result follows from \pref{lem:K2s}.
   Otherwise, $j=\mm-1$ and $\sfK_{j+1}=\sfK_{\mm}=B_F(2^{\mm})$, and it follows from \pref{lem:base-norm}.
\end{proof}

\begin{proof}[Proof of \pref{lem:contain2}]
Note that $B_F(2^{\mm}) \cap (z + \sfK_{j+1}) = z + \left((B_F(2^{\mm}) - z) \cap \sfK_{j+1}\right)$. We claim that $B_F(2^{\mm}) - z \subseteq B_F(4L^2 2^{\mm})$. Indeed, for $x \in B_F(2^{\mm})$, we have
\[ F(x-z) = \sum_{i=1}^m h_i(\l a_i, x-z\r)^2 \stackrel{\ref{ass:subadditive}}{\le} \sum_{i=1}^m (L(h_i(\l a_i,x\r) + h_i(\l a_i, z\r)))^2 \le 2L^2(F(x) + F(z)) \le 4L^2 2^\mm. \]
Thus \pref{cor:K2s} gives
\[ B_F(2^{\mm}) \cap (z + \sfK_{j+1}) \subseteq z + (4L^2 C_0 2^{\mm-j})^{1/2} \sfB_j^2\,.\qedhere \]
\end{proof}

\subsection{Sparsification analysis}
\label{sec:sparsification-analysis}

Consider now weights $\tau \in \R^m_+$ satisfying
\begin{equation}\label{eq:good-weights}
   \tau_i \geq \max_{j \in \cJ} w_i^{(j)} \|U_j^{-1} a_i\|_2^2\,,\quad i=1,\ldots,m\,,
\end{equation}
Define $\rho_i \seteq \tau_i/\|\tau\|_1$ for $i=1,\ldots,m$.
Let us analyze the error from taking $M$ independent samples according to $\rho$,
following the sparsification framework of \pref{sec:subgaussian} with $\f_i(x) \seteq f_i(\langle a_i,x\rangle)$.

In this case, our potential sparsifier is given by
\begin{equation}\label{eq:sparsifier-def}
   \tilde{F}_{\rho,\nu}(x) \seteq \sum_{j=1}^M \frac{f_{\nu_j}(\langle a_{\nu_j},x\rangle)}{M \rho_{\nu_j}}\,,
\end{equation}
and our approximation guarantee will be derived from \pref{lem:sparsification-meta}, by analyzing covering numbers
in the following family of metrics:
For a sequence $\nu \in [m]^d$ and $x, y \in \R^n$,
\begin{align*}
   d_{\rho,\nu}(x, y) &\seteq \left(\sum_{j=1}^M \left(\frac{f_{\nu_j}(\langle a_{\nu_j}, x\rangle) - f_{\nu_j}(\langle a_{\nu_j}, y\rangle)}{M\rho_{\nu_j}} \right)^2 \right)^{1/2} \\
   &= \left(\sum_{j=1}^M \frac{(h_{\nu_j}(\langle a_{\nu_j}, x \rangle) - h_{\nu_j}(\langle a_{\nu_j}, y \rangle))^2}{M \rho_{\nu_j}} \frac{(h_{\nu_j}(\langle a_{\nu_j}, x \rangle) + h_{\nu_j}(\langle a_{\nu_j}, y \rangle))^2}{M \rho_{\nu_j}} \right)^{1/2}.
\end{align*}

Let us first observe the bound: For $x,y \in B_F(2^{\mm})$,
\begin{align}
   d_{\rho,\nu}(x, y)\  
                   &\stackrel{\ref{ass:lipschitz}}{\leq} \ \ 
                   \left(2L^2 \sum_{j=1}^M  \frac{h_{\nu_j}(\langle a_{\nu_j}, x-y\rangle)^2}{M\rho_{\nu_j}}
                   \frac{h_{\nu_j}(\langle a_{\nu_j}, x\rangle)^2 + h_{\nu_j}(\langle a_{\nu_j}, y\rangle)^2}{M\rho_{\nu_j}}\right)^{1/2} \nonumber \\ 
                   &\ls  d_{\infty}(x,y) L\left(\frac{\|\tau\|_1}{M}\right)^{1/2} \left(\max_{x \in B_F(2^{\mm})} \tilde{F}_{\rho,\nu}(x)\right)^{1/2} \label{eq:dbound}\,,
\end{align}
where we define
\[
d_{\infty}(x,y) \seteq \max_{i \in [m]} \frac{h_i(\langle a_i, x-y\rangle)}{\sqrt{\tau_i}}.
\]
In particular, we have
\begin{equation}\label{eq:dinfty}
   \gamma_2(B_F(2^{\mm}), d_{\rho,\nu}) \ls L \gamma_2(B_F(2^{\mm}), d_{\infty}) \left(\frac{\|\tau\|_1}{M}\right)^{1/2} \left(\max_{x \in B_F(2^{\mm})} \tilde{F}_{\rho,\nu}(x)\right)^{1/2}.
\end{equation}

Let us first bound the $d_{\infty}$ diameter of $B_F(2^{\mm})$.

\begin{lemma}\label{lem:BF-diam}
   For any $j \in \cJ$, it holds 
   $\diam(B_F(2^{j+1}), d_{\infty}) \le C_2 2^{j/2}$ for $C_2 \seteq 4 (\alpha C_0 C^2)^{1/2}$.
\end{lemma}

\begin{proof}
   Using \eqref{eq:good-weights} and Cauchy-Schwarz yields
   \begin{equation}\label{eq:BF1}
      \frac{h_i(\langle a_i,x\rangle)^2}{\tau_i} \leq \frac{h_i(\langle a_i,x\rangle)^2}{w_i^{(j)} \|U_{j}^{-1} a_i\|_2^2}
      \stackrel{\ref{ass:monotone}}{\leq} \frac{1}{c} \frac{h_i(\|U_{j}^{-1} a_i\|_2 \|U_{j} x\|_2)^2}{w_i^{(j)} \|U_{j}^{-1} a_i\|_2^2}\,.
   \end{equation}

   Now from \pref{lem:base-norm}, we know that for $x \in B_F(2^{j+1})$, we have $\|U_{j} x\|_2^2 \le 2 C_0$,
   and therefore
   \[
      h_i(\|U_{j}^{-1} a_i\|_2 \|U_{j} x\|_2)^2 \stackrel{\ref{ass:upper-growth}}{\le} 2 C_0 C^2 h_i(\|U_{j}^{-1} a_i\|_2)^2
   \]
   In conjunction with \eqref{eq:BF1} and \eqref{eq:aw2}, we conclude that
   \[
      \frac{h_i(\langle a_i,x\rangle)^2}{\tau_i} \le 2 \alpha C_0 C^2 2^j\,,\quad i=1,\ldots,m\,,
   \]
   which implies the desired bound.
\end{proof}

Combining this with \eqref{eq:dbound} yields a diameter bound.

\begin{corollary}\label{cor:BF-diam}
   For any $j \in \cJ$, it holds that
   \[ \diam(B_F(2^{j+1}), d_{\rho,\nu}) \ls 2^{j/2} L C_2 \left(\frac{\|\tau\|_1}{M}\right)^{1/2} \left(\max_{x \in B_F(2^{\mm})} \tilde{F}_{\rho,\nu}(x)\right)^{1/2}. \]
\end{corollary}

We now handle small scales. For this, define $B_{\infty}(r) \seteq \{ x \in \R^n : d_{\infty}(x,0) \leq r \}$.
\begin{lemma}\label{lem:small-scales}
   For any $0 < r < R$, we have
   \[
      \sqrt{\log \cov(B_{\infty}(R), B_{\infty}(r))} \ls \sqrt{\frac{m}{\lo} \log \left(\frac{2R}{cr} \right)}\,.
   \]
\end{lemma}
\begin{proof}
   Denote $z_i \seteq \sup\{|z| : h_i(z) \le R\sqrt{\tau_i}\}$ for $i \in [m]$ and $\lambda \seteq (CR/(cr))^{1/\lo}$.
   For $x \in \R^n$, define the vector $a(x) := \left(\l a_1, x\r, \dots, \l a_m, x\r\right)$.
   Then
   \[
      |a(x)_i| \leq z_i\,,\quad \forall x \in B_{\infty}(R), \ i=1,\ldots,m\,.
   \]
   By the Pigeonhole principle, there is a set $S \subseteq B_{\infty}(R)$ with
$|S| \le (2\lambda+1)^m$, and such that for all $x \in B_{\infty}(R)$,
there is $\hat{x} \in S$ with $|\langle a_i, x-\hat{x}\rangle| \le z_i/\lambda$ for every $i \in \{1,\ldots,m\}$.
Thus for each $i$, we have
\begin{align*}
   \frac{h_i(\l a_i, x-\hat{x}\r)}{\sqrt{\tau_i}} \stackrel{\ref{ass:lower-growth}}{\le} \frac{1}{c\lambda^{\lo}} \frac{h_i(z_i)}{\sqrt{\tau_i}} \le r,
\end{align*}
One concludes that $\sqrt{\log \cov(B_{\infty}(R), B_{\infty}(r))} \le \sqrt{\log |S|}$, completing the proof.
\end{proof}

\begin{theorem}
\label{thm:dudley}
   There is a number $K_1 = K_1(L,C,c,\theta,\alpha)$ such that if
   $|\cJ| \geq 4 \log m$, then
   \[
      \gamma_2\left(B_F(2^{\mm}), d_{\rho,\nu}\right) \leq K_1  2^{\mm /2} (\log m)^{3/2} \left(\frac{\|\tau\|_1}{M}\right)^{1/2} \left(\max_{x \in B_F(2^{\mm})} \tilde{F}_{\rho,\nu}(x)\right)^{1/2}
   \]
\end{theorem}

\begin{proof}
First, the assumption \eqref{eq:good-weights} guarantees that
\[
   \sfK_j \subseteq B_{\infty}(2^{j/2}) \,,\quad \forall j \in \cJ,
\]
and therefore \pref{thm:main-cov} gives
\begin{equation*}\label{eq:inf-cover}
   \left(\log \cov\left(B_F(2^{\mm}), B_{\infty}(2^{j/2})\right)\right)^{1/2} \leq 2^{(\mm-j)/2} \sqrt{\hat{C} \log m}\,,\quad \forall j \in \cJ\,.
\end{equation*}

We now use Dudley's inequality (\pref{lem:dudley})  to write
\[
   \gamma_2(B_F(2^{\mm}), d_{\infty}) \ls \sum_{j \in \Z } 2^{j/2} \sqrt{\log \cov(B_F(2^{\mm}), B_{\infty}(2^{j/2}))}\,.
\]
Define $\hat{\ell} \seteq \lceil \log_2 \diam(B_F(2^{\ell}), d_{\infty})\rceil$ and consider any $\ell_0 \leq \ell$.
Splitting the sum into three pieces, we use \eqref{eq:inf-cover} to bound the first two:
\begin{align*}
   \sum_{j > \ell} 2^{j/2} \sqrt{\log \cov(B_F(2^{\mm}), B_{\infty}(2^{j/2}))} &\lesssim \diam(B_F(2^{\mm}), d_{\infty}) \sqrt{\hat{C} \log m}\\
   \sum_{j=\ell_0}^{\ell} 2^{j/2} \sqrt{\log \cov(B_F(2^{\mm}), B_{\infty}(2^{j/2}))} &\lesssim 2^{\ell/2} (\ell-\ell_0) \sqrt{\hat{C} \log m}\,.
\end{align*}
For the third piece, use \pref{lem:BF-diam} to bound $B_F(2^{\ell}) \subseteq B_{\infty}(2^{\hat{\ell}})$ so that
\begin{align*}
   \sum_{j < \ell_0} 2^{j/2} \sqrt{\log \cov(B_F(2^{\mm}), B_{\infty}(2^{j/2}))} &\leq
   \sum_{j < \ell_0} 2^{j/2} \sqrt{\log \cov(B_{\infty}(2^{\hat{\mm}}), B_{\infty}(2^{j/2}))}  \\
                                                                                 &\lesssim 2^{\ell_0/2} \sqrt{\frac{m}{\theta} \left((\hat{\ell}-\ell_0) + \log \frac{2}{c}\right)}\,.
\end{align*}
As long as we can choose $\mm_0 \leq \mm - 4 \log m$, we have therefore bounded
\begin{equation}
   \gamma_2(B_F(2^{\mm}), d_{\infty}) \leq K_1 2^{\ell/2} (\log m)^{3/2}\,, \label{eq:dinfbound}
\end{equation}
for some $K_1$ depending on the indicated parameters.
Combining this with \eqref{eq:dinfty} yields the desired bound.
\end{proof}

Finally, we can prove \pref{thm:main}.

\begin{proof}[Proof of \pref{thm:main}]
Let us denote $\cJ \seteq \{\jmin, \dots, \jmax\}$, where 
\begin{align*}
   \jmax \seteq \lceil \log_2 \smax \rceil, \ \jmin \seteq \lfloor \log_2 \smin - 4 \log m \rfloor\,.
\end{align*}
\pref{thm:weight-schemes} yields the existence of an $\alpha$-approximate weight scheme
$\{ w^{(j)} : j \in \cJ \}$ with $\alpha \ls_{L,c,\lo} 1$.
This yields $\tau$ as in \eqref{eq:good-weights} satisfying
\[
   \|\tau\|_1 \ls_{L,c,\lo} n |\cJ| \ls n \log(m\smax/\smin)\,.
\]

Now let us choose an integer $M \seteq C \eps^{-2}\|\tau\|_1 (\log m)^3$ for a sufficiently large constant $C \geq 1$,
and denote $\Omega_j \seteq \{ x \in \R^n : \hat{F}(x) \le 2^j\}$ for $j \in \cJ$.
Then \pref{thm:dudley} gives us the following, in relation to the assumptions of \pref{lem:sparsification-meta}:
\[ \delta \ls 2^{j/2} \log(m) \sqrt{C_1 \log m} \left(\frac{\|\tau\|_1}{M}\right)^{1/2} \leq \eps 2^{j/2}, \]
and \pref{cor:BF-diam} gives
\[ \bdelta \ls 2^{j/2} L\,C_2 \left(\frac{\|\tau\|_1}{M}\right)^{1/2} \leq \frac{\eps 2^{j/2}}{(\log m)^{3/2}}\,. \]

Thus \pref{lem:sparsification-meta} shows that for some constant $c_0 \gtrsim_{L,c,\theta} 1$, with probability at least $1 - e^{-c_0(\log m)^3}|\cJ|$
over the choice of $\bm{\nu} \in [m]^M$,
\begin{equation}\label{eq:Fapprox0}
   \max_{x \in \Omega_j} |\tilde{F}_{\rho,\bm{\nu}}(x)-F(x)| \ls \e 2^j\quad \forall j \in \cJ\,,
\end{equation}
where $\tilde{F}_{\rho,\nu}$ is defined as in \eqref{eq:sparsifier-def}.
Moreover, $\tilde{F}_{\rho,\nu}$ is manifestly $s$-sparse with 
\begin{equation}\label{eq:sparse0}
   s \leq M \ls_{L,c,\theta} \frac{n}{\e^2} \log(m\smax/\smin) (\log m)^3\,.
\end{equation}
Observe that \eqref{eq:Fapprox0} gives
\begin{equation}\label{eq:Fapprox1}
   |\tilde{F}(x)-F(x)| \ls \e F(x)\,,\quad \forall x \in \R^n \textrm{ with } F(x) \in [\smin,\smax]\,,
\end{equation}
with $\tilde{F}=\tilde{F}_{\rho,\nu}$.
We remark that \eqref{eq:sparse0} is slightly worse than the bound claimed in \pref{thm:main}.
We address this at the end of the proof.

\paragraph{The algorithmic argument}
To produce a sparsifier algorithmically, we need to be slightly more careful.
Let $\hat{f}_1, \dots, \hat{f}_m : \R^n \to \Rplus$ be the perturbations guaranteed by \pref{thm:perturb}
with $\delta \asymp_{L,c,\theta} \e \smin/(m^3 \smax)$ so that \eqref{eq:hat-approx} gives
\begin{equation}\label{eq:hat-approx1}
   |F(x)-\hat{F}(x)| \leq \e \smin\,,\quad \forall x \in \R^n \textrm{ with } F(x) \leq \smax\,,
\end{equation}
where $\hat{F}(x) := \sum_{i=1}^m \hat{f}_i(\l a_i, x\r)$.

From a combination of \pref{thm:perturb} and \pref{thm:weightalgo}, we can algorithmically produce an $\alpha$-approximate
weight scheme
$\{ w^{(j)} : j \in \cJ \}$ for $\{\hat{f}_i\}$ and $\{a_i\}$ with $\alpha \ls_{L,c,\lo} 1$, in time 
\[
   \O_{L,c,\lo}\!\left((\nnz(a_1,\ldots,a_m) + n^\omega + m \tEval) \log \frac{m \smax}{\e \smin}\right)\,.
\]
By \pref{fact:hf-props}, the preceding analysis applies to $\hat{f}_1,\ldots,\hat{f}_m$, yielding, with probability
at least $1 - e^{-c_0 (\log m)^3} |\cJ|$, a sparsifier $\tilde{F}$ of $\hat{F}$ with
\[
   |\tilde{F}(x)-\hat{F}(x)| \ls \e \hat{F}(x)\,,\quad \forall x \in \R^n \textrm{ with } \hat{F}(x) \in [\tfrac12 \smin, 2\smax]\,,
\]
Combined with \eqref{eq:hat-approx1} and the triangle inequality, this gives
\[
   |\tilde{F}(x)-F(x)| \ls \e F(x)\,,\quad \forall x \in \R^n \textrm{ with } F(x) \in [\smin, \smax]\,,
\]
completing the proof.

\paragraph{Improving the $m$ dependence}
The preceding arguments yield a sparsity bound of the form \eqref{eq:sparse0}, together
with an approximation guarantee \eqref{eq:Fapprox1}.
Define $\tilde{F}_0 \seteq F$ and $m_0 \seteq m$.
Suppose, inductively, that $\tilde{F}_{i}$ has $m_i$ non-zero terms. Let $\e_{i} > 0$
be a given, and let $\tilde{F}_{i+1}$ be such that \eqref{eq:Fapprox1}
holds with $F = \tilde{F}_i$ and $\tilde{F} = \tilde{F}_{i+1}$, and such that $\tilde{F}_{i+1}$
has at most $m_{i+1}$ non-zero terms, where
\begin{equation}\label{eq:mi}
   m_{i+1} \ls \frac{n}{\e_i^2} \log(m_i\hatsmax/\hatsmin) (\log m_i)^3\,.
\end{equation}
By scaling $\tilde{F}_i$ and adjusting $\e_i$ by a constant factor, we may assume that
\begin{equation}\label{eq:Fapprox2}
   (1-\e_i) \tilde{F}_i(x) \leq \tilde{F}_{i+1}(x) \leq \tilde{F}_i(x)\,,\quad \forall x \in \R^n \textrm{ with } \tilde{F}_i(x) \in [\hatsmin, \hatsmax]\,.
\end{equation}
Then the triangle inequality gives, for any $h \geq 1$, and all $x \in \R^n$ satisfying
$\tilde{F}_0(x),\ldots,\tilde{F}_{h-1}(x) \in [\hatsmin,\hatsmax]$,
\begin{equation}\label{eq:Fh-triangle}
   |F(x)-\tilde{F}_h(x)| \leq \sum_{i=0}^{h-1} |\tilde{F}_i(x) - \tilde{F}_{i+1}(x)| \leq \sum_{i=0}^{h-1} \e_i \tilde{F}_i(x)
   \leq F(x) \sum_{i=0}^{h-1} \e_i\,.
\end{equation}

It is straightforward that one can choose a geometrically increasing sequence $\e_0 < \e_1 < \cdots < \e_{h-1} = \e$ 
such that $m_h \ls  \frac{n}{\e^2} \log (S \frac{\hatsmax}{\hatsmin}) (\log S)^3$, with $S \seteq \frac{n}{\e} \log \frac{2\hatsmax}{\hatsmin}$.
Since the sequence $\{ \e_i \}$ increases geometrically, the last expression in \eqref{eq:Fh-triangle} is at most $C \e F(x)$
for some $C > 1$.

Let us now take $\hatsmin \seteq \frac12 \smin, \hatsmax \seteq \smax$.
We may assume that $\e < 1/(2C)$, meaning that if $F_0(x) = F(x) \in [\smin,\smax]$, then
$\tilde{F}_0(x),\ldots,\tilde{F}_{h-1}(x) \in [\hatsmin,\hatsmax]$, and therefore $\tilde{F}_h$
yields the desired sparsifier.
\end{proof}

\subsection{The $\gamma_p$ losses}
\label{sec:gammap}

First, let us verify that the $\gamma_p$ losses satisfy the assumptions of \pref{thm:main}.
Define a generalization with variable thresholds: For $t \geq 0$,
\begin{equation}\label{eq:gammap-gen}
   \gamma_p(t, z) := \begin{cases} \frac{p}{2}t^{p-2}z^2\,, & |z| \le t \\ |z|^p - (1-\frac{p}{2})t^p\,, & |z| \ge t, \end{cases}
\end{equation}

\begin{lemma}
   \label{lem:approx-concave}
   Suppose $h : \R \to \R$ is continuous, symmetric and differentiable at all but finitely many points. If $h$ is additionally approximately concave in the sense that $h'(x) \leq L\, h'(y)$ for $0 < x < y$, then
   $h$ is $L$-auto-Lipschitz.
\end{lemma}

\begin{proof}
   Consider $0 < x < y$. Then,
   \[
      h(y) - h(x) = \int_x^y h'(z)\,dz \leq L \int_0^{y-x} h'(z)\,dz = L \, h(y-x)\,.\qedhere
   \]
\end{proof}

\begin{lemma}\label{lem:gammap-nice}
   For every $p \in (0,2]$ and $t > 0$, the function $z \mapsto \gamma_p(t,z)^{1/2}$ is $1$-auto-Lipschitz \ref{item:lipschitz} and lower $p/2$-homogeneous \ref{item:lower-growth}
   with constant $1$.
\end{lemma}

\begin{proof}
   Fix $p \in (0,2]$ and $t > 0$, and define $h(z) \seteq \gamma_p(t,z)^{1/2}$.
By \pref{lem:approx-concave}, to show that $h$ is $1$-auto-Lipschitz, it suffices to show that $h'(z)$ is non-increasing for $z > 0$.
This is clear for $z < t$, since $h'(z)$ is independent of $z$. For $z > t$, we calculate
\begin{align*}
   h''(z) &= \frac{p(p-1)z^{p-2}}{2(z^p + (1-p/2)t^p)^{1/2}} - \frac{p^2 z^{2p-2}}{4(z^p + (1-p/2)t^p)^{3/2}} \\
          &= \frac{p (2-p)z^{p-2}}{4(z^p + (1-p/2) t^p)^{3/2}} \left((p-1) t^p - z^p\right) \leq 0\,.
\end{align*}

To see that $h(z)^2$ is lower $p$-homogeneous with constant $1$, we first check this for $x < y \le t$. Then, $h(y)/h(x) = (y/x)^2 \ge (y/x)^p$. If $t \le x < y$, then
\[ \frac{h(y)^2}{h(x)^2} = \frac{y^p - (1-\frac{p}{2})t^p}{x^p - (1-\frac{p}{2})t^p} \ge (y/x)^p. \]
For $x < t < y$, we write
\[ \frac{h(y)^2}{h(x)^2} = \frac{h(y)^2}{h(t)^2} \cdot \frac{h(t)^2}{h(x)^2} \ge (y/t)^p (t/x)^p = (y/x)^p.\qedhere \]
\end{proof}

\subsubsection{The Huber Loss}
\label{sec:huberproofs}

Let us now prove \pref{lem:huber} and \pref{cor:huber}.

\begin{proof}[Proof of \pref{lem:huber}]
   By scaling, we may assume that $w_{\max}=1$.
   Accordingly, let us define $\smin \seteq w_{\min}$ and $\smax \seteq 6 m^2$.
   Consider $x \in \R^n$ with $F(x) \leq w_{\min}$.
   In that case, $\gamma_1(\langle a_i,x\rangle) \leq 1$,
   and therefore $|\langle a_i,x\rangle| \leq 1$ for all $i=1,\ldots,m$.
   In particular, this implies that $\gamma_1(\langle a_i,x\rangle) = \langle a_i,x\rangle^2/2$ for each $i$
   so that $F(\lambda x) = \lambda^2 F(x)$ and $\tilde{F}(\lambda x) = \lambda^2 \tilde{F}(x)$ for $\lambda \seteq (F(x)/w_{\min})^{-1/2}$.
   If we take $y \seteq \lambda x$, then $F(y) = w_{\min} \in [\smin,\smax]$, so it holds that $|F(y)-\tilde{F}(y)| \leq \e F(y)$.
   Scaling by $\lambda$ recovers the same guarantee for $x$.

   \smallskip

   Now consider $x \in \R^n$ with $F(x) \geq \smax$.
   Note that $|\langle a_i,y\rangle| - \tfrac12 \leq \gamma_1(\langle a_i,y\rangle) \leq |\langle a_i,y\rangle|$ holds for
   any $y \in \R^n$ and $i \in \{1,\ldots,m\}$, 
   therefore
   \[
      F(x) \leq \sum_{i=1}^m |\langle a_i,x\rangle| \leq F(x) + m/2\,.
   \]

   Denote $\lambda \seteq \smax/(2 F(x))$ and $y \seteq \lambda x$, and note that
   \begin{align*}
      F(y) = \sum_{i=1}^m \gamma_1(\langle a_i,y\rangle) \leq \lambda \sum_{i=1}^m |\langle a_i,x\rangle| \leq \lambda(F(x) + m/2) &\leq \smax \\
      F(y) \geq w_{\min} \left(- m/2 + \lambda \sum_{i=1}^m |\langle a_i,x\rangle|\right) \geq w_{\min} \left(\lambda F(x) - m/2\right) &\geq \smin\,.
   \end{align*}
   Therefore $|F(y)-\tilde{F}(y)| \leq \e F(y)$ by assumption.

   Now use the fact that $|\gamma_1(\beta z)-\beta \gamma_1(z)| \leq \beta$ for all $\beta > 1$ and $z \in \R$, implying
   \begin{align*}
      |F(x)-\lambda^{-1} F(y)| &\leq \frac{1}{\lambda} \sum_{i=1}^m w_i \leq \frac{m}{\lambda}\,, \\
      |\tilde{F}(x)-\lambda^{-1} \tilde{F}(y)| &\leq \frac{1}{\lambda} \sum_{i=1}^m \tilde{w}_i \leq \frac{m}{\lambda}\,.
   \end{align*}
   The triangle inequality gives
   \[
      |F(x)-\tilde{F}(x)| \leq \frac{2 m}{\lambda} + \frac{|F(y)-\tilde{F}(y)|}{\lambda} 
      \leq \frac{4m}{\smax} F(x) + \frac{\e}{\lambda} F(y) \leq 2 \e F(x)\,,
   \]
   where we have used $\smax \geq 4 m^2$ and $\e > 1/m$.
\end{proof}

\begin{proof}[Proof of \pref{cor:huber}]
Note that the sparsifier $\tilde{F} = \tilde{F}_{\rho,\nu}$ satisfying the guarantee \eqref{eq:Fapprox1}
is produced via importance sampling with respect to the distribution specified by $\rho \in \Rplus^m$,
and if the $i$th term has probability $\rho_i$ to be sampled, then it is weighted by $(\rho_i M)^{-1}$
(recalling \eqref{eq:sparsifier-def}).
Moreover, $\tilde{F}_{\rho,\bm{\nu}}$ satisfies \eqref{eq:Fapprox1} with probability $1-o(1)$ as $m \to \infty$.

By a union bound, $\Pr(\min_{j \in [M]} \rho_{\bm{\nu}_j} \leq \delta) \leq \delta Mm$.
Therefore with probability at least $1/2$, it holds $\min_{j \in [M]} \rho_{\bm{\nu}_j} \geq \frac{1}{2M m}$.
If we start with a uniform-weighted function $F(x)=\gamma_1(\langle a_1,x\rangle)+\cdots+\gamma_1(\langle a_m,x\rangle)$,
then with probability at least $1/2$, the maximum weight in $\tilde{F}$ will be $2m$.

Thus if $\tilde{F}$ satisfies \eqref{eq:Fapprox1} with $\smin = 1/2$ and $\smax \seteq 8m^3$, then \pref{lem:huber}
shows that that $\tilde{F}$ is a genuine $O(\e)$-approximation to $F$. Now the sparsity guarantee of \pref{thm:main} (or see
\eqref{eq:sparse0}) completes the proof.
\end{proof}

\subsubsection{The Tukey Loss}
\label{sec:tukey}

\newcommand{\tukey}{\mathsf{T}}

Let us denote the Tukey loss $\tukey(z) := \min\{z^2, 1\}$, and use $h(z) \seteq \tukey(z)^{1/2} = \min\{|z|, 1\}$.
We use the notation from the beginning of the section.

\begin{theorem}
\label{thm:tukey}
Suppose $a_1, \dots, a_m \in \R^n$ satisfy $\|a_i\|_2 \le n^{O(1)}$ for $i=1,\ldots,m$, and denote $F(x) \seteq \sum_{i=1}^m \tukey(\l a_i, x\r)$. Then there are weights $w \in \Rplus^m$ with $|\supp(w)| \le n^{1+o(1)}\eps^{-2}(\log m)^4$, and such that
\begin{equation}\label{eq:tukey-approx}
(1-\eps)F(x) \le \sum_{i=1}^m w_i \tukey(\l a_i, x\r) \le (1+\eps)F(x)\,,\quad \forall x \in \R^n \text{ with } \|x\|_2 \le n^{O(1)}\,.
\end{equation}
The weights $w$ can be computed in time $\tilde{O}(\nnz(a_1,\ldots,a_m) + n^\omega + m)$ with high probability.
\end{theorem}

The remainder of the section is sketching a proof of \pref{thm:tukey}.
Note that the Tukey loss satisfies \ref{item:lipschitz} and \ref{item:symmetric}--\ref{item:upper-growth}, but not \ref{item:lower-growth}, which
is why \pref{thm:main} is not directly applicable.
This motivates us to define the function $\bar{h}(x) := \min\{|x|, |x|^\eta\}$ for some small $\eta > 0$ to be chosen later, and
$\bar{f_i}(x) \seteq \bar{h}(\langle a_i,x\rangle)^2$ for $i=1,\ldots,m$, along with 
$\bar{F}(x) \seteq \sum_{i=1}^m \bar{f_i}(\l a_i, x\r)$.

\smallskip

By \pref{thm:weight-schemes},
for some $\alpha = \alpha(\eta)$, 
there is an $\alpha$-approximate weight scheme $\{ w^{(j)} : j \in \cJ \}$ for $\{\bar{f_i}\}$ and $\{a_i\}$, where
$\cJ \seteq \{ j \in \Z : |j| \leq C \log m \}$, with the constant $C > 1$ chosen sufficiently large.
Define $\tau_i \seteq \max_{j\in J} w_i^{(j)} \|U_j^{-1} a_i\|_2^2$ for $i=1,\ldots,m$ so that $\|\tau\|_1 \ls n |\cJ| \ls n \log m$,
along with sampling weights $\rho_i \seteq \tau_i/\|\tau\|_1$ for $i=1,\ldots,m$. 
\smallskip

Take $\Omega \seteq \{x : \|x\|_2 \le n^{O(1)}\}$, with the notation of \pref{sec:subgaussian} and $\f_i(x) \seteq \tukey(\langle a_i,x\rangle)$, and $\tilde{F}(x) = \sum_{i=1}^M \frac{1}{M\rho_{\nu_i}} \tukey(\langle a_{\nu_i}, x\rangle)$ as constructed in \pref{sec:subgaussian}. Consider $\mm \in \Z$ with $|\mm| \le 2\log m$.
We can estimate (as in \eqref{eq:dbound}): for any $x,y \in \Omega \cap B_F(2^{\mm})$,

\[ d_{\rho,\nu}(x, y) \ls d_{\infty}(x,y) \left(\frac{\|\tau\|_1}{M}\right)^{1/2} \left(\max_{x \in \Omega \cap B_F(2^{\mm})} \tilde{F}_{\rho,\nu}(x)\right)^{1/2} \enspace \text{ for } \enspace d_{\infty}(x,y) \seteq \max_{i \in [m]} \frac{h(\langle a_i, x-y\rangle)}{\sqrt{\tau_i}}, \]
One should observe that \eqref{eq:dbound} only requires \ref{ass:lipschitz}, which $\tukey(z)^{1/2}$ satisfies.

Let us also define the distance $\bar{d}_{\infty}(x, y) \seteq \max_{i \in [m]} \frac{\bar{h}(\langle a_i, x-y\rangle)}{\sqrt{\tau_i}} \ge d_\infty(x, y)$.
Note that $\Omega \cap B_F(2^{\mm}) \subseteq B_{\bar{F}}(n^{O(\eta)}2^{\mm})$. Therefore,
\begin{align*}
\diam(\Omega \cap B_F(2^{\mm}), d_{\rho,\nu}) &\ls \left(\frac{\|\tau\|_1}{M}\right)^{1/2} \left(\max_{x \in \Omega \cap B_F(2^{\mm})} \tilde{F}_{\rho,\nu}(x)\right)^{1/2} \cdot \diam(\Omega \cap B_F(2^{\mm}), \bar{d}_\infty) \\
&\leq \left(\frac{\|\tau\|_1}{M}\right)^{1/2} \left(\max_{x \in \Omega \cap B_F(2^{\mm})} \tilde{F}_{\rho,\nu}(x)\right)^{1/2} \cdot \diam(B_{\bar{F}}(n^{O(\eta)}2^{\mm}), \bar{d}_\infty) \\
&\leq \left(\frac{\|\tau\|_1}{M}\right)^{1/2} \left(\max_{x \in \Omega \cap B_F(2^{\mm})} \tilde{F}_{\rho,\nu}(x)\right)^{1/2} \cdot C(\eta) n^{O(\eta)}2^{\mm/2}\,,
\end{align*}
where the last inequality is from \pref{lem:BF-diam}, and $C(\eta)$ is a number depending on $\eta > 0$.

Thus for some choice $M \ls C(\eta)\eps^{-2}n^{1+O(\eta)} (\log m)^4$, we have
\[ \diam(\Omega \cap B_F(2^{\mm}), d_{\rho,\nu}) \ls \frac{\eps 2^{\mm/2}}{(\log m)^{3/2}} \left(\max_{x \in \Omega \cap B_F(2^{\mm})} \tilde{F}_{\rho,\nu}(x)\right)^{1/2}. \]

Similarly,
\begin{align*}
\gamma_2(\Omega \cap B_F(2^{\mm}), d_{\rho,\nu}) &\ls\  \left(\frac{\|\tau\|_1}{M}\right)^{1/2} \left(\max_{x \in \Omega \cap B_F(2^{\mm})} \tilde{F}_{\rho,\nu}(x)\right)^{1/2} \cdot \gamma_2(\Omega \cap B_F(2^{\mm}), \bar{d}_\infty) \\
&\leq\  \left(\frac{\|\tau\|_1}{M}\right)^{1/2} \left(\max_{x \in \Omega \cap B_F(2^{\mm})} \tilde{F}_{\rho,\nu}(x)\right)^{1/2} \cdot \gamma_2(B_{\bar{F}}(n^{O(\eta)}2^{\mm}), \bar{d}_\infty) \\
&\stackrel{\mathclap{\eqref{eq:dinfbound}}}{\leq}\ \left(\frac{\|\tau\|_1}{M}\right)^{1/2} \left(\max_{x \in \Omega \cap B_F(2^{\mm})} \tilde{F}_{\rho,\nu}(x)\right)^{1/2} \cdot C(\eta) n^{O(\eta)}2^{\mm/2} (\log m)^{3/2} \\
&\leq\ \eps 2^{\mm/2} \left(\max_{x \in \Omega \cap B_F(2^{\mm})} \tilde{F}_{\rho,\nu}(x)\right)^{1/2},
\end{align*}
again for
$M \ls C(\eta)\eps^{-2}n^{1+O(\eta)} (\log m)^4$ chosen sufficiently large. This implies $|F(x) - \tilde{F}(x)| \le \eps \cdot F(x)$ for all $x \in \Omega \cap B_F(2^{\mm})$. This holds for all $|\ell| \le 2 \log m$, so $|F(x) - \tilde{F}(x)| \le \eps \cdot F(x)$ for all $x \in \Omega$.

Thus applying \pref{lem:sparsification-meta} shows that $\tilde{F}$ satisfies \eqref{eq:tukey-approx}.
The claimed result now follows by choosing $\eta = o(1)$ as $n \to \infty$.

We remark that if the proofs of \pref{lem:BF-diam} and \eqref{eq:dinfbound} are unraveled, one can check that $C(\eta) \ls \exp(O(\eta^{-2}))$, so the optimal choice is $\eta \asymp (\log n)^{-1/3}$, yielding sparsifiers of size $\e^{-2} n^{1+O((\log n)^{-1/3})}(\log m)^4$.

\section{Algorithms for generalized linear models}
\label{sec:app}

\subsection{Optimizing generalized linear models}
\label{subsec:optglm}

Here we describe an algorithm that optimizes a GLM (recall \eqref{eq:glm}) under certain assumptions on its divergence. The algorithm reduces optimizing a GLM to a few queries to what we call a \emph{sparse GLM oracle} that optimizes convex GLMs induced by a limited number of functions plus a linear term. We start by defining convex GLMs (\pref{def:convex-glm}), a sparse GLM oracle (\pref{def:glmoracle}), the assumptions on the functions we consider (\pref{ass:refine}), and giving the main theorem on the algorithm for optimizing convex GLMs with a sparse GLM oracle (\pref{thm:optglm}).

\begin{definition}[Convex GLM]
\label{def:convex-glm}
We call a a family $\cF = \{f_1,\ldots,f_m\}$ of differentiable, convex functions, along with vectors $a_1,\ldots,a_m \in \R^n$ and $b \in \R^m$ a \emph{convex GLM}. We let $F(x) \defeq \sum_{i=1}^m f_i(\langle a_i, x \rangle - b_i)$ and let $F^* \seteq \min_{x \in \R^n} F(x)$.
\end{definition}

\begin{definition}[Sparse GLM oracle]
\label{def:glmoracle}
We call an algorithm a \emph{sparse GLM oracle} for a convex GLM (\pref{def:convex-glm}) if when given as input $\e > 0$, $w \in \Rplus^m$, and $y, x_{\mathrm{in}} \in \R^n$, it outputs a vector $x_{\mathrm{out}} \in \R^n$ such that
\begin{align*} &G(x_{\mathrm{out}}) - G(x^*) \le \eps(G(x_{\mathrm{in}}) - G(x^*))\,, \\ \text{ where } ~&G(x) := \langle y, x\rangle + \sum_{i=1}^m w_i f_i(\l a_i, x\r - b_i), \enspace \text{ and } \enspace x^* \seteq \argmin_{x \in \R^n} G(x) \,. \end{align*}
We use $\TGLM^{\cF}(s, \eps)$ to denote the worst-case runtime of such a GLM oracle on inputs where $w$ has at most $s$ nonzero entries, over all convex GLMs with a given family of functions $\cF = \{ f_1,\ldots, f_m \}$.
\end{definition}

\begin{assumption}[Approximate Divergence Sparsifiability]
   \label{ass:refine}
   Assume that, for some parameters, $\theta > 1$ and $\alpha,c,L > 0$, the function $f : \R \to \R$ is convex, differentiable,
   and satisfies the following property: For every $x_0 \in \R$, there is a convex function $r_{x_0} : \R \to \R$ that
   can be evaluated in $\O(\tEval)$ time, and such that
   \begin{itemize}
      \item $r_{x_0}(\Delta) \le D_{x_0}^{f}(x_0 + \Delta) \le \alpha \cdot r_{x_0}(\Delta)$ for all $\Delta \in \R$,
      \item $r_{x_0}^{1/2}$ is $L$-auto-Lipschitz (property \ref{item:lipschitz}).
      \item $r_{x_0}$ is lower $\lo$-homogeneous with constant $c$ (property \ref{item:lower-growth}).
   \end{itemize}
\end{assumption}

\begin{theorem}
\label{thm:optglm}
Consider $x^{(0)} \in \R^n$ and any convex GLM (\pref{def:convex-glm}) where each $f_i \in \cF$ satisfies \pref{ass:refine} with uniform constants $\theta > 1, \alpha,c,L > 0$.
Provided an upper bound $F(x^{(0)}) - F^* \leq \Gamma$, with high probability $\SolveGLM$ (\pref{alg:glm}) outputs a vector $y \in \R^n$ satisfying
$F(y) - F^* \leq \delta$ in time
\[ \O\left(\eta^{-1}(\nnz(A) + \TGLM^{\cF}(\O_{\alpha,c,\lo,L}(n), O(\eta/\alpha)) + m\tEval) \log\left(\Gamma/\delta\right)\right), \]
where $\eta := (10\alpha^2/c)^{-1/(\lo-1)}$, and $\tEval$ is the worst-case time needed to evaluate a function $f_i$.
\end{theorem}

Under mild assumptions on the family $\cF = \{f_1, \dots, f_m\}$, it holds that,
$\TGLM^{\cF}(\O(n), \e) \le \O(n^{\omega_0}\log(1/\e))$, where $\omega_0$ is as defined as in \pref{sec:lpintro} (see \cite[Theorem 4.2]{LSZ19-arxiv} for the formal conditions).

\smallskip

The algorithm establishing \pref{thm:optglm} is straightforward. In each
iteration, we write our optimization problem as a divergence with respect to
the current point $x^{(t)}$. Then we sparsify the divergence using
\pref{thm:main}. Finally, we call a GLM oracle (\pref{def:glmoracle}), which is
efficient, as we have sparsified the sum down to $\O(n)$ terms. We use this
to take a step, obtaining an improved point $x^{(t+1)}$, and repeat this iteration
$\O(1)$ times.  The main technical point is in handling the dependence of
\pref{thm:main} on the scale parameters $\smin$ and $\smax$.

Recall from \pref{def:F-divergence} that $T_x^f$ is the first-order Taylor approximation to $f$ at $x$,
and $D_x^f = f - T_x^f$ is the associated divergence.

\begin{algorithm2e}[!ht]
\label{alg:glmiter}
\caption{$\GLMIterate(\{f_1,\dots,f_m\}, \{a_1, \dots, a_m\}, b, x, \tilde{\Gamma})$}
\SetKwInOut{Input}{input}
\Input{Functions $f_i: \R \to \R$, vectors $a_i \in \R^n$, $b \in \R^m$, initial point $x \in \R^n$. Performs one step of iterative refinement to decrease function value.}
$y_i \assign \l a_i, x \r - b_i$ for $i \in [m]$. \\
Let $r_{i,y_i}$ denote the approximation to $D_{y_i}^{f_i}$, as in \pref{ass:refine}. \\
$\smin \gets m^{-O(1)}\tilde{\Gamma}, \smax \gets m^{O(1)}\tilde{\Gamma}$. \\
Use \pref{thm:main} to find weights $w \in \R^m_{\ge0}$ that induce an $\O_{\alpha,c,\lo,L}(n)$-sparse $(1/10)$-approximation of $r(\Delta) := \sum_{i\in[m]} r_{i,y_i}(\l a_i, \Delta\r)$ for all $\smin \le r(\Delta) \le \smax$. \\
Let $h(\Delta) := T^F_x(x + \Delta) + \frac{2}{3\alpha}\sum_{i \in [m]} w_i D_{y_i}^{f_i}(\l a_i, x+\Delta\r - b_i)$. \\
$\eta \gets (10\alpha^2/c)^{-1/(\lo-1)}$. \\
Find $\hat{\Delta}$ satisfying $h(\hat{\Delta}) \le h(\Delta^*) + \tilde{\Gamma}/10$ by calling a GLM oracle (\pref{def:glmoracle}) with $\cF = \{f_1, \dots, f_m\}$, $x_{\mathrm{in}} = 0$, and $\e = \eta/(30\alpha)$. \label{line:lsz} \\
\lIf*{$F(x + \eta \hat{\Delta}) \le F(x)$}{
	\textbf{return} $x + \eta \hat{\Delta}$ \lElse*{\textbf{return} $x$ \label{line:end}}}
\end{algorithm2e}

\begin{algorithm2e}[!ht]
\caption{$\SolveGLM(\{f_1,\dots,f_m\}, \{a_1, \dots, a_m\}, b, x^{(0)}, \Gamma, \delta)$}
\label{alg:glm}
\SetKwInOut{Input}{input}
\Input{Functions $f_i: \R \to \R$, vectors $a_i \in \R^n$, $b \in \R^m$, initial point $x^{(0)} \in \R^n$. Output a point that is a $\delta$-approximate minimizer of $\min_x \sum_{i\in[m]} f_i(\l a_i,x\r - b_i)$.}
$\tau \assign \lceil 2\eta^{-1} \log(\Gamma/\delta) \rceil$ for $\eta = (10\alpha^2/c)^{-1/(\lo-1)}$. \\
\For{$T = 0, 1, \dots, \tau-1$}{
	$x^{(T+1)} \assign \GLMIterate(\{f_1,\dots,f_m\}, \{a_1, \dots, a_m\}, b, x^{(T)}, (1-\eta/2)^T\Gamma)$.
}
\textbf{return} $x^{(\tau)}$.
\end{algorithm2e}

Let $E = F(x) - F(x^*)$ denote the function value error. The analysis of this algorithm proceeds by showing that one step of $\GLMIterate$ decreases $E$ as long as $C_0\smin \le E \le \smax/C_0$ for sufficiently large constant $C_0$ (and otherwise, does not increase the error). Intuitively, as long as this is true, we will show that $\Delta^* := x^* - x$ and $\hat{\Delta}$ are points where we indeed have sparsification guarantees by \pref{thm:main}. Towards this, we first establish that $r(\Delta^*) \asymp E$, where throughout this section
we allow the implicit constants in the $\asymp, \ls, \gtrsim$ notation to depend on the parameters $c, L, \lo, \alpha$ in \pref{thm:optglm}.
\begin{lemma}
\label{lem:basic}
In the setting of \pref{thm:optglm}, for $\Delta^* := x^* - x$ we have $E/\alpha \ls r(\Delta^*) \le 2E/\eta$.
\end{lemma}
\begin{proof}
Let $g = \nabla F(x)$. 
By optimality of $\Delta^*$, and because \ref{item:upper-growth} holds by \pref{lem:relate},
\[ -E \le \l g, 2\Delta^* \r + D^F_{x}(x + 2\Delta^*) \le 2\l g, \Delta^* \r + \alpha r(2\Delta^*) \stackrel{\ref{item:upper-growth}}{\le} 2\l g, \Delta^*\r + O(\alpha \cdot r(\Delta^*)). \] Because $\l g, \Delta^* \r \le -E$ by convexity, we deduce that $r(\Delta^*) \gtrsim -E/\alpha$.

For the upper bound, choose $\eta = (10\alpha/c)^{-1/(\lo-1)}$. By optimality of $\Delta^*$ we have
\begin{align*} -E &\le \l g, \eta\Delta^* \r + D^F_{x}(x + \eta\Delta^*) \le \eta \l g, \Delta^* \r + \alpha r(\eta \Delta^*) \\
&\le \eta \l g, \Delta^* \r + \eta^\lo \alpha/c \cdot r(\Delta^*) \le \eta \l g, \Delta^* \r + \eta/2 \cdot r(\Delta^*) \le -\eta/2 \cdot r(\Delta^*),
\end{align*}
where the final step uses $\l g, \Delta^*\r + r(\Delta^*) \le 0$. Thus, $r(\Delta^*) \le 2E/\eta$, as desired.
\end{proof}

Define $\wt{r}(\Delta) = \frac{2}{3\alpha}\sum_{i \in [m]} w_i D_{y_i}^{f_i}(\l a_i, x+\Delta\r - b_i)$. Note that $D_{y_i}^{f_i}(\l a_i, x+\Delta\r - b_i) \le \alpha r_{i,y_i}(\l a_i, \Delta\r)$ under the assumptions of \pref{thm:optglm}.
\begin{lemma}
\label{lem:basic2}
If $\tilde{\Gamma}/2 \le E \le \tilde{\Gamma}$, then $-E \lesssim \l \nabla F(x), \hat{\Delta} \r + \wt{r}(\hat{\Delta}) \le -0.8E$. Also, $r(\hat{\Delta}) \asymp E$.
\end{lemma}
\begin{proof}
Note that $C_0 \smin \le E \le \smax/C_0$ for a large constant $C_0$ because $\tilde{\Gamma}/2 \le E \le \tilde{\Gamma}$, and the choice of $\smin$ and $\smax$. \pref{thm:main} tells us that $\wt{r}(\Delta) \le r(\Delta) \le 2 \wt{r}(\Delta)$ whenever $\smin \le r(\Delta) \le \smax$, with high probability. Combining $C_0 \smin \le E \le \smax/C_0$ and \pref{lem:basic} gives that $\smin \le r(\Delta^*) \le \smax$. Thus, we know that $\wt{r}(\Delta^*) \le r(\Delta^*) \le 2\wt{r}(\Delta^*).$ By the definition of $\hat{\Delta}$ as in line \ref{line:lsz} of \pref{alg:glmiter},
\begin{align*}
h(\hat{\Delta}) &\le h(\Delta^*) + \tilde{\Gamma}/10 = T^F_{x}(x + \Delta^*) + \wt{r}(\Delta^*) + \tilde{\Gamma}/10 \\
&\le T^F_{x}(x + \Delta^*) + r(\Delta^*) + \tilde{\Gamma}/10 \le T^F_{x}(x + \Delta^*) + D^F_{x}(x + \Delta^*) + \tilde{\Gamma}/10 \\
&= F(x^*) - F(x) + \tilde{\Gamma}/10 \le -0.8E,
\end{align*}
as long as $E \ge \tilde{\Gamma}/2$.
This demonstrates the upper bound. Let us proceed to the lower bound. Let $g = \nabla F(x)$. We first consider the case where $r(\hat{\Delta}) < \smin$. Let $\beta$ be minimal so that $r(\beta\hat{\Delta}) \ge \smin$. By \ref{item:monotone}, we have that $r(\beta\hat{\Delta}) \ls \smin$ in fact. Then
\[ \l g, 2\beta \hat{\Delta} \r + \wt{r}(2\beta\hat{\Delta}) \le 2\beta \l g, \hat{\Delta} \r + r(2\beta\hat{\Delta}) \le 2\beta \l g, \hat{\Delta} \r + O(\smin) < 1.5\l g, \hat{\Delta} \r, \] contradicting the optimality of $\hat{\Delta}$ (as we have noted that $\l g, \hat{\Delta} \r \le -0.8E$). If $\smin \le r(\hat{\Delta}) \le \smax$, then for $\eta = (10\alpha^2/c)^{-1/(\lo-1)}$,
\begin{align*}
\l g, \hat{\Delta} \r + \wt{r}(\hat{\Delta}) &\ge \l g, \hat{\Delta} \r + \frac{1}{2\alpha}r(\hat{\Delta}) \ge \eta^{-1}(\l g, \eta \hat{\Delta}\r + \alpha r(\eta \hat{\Delta})) \\
&\ge \eta^{-1}(\l g, \eta \hat{\Delta}\r + D^F_x(x + \eta \hat{\Delta})) \ge -\eta^{-1}E.
\end{align*}
Finally, if $r(\hat{\Delta}) > \smax$, first pick $\eta$ to be maximal so that $r(\eta\hat{\Delta}) \le \smax$. It is easy to see that $r(\eta\hat{\Delta}) \gtrsim \smax$. Then $\wt{r}(\eta\hat{\Delta}) \gtrsim \smax$. Because $\wt{r}$ lower $\lo$-homogeneous with constant $c$, we deduce that \[ \eta |\l g, \hat{\Delta} \r| \ge \eta \wt{r}(\hat{\Delta}) \gtrsim \wt{r}(\eta\hat{\Delta}) \gtrsim \smax.\] For $\eta' = (C\alpha/c)^{-1/(\lo-1)}$ for sufficiently large $C$, we have that
\begin{align*}
-E \le \l g, \eta\eta'\hat{\Delta} \r + D^F_x(x + \eta\eta'\hat{\Delta}) \le \eta' \l g, \eta\hat{\Delta} \r + (\eta')^\lo\alpha/c \cdot r(\eta\hat{\Delta}) \lesssim -\smax,
\end{align*}
as $r(\eta\hat{\Delta}) \le \smax$ and $\eta \l g, \hat{\Delta} \r \lesssim -\smax$. This is a contradiction.

For the second point, we have shown $\smin \le r(\hat{\Delta}) \le \smax$, so $r(\hat{\Delta}) \asymp \wt{r}(\hat{\Delta}) \asymp E$, as desired.
\end{proof}

We now prove that taking one iteration makes sufficient progress.
\begin{lemma}
\label{lem:singleiter}
If $\tilde{\Gamma}/2 \le E \le \tilde{\Gamma}$, then for $\eta = (10\alpha^2/c)^{-1/(\lo-1)}$, we have $F(x + \eta\hat{\Delta}) - F(x^*) \le (1 - \eta/2)(F(x) - F(x^*))$.
\end{lemma}
\begin{proof}
Note that $r(\eta\hat{\Delta}) \asymp E$ by \pref{lem:basic2}, so we deduce that
\begin{align*}
F(x + \eta\hat{\Delta}) - F(x^*) &= \eta \l g, \hat{\Delta} \r + D^F_x(x + \eta\hat{\Delta}) \le \eta \l g, \hat{\Delta} \r + \alpha r(\eta\hat{\Delta}) \le \eta \l g, \hat{\Delta} \r + 2\alpha^2 \wt{r}(\eta\hat{\Delta})
\\&\le \eta(\l g, \hat{\Delta} \r + \wt{r}(\hat{\Delta})) \le -0.8\eta E,
\end{align*}
where we have applied \pref{lem:basic2} again.
\end{proof}
Now \pref{thm:optglm} follows easily.
\begin{proof}[Proof of \pref{thm:optglm}]
We first bound the runtime. We consider a single iteration of \GLMIterate. One part of the runtime is from sparsifying $r(\Delta)$, and is bounded by \pref{thm:main}. The other part is from calling an GLM oracle in line \ref{line:lsz}.
The sparsity of $w_1, \dots, w_m$ is at most $\O_{\alpha,c,\lo,L}(n)$ by \pref{thm:main}, and $\eps = \eta/(30\alpha)$, so the runtime is $\TGLM^{\cF}(\O_{\alpha,c,\lo,L}(n), \eta/(30\alpha))$.

Now we verify that $h(\Delta)$ takes the form of the GLM oracle as described in \pref{def:glmoracle}. Indeed,
\begin{align*}
h(\Delta) &= T^F_x(x + \Delta) + \frac{2}{3\alpha}\sum_{i \in [m]} w_i D_{y_i}^{f_i}(\l a_i, x+\Delta\r - b_i) \\
&= F(x) + \l \nabla F(x), \Delta \r + \frac{2}{3\alpha} \sum_{i \in [m]} w_i(f_i(\l a_i, x+\Delta\r - b_i) - f_i(\l a_i, x\r - b_i) - f_i'(\l a_i, x\r-b_i) \l a_i, \Delta\r).
\end{align*}
Now define $\bar{\Delta} \seteq x + \Delta$. Upon setting $\Delta = \bar{\Delta}-x$ in the above expression, all terms become either reweighted versions of $f_i(\langle a_i, \bar{\Delta} \rangle - b_i)$, constants, or linear terms in $\bar{\Delta}$. Thus, $h(\Delta)$ has the form of the GLM oracle. 

Now we check that setting $\eps = \eta/(30\alpha)$ and $x_{\mathrm{in}} = 0$ suffices to achieve $h(\hat{\Delta}) \le h(\Delta^*) + \tilde{\Gamma}/10$. It suffices to verify that $\eps(h(0) - h(\Delta^*)) \le \tilde{\Gamma}/10$. Indeed, a calculation yields
\begin{align*}h(0) - h(\Delta^*) &= F(x) - (T_x^F(x^*) + D_x^F(x^*)) + (D_x^F(x^*) - \tilde{r}(\Delta^*)) \\
&= F(x) - F(x^*) + (D_x^F(x^*) - \tilde{r}(\Delta^*)) \le E + D_x^F(x^*) \stackrel{(i)}{\le} E + \alpha r(\Delta^*) \stackrel{(ii)}{\le} E + 2\alpha E/\eta \le 3\alpha \tilde{\Gamma}/\eta, \end{align*}
where $(i)$ follows by the first bullet of \pref{ass:refine}, and $(ii)$ follows from \pref{lem:basic}.

Now we check correctness of the overall algorithm. We show by induction that $F(x^{(T)}) - F(x^*) \le \tilde{\Gamma} = (1-\eta/2)^T \Gamma$. This holds for $T = 0$ by the hypothesis in \pref{thm:optglm} that $F(x^{(0)}) - F^* \le \Gamma$. If $F(x^{(T)}) - F(x^*) \le (1-\eta/2)^{T+1} \Gamma$, then there is nothing to show because $F(x^{(T+1)}) \le F(x^{(T)})$, by line \ref{line:end} of \pref{alg:glmiter}. Otherwise, we conclude by \pref{lem:singleiter}, whose hypothesis that $\tilde{\Gamma}/2 \le E \le \tilde{\Gamma}$ holds for $\tilde{\Gamma} = (1-\eta/2)^T \Gamma$.
\end{proof}

\subsection{Applications}
\label{sec:lp}

\subsubsection{$\ell_p$-regression}

To establish \pref{thm:pnorm}, it suffices to check that the function $g(x) = |x|^p$ satisfies the conditions in \pref{thm:optglm}.

\begin{proof}[Proof of \pref{thm:pnorm}]
Consider the function $g(z) = |z|^p$ for $p > 1$.
From \cite[Lemma 4.5]{AKPS19-arxiv}, it holds that
\begin{equation}
\gamma_p(|x_0|, \Delta) \asymp_p D^g_{x_0}(x_0 + \Delta) \enspace \text{ for any } \enspace x_0,\Delta \in \R\,, \label{eq:akps}
\end{equation}
where $\gamma_p(t,z)$ is defined in \eqref{eq:gammap-gen}.
Therefore it suffices to check that the function $r_{x_0}(\Delta) \seteq \gamma_p(|x_0|,\Delta)$ satisfies
the required conditions, and this is the content of \pref{lem:gammap-nice}.
\end{proof}

\subsubsection{Dual of $\ell_p$-regression}
\label{sec:duallp}

We now establish \pref{thm:qnorm} by reducing the dual problem to a primal problem and applying \pref{thm:pnorm}.
Consider $p \in (1,2]$ and the dual exponent $q \seteq \frac{p}{p-1} \geq 2$, along with
$A \in \R^{m \times n}$. Then,
\begin{align*}
\min_{A^\top y = c} \|y\|_q &= \min_{A^\top y = c} \max_{\|z\|_p \le 1} z^\top y = \max_{\|z\|_p \le 1} \min_{A^\top y = c} z^\top y = \max_{\|Ax\|_p \le 1} c^\top x = \left(\min_{c^\top x = 1} \|Ax\|_p \right)^{-1},
\end{align*}
where the second equality uses convex duality (von Neumann's minimax theorem).

Denote $x^* \seteq \argmin_{c^\top x = 1} \|Ax\|_p$. By the KKT conditions, it holds that
\[ y^* := \argmin_{A^\top y = c} \|y\|_q = \frac{\mathrm{sign}(Ax^*) |Ax^*|^{p-1}}{\|Ax^*\|_p^p}\,, \] 
where we apply scalar functions to vectors in the straightforward way, e.g., $|A x^*|^{p-1} = \left(|(A x^*)_1|^{p-1},\ldots,|(A x^*)_m|^{p-1}\right)$.

Suppose that we
have a vector $\bar{x} \in \R^n$ satisfying $c^\top \bar{x} = 1$ and $\|A\bar{x}\|_p \le (1+\eps_0) \|Ax^*\|_p$, for some choice of $\e_0 > 0$
that we will make momentarily. We can find $\bar{x}$ using \pref{thm:pnorm} with the objective $\min_{x \in \R^n} K|\langle c, x\rangle-1|^p + \|Ax\|_p^p$, for sufficiently large constant $K$ depending on $\e_0, A$.
Define
\[ \bar{y} := \frac{\mathrm{sign}(A\bar{x}) |A\bar{x}|^{p-1}}{\|A\bar{x}\|_p^p}. \]

Our output $y$ will be the orthogonal projection of $\bar{y}$ onto the affine subspace $\{y : A^\top y = c\}$, which
can be found by solving a single linear system in $A^\top A$. Denote $\delta_y := y - \bar{y}$.
Clearly $\|y\|_q \le \|\bar{y}\|_q + \|\delta_y\|_q$.
Noting that
\begin{equation}
\label{lem:bary}
 \|\bar{y}\|_q = \|A\bar{x}\|_p^{-1} \le \|Ax^*\|_p^{-1} = \|y^*\|_q\,,
\end{equation}
our goal is to bound $\|\delta_y\|_q$.

The next lemma shows that choosing
$\eps_0 \asymp_p (\eps m^{1/q-1/2})^{2q/p}$ suffices to obtain $\|y\|^q_q \leq (1+\e) \|y^*\|^q_q$.

\begin{lemma}
\label{lemma:deltay}
It holds that $\|\delta_y\|_q \ls_p m^{1/2-1/q} \eps_0^{p/(2q)} \|y^*\|_q.$
\end{lemma}
\begin{proof}
   Write
\[
   \|\delta_y\|_q \le \|\delta_y\|_2 \leq \|y^* - \bar{y}\|_2 \le m^{1/2-1/q} \|y^* - \bar{y}\|_q\,,
\]
where the first inequality uses $q \geq 2$, the second uses the fact that $y$ is the orthogonal projection onto $\{ y : A^{\top} y = c\}$,
which contains $y^*$, and the third uses H\"{o}lder's inequality.

To bound the last expression, first write
\begin{align*}
\|y^* - \bar{y}\|_q &= \left\|\frac{\mathrm{sign}(Ax^*) |Ax^*|^{p-1}}{\|Ax^*\|_p^p} -  \frac{\mathrm{sign}(A\bar{x}) |A\bar{x}|^{p-1}}{\|A\bar{x}\|_p^p} \right\|_q \\
&\le \left|\frac{1}{{\|Ax^*\|_p^p}} - \frac{1}{\|A\bar{x}\|_p^p} \right| \cdot \left\|\mathrm{sign}(Ax^*) |Ax^*|^{p-1}\right\|_q + \left\| \frac{\mathrm{sign}(Ax^*) |Ax^*|^{p-1} - \mathrm{sign}(A\bar{x}) |A\bar{x}|^{p-1}}{\|A\bar{x}\|_p^p} \right\|_q.
\end{align*}
Because $\|A\bar{x}\|_p \le (1+\eps_0)\|Ax^*\|_p$, the first term is bounded by
\[ O(\eps_0) \cdot \|Ax^*\|_p^{-p} \cdot \|Ax^*\|_p^{p/q} \ls \eps_0 \|Ax^*\|_p^{-1} \ls  \eps_0 \|y^*\|_q\,. \]

For the second term, we use $2$-uniform convexity (with constant $p-1$) of the $\ell_p$ norm for $p \in (1,2]$
to obtain
\[\|A(\bar{x}-x^*)\|_p^2 \le \frac{2}{p-1}(\|A\bar{x}\|_p^2 - \|Ax^*\|_p^2) \ls_p \eps_0 \|Ax^*\|_p^2\,.\]

Because the function $f(z) := \mathrm{sign}(z) |z|^{p-1}$ satisfies $|f(y) - f(z)| \ls f(y-z)$ for $1 < p \leq 2$, we get
\begin{align*}
&\left\| \frac{\mathrm{sign}(Ax^*) |Ax^*|^{p-1} - \mathrm{sign}(A\bar{x}) |A\bar{x}|^{p-1}}{\|A\bar{x}\|_p^p} \right\|_q \ls_p  \|A\bar{x}\|_p^{-p} \||A(x^*-\bar{x})|^{p-1}\|_q \\ 
\leq ~&\|A\bar{x}\|_p^{-p} \|A(x^*-\bar{x})\|_p^{p/q} \ls_p \eps_0^{p/(2q)} \|Ax^*\|_p^{-p} \|Ax^*\|_p^{p/q} = \eps_0^{p/(2q)} \|y^*\|_q\,.\qedhere
\end{align*}
\end{proof}

\subsubsection{$\gamma_p$ regression}

Consider the case when the loss functions are of the form $f_1(z) = \cdots = f_m(z) = \gamma_p(z)$, and
we with to minimize $F(x) \seteq f_1(\langle a_1,x\rangle) + \cdots + f_m(\langle a_m,x\rangle)$.
For $p=1$ (Huber regression), one can compute a $(1+\e)$-approximate solution by first sparsifying
down to $\frac{n}{\e^2} (\log m)^{O(1)}$ terms using \pref{cor:huber}, and then solving the resulting problem
using a GLM oracle for the Huber loss on sparse instances.

For $p \in (1, 2]$, our framework allows us to find a minimizer to high
accuracy, in analogy with the case of $\ell_p$ regression.
This follows because the divergence of $\gamma_p$ around any point is,  up to constants, equal to $\gamma_p(t, z)$ for some threshold $t \geq 0$ (defined in \eqref{eq:gammap-gen}), and thus we can apply \pref{thm:optglm}.

\begin{lemma}
\label{lem:divgamma}
For all $p \in (1,2]$, the following holds:
If $|z| \le 1$ then $D^{\gamma_p}_z(\Delta+z) \asymp_p \gamma_p(\Delta)$, and if $|z| \ge 1$ then $D^{\gamma_p}_z(\Delta+z) \asymp \gamma_p(|z|, \Delta)$.
\end{lemma}
\begin{proof}
We will make use of \eqref{eq:akps} and the following fact: For any continuously differentiable function $h:\R \to \R$ that is
twice-differentiable at all but finitely many points,
\begin{equation}
\label{eq:divformula}
D^h_z(z+\Delta) = h(z+\Delta) - \left[h(z) + h'(z) \Delta\right] = \int_0^{\Delta} (\Delta-t)h''(z+t) \, dt \enspace \text{ for all } \enspace z, \Delta \in \R\,.
\end{equation}
We now prove the lemma by case analysis on $z,\Delta \in \R$. By symmetry, we may assume that $z \geq 0$.

\medskip
\noindent
{\bf Case (1): $z \in [0,1], |\Delta| \leq 4$.}
Here we have $z+\Delta \in [-4,5]$, and on this interval $\gamma_p''(z+\Delta) \asymp 1$.
From \eqref{eq:divformula}, this yields $D_z^{\gamma_p}(z+\Delta) \asymp_p \Delta^2 \asymp \gamma_p(\Delta)$.

\medskip
\noindent
{\bf Case (2): $z \in [0,1], |\Delta| > 4$.} Since $|z+\Delta| > 1$ and $p > 1$, it holds that
\[
   D^{\gamma_p}_z(z+\Delta) = |z+\Delta|^p - (1-p/2) - \frac{p}{2} z^2 - p z |\Delta| \asymp_p |\Delta|^p \asymp \gamma_p(\Delta)\,.
\]

\noindent
{\bf Case (3): $z > 1, |z+\Delta| \geq 1$.}
In this case, we have $D_z^{\gamma_p}(z+\Delta) = D_z^g(z+\Delta)$, where $g(z) \seteq |z|^p$.
From \eqref{eq:akps}, we have $D_z^g(z+\Delta) \asymp_p \gamma_p(|z|,\Delta)$.

\medskip
\noindent
{\bf Case (4): $z > 1, z+\Delta \in [1/2,1]$.} Because $\gamma_p''(y) \asymp_p g''(y)$ for all $y \in [1/2,\infty) \setminus \{1\}$,
\eqref{eq:divformula} gives $D_z^{\gamma_p}(z+\Delta) \asymp_p D_z^g(z+\Delta)$.

\medskip
\noindent
{\bf Case (5): $z > 1, z+\Delta \in [-1,1/2]$.}
Convexity of $\gamma_p$
implies that \[ D^{\gamma_p}_z(-1) \le D^{\gamma_p}_z(z+\Delta) \le
D^{\gamma_p}_z(1/2). \]
We have already argued that $D^{\gamma_p}_z(-1) \asymp_p \gamma_p(|z|, -1-z)$ and $D^{\gamma_p}_z(1/2) \asymp_p \gamma_p(|z|, 1/2-z)$. Because $-1-z \asymp 1/2-z$,
we conclude that $D^{\gamma_p}_z(-1) \asymp_p D^{\gamma_p}_z(1/2) \asymp_p \gamma_p(|z|,\Delta)$, for any $\Delta \in [-1-z, 1/2-z]$.
\end{proof}
 
\section*{Acknowledgments}
Part of this work was conducted while the authors were visiting the Simons Institute for the Theory of Computing.
James R. Lee is supported in part by NSF CCF-2007079 and a Simons Investigator Award.
Yang P. Liu is partially supported by the Google Research Fellowship and NSF DMS-1926686.
Aaron Sidford is supported in part by a Microsoft Research Faculty Fellowship, NSF CCF-1844855, NSF CCF-1955039, a PayPal research award, and a Sloan Research Fellowship.

\bibliographystyle{alpha}
\bibliography{glm-sparsify}

\newcommand{\etalchar}[1]{$^{#1}$}
\begin{thebibliography}{MMWY22}

\bibitem[ABKS21]{ABKS21}
Deeksha Adil, Brian Bullins, Rasmus Kyng, and Sushant Sachdeva.
\newblock Almost-linear-time weighted $\ell_p$-norm solvers in slightly dense
  graphs via sparsification.
\newblock In Nikhil Bansal, Emanuela Merelli, and James Worrell, editors, {\em
  48th International Colloquium on Automata, Languages, and Programming,
  {ICALP} 2021, July 12-16, 2021, Glasgow, Scotland (Virtual Conference)},
  volume 198 of {\em LIPIcs}, pages 9:1--9:15. Schloss Dagstuhl -
  Leibniz-Zentrum f{\"{u}}r Informatik, 2021.

\bibitem[AKPS19a]{AKPS19-arxiv}
Deeksha Adil, Rasmus Kyng, Richard Peng, and Sushant Sachdeva.
\newblock Iterative refinement for $\ell_p$-norm regression.
\newblock 2019.
\newblock Available at
  \href{https://arxiv.org/abs/1901.06765}{arXiv:1901.06765}.

\bibitem[AKPS19b]{AKPS19}
Deeksha Adil, Rasmus Kyng, Richard Peng, and Sushant Sachdeva.
\newblock Iterative refinement for {\(\mathscr{l}\)}p-norm regression.
\newblock In Timothy~M. Chan, editor, {\em Proceedings of the Thirtieth Annual
  {ACM-SIAM} Symposium on Discrete Algorithms, {SODA} 2019, San Diego,
  California, USA, January 6-9, 2019}, pages 1405--1424. {SIAM}, 2019.

\bibitem[AS20]{AS20}
Deeksha Adil and Sushant Sachdeva.
\newblock Faster \emph{p}-norm minimizing flows, via smoothed \emph{q}-norm
  problems.
\newblock In Shuchi Chawla, editor, {\em Proceedings of the 2020 {ACM-SIAM}
  Symposium on Discrete Algorithms, {SODA} 2020, Salt Lake City, UT, USA,
  January 5-8, 2020}, pages 892--910. {SIAM}, 2020.

\bibitem[BCLL18]{BCLL18}
S{\'{e}}bastien Bubeck, Michael~B. Cohen, Yin~Tat Lee, and Yuanzhi Li.
\newblock An homotopy method for $\ell_p$ regression provably beyond
  self-concordance and in input-sparsity time.
\newblock In Ilias Diakonikolas, David Kempe, and Monika Henzinger, editors,
  {\em Proceedings of the 50th Annual {ACM} {SIGACT} Symposium on Theory of
  Computing, {STOC} 2018, Los Angeles, CA, USA, June 25-29, 2018}, pages
  1130--1137. {ACM}, 2018.

\bibitem[BLL{\etalchar{+}}21]{BLLSSSW21}
Jan van~den Brand, Yin~Tat Lee, Yang~P. Liu, Thatchaphol Saranurak, Aaron
  Sidford, Zhao Song, and Di~Wang.
\newblock Minimum cost flows, mdps, and
  {\(\mathscr{l}\)}\({}_{\mbox{1}}\)-regression in nearly linear time for dense
  instances.
\newblock In Samir Khuller and Virginia~Vassilevska Williams, editors, {\em
  {STOC} '21: 53rd Annual {ACM} {SIGACT} Symposium on Theory of Computing,
  Virtual Event, Italy, June 21-25, 2021}, pages 859--869. {ACM}, 2021.

\bibitem[BLM89]{BLM89}
J.~Bourgain, J.~Lindenstrauss, and V.~Milman.
\newblock Approximation of zonoids by zonotopes.
\newblock {\em Acta Math.}, 162(1-2):73--141, 1989.

\bibitem[BLN{\etalchar{+}}20]{BLNPSSSW20}
Jan van~den Brand, Yin~Tat Lee, Danupon Nanongkai, Richard Peng, Thatchaphol
  Saranurak, Aaron Sidford, Zhao Song, and Di~Wang.
\newblock Bipartite matching in nearly-linear time on moderately dense graphs.
\newblock In Sandy Irani, editor, {\em 61st {IEEE} Annual Symposium on
  Foundations of Computer Science, {FOCS} 2020, Durham, NC, USA, November
  16-19, 2020}, pages 919--930. {IEEE}, 2020.

\bibitem[BLSS20]{BLSS20}
Jan van~den Brand, Yin~Tat Lee, Aaron Sidford, and Zhao Song.
\newblock Solving tall dense linear programs in nearly linear time.
\newblock In Konstantin Makarychev, Yury Makarychev, Madhur Tulsiani, Gautam
  Kamath, and Julia Chuzhoy, editors, {\em Proccedings of the 52nd Annual {ACM}
  {SIGACT} Symposium on Theory of Computing, {STOC} 2020, Chicago, IL, USA,
  June 22-26, 2020}, pages 775--788. {ACM}, 2020.

\bibitem[CLM{\etalchar{+}}15]{CLMMPS15}
Michael~B. Cohen, Yin~Tat Lee, Cameron Musco, Christopher Musco, Richard Peng,
  and Aaron Sidford.
\newblock Uniform sampling for matrix approximation.
\newblock In {\em {ITCS}}, pages 181--190. {ACM}, 2015.

\bibitem[CLS21]{CLS19}
Michael~B. Cohen, Yin~Tat Lee, and Zhao Song.
\newblock Solving linear programs in the current matrix multiplication time.
\newblock {\em J. {ACM}}, 68(1):3:1--3:39, 2021.

\bibitem[CP15]{CP15}
Michael~B. Cohen and Richard Peng.
\newblock $l_p$ row sampling by lewis weights.
\newblock In Rocco~A. Servedio and Ronitt Rubinfeld, editors, {\em Proceedings
  of the Forty-Seventh Annual {ACM} on Symposium on Theory of Computing, {STOC}
  2015, Portland, OR, USA, June 14-17, 2015}, pages 183--192. {ACM}, 2015.

\bibitem[CW15]{CW15}
Kenneth~L. Clarkson and David~P. Woodruff.
\newblock Input sparsity and hardness for robust subspace approximation.
\newblock In Venkatesan Guruswami, editor, {\em {IEEE} 56th Annual Symposium on
  Foundations of Computer Science, {FOCS} 2015, Berkeley, CA, USA, 17-20
  October, 2015}, pages 310--329. {IEEE} Computer Society, 2015.

\bibitem[CWW19]{CWW19}
Kenneth~L. Clarkson, Ruosong Wang, and David~P. Woodruff.
\newblock Dimensionality reduction for tukey regression.
\newblock In Kamalika Chaudhuri and Ruslan Salakhutdinov, editors, {\em
  Proceedings of the 36th International Conference on Machine Learning, {ICML}
  2019, 9-15 June 2019, Long Beach, California, {USA}}, volume~97 of {\em
  Proceedings of Machine Learning Research}, pages 1262--1271. {PMLR}, 2019.

\bibitem[GPV21]{GPV21}
Mehrdad Ghadiri, Richard Peng, and Santosh~S. Vempala.
\newblock Faster $p$-norm regression using sparsity.
\newblock abs/2109.11537, 2021.

\bibitem[JLLS23]{JLLS23}
Arun Jambulapati, James~R. Lee, Yang~P. Liu, and Aaron Sidford.
\newblock Sparsifying sums of norms.
\newblock {\em arXiv preprint arXiv:2305.09049}, 2023.
\newblock To appear in FOCS 2023. Available at
  \href{https://arxiv.org/pdf/2305.09049.pdf}{arXiv:2305.09049}.

\bibitem[JLS22]{JLS22}
Arun Jambulapati, Yang~P. Liu, and Aaron Sidford.
\newblock Improved iteration complexities for overconstrained \emph{p}-norm
  regression.
\newblock In Stefano Leonardi and Anupam Gupta, editors, {\em {STOC} '22: 54th
  Annual {ACM} {SIGACT} Symposium on Theory of Computing, Rome, Italy, June 20
  - 24, 2022}, pages 529--542. {ACM}, 2022.

\bibitem[JSWZ21]{JSWZ21}
Shunhua Jiang, Zhao Song, Omri Weinstein, and Hengjie Zhang.
\newblock A faster algorithm for solving general lps.
\newblock In Samir Khuller and Virginia~Vassilevska Williams, editors, {\em
  {STOC} '21: 53rd Annual {ACM} {SIGACT} Symposium on Theory of Computing,
  Virtual Event, Italy, June 21-25, 2021}, pages 823--832. {ACM}, 2021.

\bibitem[Lew79]{Lewis79}
D.~R. Lewis.
\newblock Ellipsoids defined by {B}anach ideal norms.
\newblock {\em Mathematika}, 26(1):18--29, 1979.

\bibitem[LMP13]{LMP13}
Mu~Li, Gary~L Miller, and Richard Peng.
\newblock Iterative row sampling.
\newblock In {\em 2013 IEEE 54th Annual Symposium on Foundations of Computer
  Science}, pages 127--136. IEEE, 2013.

\bibitem[LSZ19a]{LSZ19}
Yin~Tat Lee, Zhao Song, and Qiuyi Zhang.
\newblock Solving empirical risk minimization in the current matrix
  multiplication time.
\newblock In Alina Beygelzimer and Daniel Hsu, editors, {\em Conference on
  Learning Theory, {COLT} 2019, 25-28 June 2019, Phoenix, AZ, {USA}}, volume~99
  of {\em Proceedings of Machine Learning Research}, pages 2140--2157. {PMLR},
  2019.

\bibitem[LSZ19b]{LSZ19-arxiv}
Yin~Tat Lee, Zhao Song, and Qiuyi Zhang.
\newblock Solving empirical risk minimization in the current matrix
  multiplication time.
\newblock 2019.
\newblock Available at
  \href{https://arxiv.org/abs/1905.04447}{arXiv:1905.04447}.

\bibitem[LT11]{LedouxTalagrand2011}
Michel Ledoux and Michel Talagrand.
\newblock {\em Probability in {B}anach spaces}.
\newblock Classics in Mathematics. Springer-Verlag, Berlin, 2011.
\newblock Isoperimetry and processes, Reprint of the 1991 edition.

\bibitem[MMWY22]{MMWY22}
Cameron Musco, Christopher Musco, David~P. Woodruff, and Taisuke Yasuda.
\newblock Active linear regression for $\ell_p$ norms and beyond.
\newblock In {\em 63rd {IEEE} Annual Symposium on Foundations of Computer
  Science, {FOCS} 2022, Denver, CO, USA, October 31 - November 3, 2022}, pages
  744--753. {IEEE}, 2022.

\bibitem[PT85]{PT85}
Alain Pajor and Nicole {Tomczak-Jaegermann}.
\newblock Remarques sur les nombres d'entropie d'un op\'{e}rateur et de son
  transpos\'{e}.
\newblock {\em C. R. Acad. Sci. Paris S\'{e}r. I Math.}, 301(15):743--746,
  1985.

\bibitem[SS11]{SS11}
Daniel~A. Spielman and Nikhil Srivastava.
\newblock Graph sparsification by effective resistances.
\newblock {\em {SIAM} J. Comput.}, 40(6):1913--1926, 2011.

\bibitem[ST14]{ST14}
Daniel~A. Spielman and Shang{-}Hua Teng.
\newblock Nearly linear time algorithms for preconditioning and solving
  symmetric, diagonally dominant linear systems.
\newblock {\em {SIAM} J. Matrix Anal. Appl.}, 35(3):835--885, 2014.

\bibitem[SZ01]{SZ01}
Gideon Schechtman and Artem Zvavitch.
\newblock Embedding subspaces of {$L_p$} into {$l^N_p$}, {$0<p<1$}.
\newblock {\em Math. Nachr.}, 227:133--142, 2001.

\bibitem[Tal90]{Tal90}
Michel Talagrand.
\newblock Embedding subspaces of $l_1$ into $\ell_1^n$.
\newblock {\em Proceedings of the American Mathematical Society},
  108(2):363--369, 1990.

\bibitem[Tal95]{Tal95}
M.~Talagrand.
\newblock Embedding subspaces of {$L_p$} in {$l^N_p$}.
\newblock In {\em Geometric aspects of functional analysis ({I}srael,
  1992--1994)}, volume~77 of {\em Oper. Theory Adv. Appl.}, pages 311--325.
  Birkh\"{a}user, Basel, 1995.

\bibitem[Tal14]{TalagrandBook2014}
Michel Talagrand.
\newblock {\em Upper and lower bounds for stochastic processes}, volume~60 of
  {\em Ergebnisse der Mathematik und ihrer Grenzgebiete. 3. Folge. A Series of
  Modern Surveys in Mathematics [Results in Mathematics and Related Areas. 3rd
  Series. A Series of Modern Surveys in Mathematics]}.
\newblock Springer, Heidelberg, 2014.
\newblock Modern methods and classical problems.

\bibitem[WY23]{WY23}
David~P. Woodruff and Taisuke Yasuda.
\newblock Sharper bounds for $\ell_p$ sensitivity sampling.
\newblock In Andreas Krause, Emma Brunskill, Kyunghyun Cho, Barbara Engelhardt,
  Sivan Sabato, and Jonathan Scarlett, editors, {\em International Conference
  on Machine Learning, {ICML} 2023, 23-29 July 2023, Honolulu, Hawaii, {USA}},
  volume 202 of {\em Proceedings of Machine Learning Research}, pages
  37238--37272. {PMLR}, 2023.

\end{thebibliography}

\end{document}